%% file: main.tex
\documentclass{jfm}
\usepackage{graphicx}
\usepackage{amsmath}
\usepackage{amssymb}
\usepackage{mathtools}
\usepackage{esint}
\usepackage{mathrsfs}
\usepackage{xcolor}
\usepackage{tikz}
\usetikzlibrary{decorations.pathmorphing}

\usepackage{hyperref}
\hypersetup{
    colorlinks = true,
    urlcolor   = blue,
    citecolor  = blue,
    linkcolor  = blue,
}
\newtheorem{lemma}{Lemma}
\newtheorem{proposition}{Proposition}
\newtheorem{definition}{Definition}

\newcommand{\red}[1]{{\color{red}#1}}

\newcommand*\Eval[3]{\left.#1\right\rvert_{#2}^{#3}}
\newcommand{\uvec}[1]{\boldsymbol{\hat{\mathbf{#1}}}}
\newcommand\Ra{\mbox{\textit{R}}}  % Reynolds number
\newcommand{\abs}[1]{\left\vert #1 \right\vert}
\newcommand{\dVol}{{\rm d}\boldsymbol{x}}
\newcommand{\quinopt}{{\sc quinopt}}

\newcommand{\sdpagmp}{{\sc sdpa-gmp}}
\newcommand{\sparsecolo}{{\sc SparseCoLO}}

\definecolor{matlabblue}{RGB}{0,113,188}
\definecolor{matlabred}{RGB}{216,82,24}
\definecolor{matlabgray}{RGB}{128,128,128}
\definecolor{matlabgreen}{rgb}{0.4660    0.6740    0.1880} % previously 0, 0.7 0
\definecolor{matlabyellow}{rgb}{0.93      0.69      0.13} % previously 0, 0.7 0
\definecolor{matlabpurple}{rgb}{0.49      0.18      0.56}
\definecolor{colorbar1}{rgb}{1.000000,0.909091,0.000000}
\definecolor{colorbar2}{rgb}{1.000000,0.818182,0.000000}
\definecolor{colorbar3}{rgb}{1.000000,0.727273,0.000000}
\definecolor{colorbar4}{rgb}{1.000000,0.636364,0.000000}
\definecolor{colorbar5}{rgb}{1.000000,0.545455,0.000000}
\definecolor{colorbar6}{rgb}{1.000000,0.454545,0.000000}
\definecolor{colorbar7}{rgb}{1.000000,0.363636,0.000000}
\definecolor{colorbar8}{rgb}{1.000000,0.272727,0.000000}
\definecolor{colorbar9}{rgb}{1.000000,0.181818,0.000000}
\definecolor{colorbar10}{rgb}{1.000000,0.090909,0.000000}
\definecolor{colorbar11}{rgb}{1.000000,0.000000,0.000000}
\definecolor{colorbar12}{rgb}{0.909091,0.000000,0.000000}
\definecolor{colorbar13}{rgb}{0.818182,0.000000,0.000000}
\definecolor{colorbar14}{rgb}{0.727273,0.000000,0.000000}
\definecolor{colorbar15}{rgb}{0.636364,0.000000,0.000000}
\definecolor{colorbar16}{rgb}{0.545455,0.000000,0.000000}
\definecolor{colorbar17}{rgb}{0.454545,0.000000,0.000000}
\definecolor{colorbar18}{rgb}{0.363636,0.000000,0.000000}
\definecolor{colorbar19}{rgb}{0.272727,0.000000,0.000000}
\definecolor{colorbar20}{rgb}{0.181818,0.000000,0.000000}
\definecolor{colorbar21}{rgb}{0.090909,0.000000,0.000000}
\definecolor{grey}{rgb}{0.6,0.6,0.6}

\newcommand\solidrule[1][10pt]{\rule[0.5ex]{#1}{1.5pt}}
\newcommand\dashedrule{\mbox{%
		\solidrule[2pt]\hspace{2pt}\solidrule[2pt]\hspace{2pt}\solidrule[2pt]}}
\newcommand\dotdashedrule{\mbox{%
		\solidrule[3pt]\hspace{1.5pt}\solidrule[1pt]\hspace{1.5pt}\solidrule[3pt]}}
\newcommand\dottedrule{\mbox{%
		\solidrule[1pt]\hspace{1pt}\solidrule[1pt]\hspace{1pt}\solidrule[1pt]\hspace{1pt}\solidrule[1pt]\hspace{1pt}\solidrule[1pt]\hspace{1pt}}}

\newcommand{\mysquare}[1]{%
	\protect\begin{tikzpicture}%
	\protect\draw[thick,color=#1] (0,0) -- (0.75ex,0) -- (0.75ex,0.75ex) -- (0,0.75ex) -- (0,0);
	\protect\end{tikzpicture}%
}
\newcommand{\mytriangle}[1]{%
	\protect\begin{tikzpicture}%
	\protect\draw[thick,color=#1] (0,0) -- (1ex,0) -- (0.5ex,0.866ex) -- (0,0);
	\protect\end{tikzpicture}%
}

\newcommand{\mycross}[1]{%
	\protect\begin{tikzpicture}%
	\protect\draw[thick,color=#1] (0,0) -- (1ex,1ex);
	\protect\draw[thick,color=#1] (0,1ex) -- (1ex,0);
	\protect\end{tikzpicture}%
}

\title{Bounds on heat transport for convection driven by internal heating}

\author{Ali Arslan\aff{1}
  \corresp{\email{a.arslan18@imperial.ac.uk}},
  Giovanni Fantuzzi\aff{1},
  John Craske\aff{2}
 \and Andrew Wynn\aff{1}}

\affiliation{\aff{1}Department of Aeronautics, Imperial College London
\aff{2}Department of Civil and Envirnonmental Engineering, Imperial College London}

\begin{document}
\maketitle

\begin{abstract}
The mean vertical heat transport $\langle wT \rangle$ in convection between isothermal plates driven by uniform internal heating is investigated by means of rigorous bounds. These are obtained as a function of the Rayleigh number \Ra\  by constructing feasible solutions to a convex variational problem, derived using a formulation of the classical background method in terms of a quadratic auxiliary function. When the fluid's temperature relative to the boundaries is allowed to be positive or negative, numerical solution of the variational problem shows that best previous bound $\langle wT \rangle \leq 1/2$ \citep[][\emph{Physics Letters A}, {vol. 377},  pp.83-92]{goluskin2012convection} can only be improved up to finite \Ra. Indeed, we demonstrate analytically that $ \langle wT \rangle \leq 2^{-21/5} \Ra^{1/5}$ and therefore prove that $\langle wT\rangle< 1/2$ for $\Ra < 65\,536$. However, if the minimum principle for temperature is invoked, which asserts that internal temperature is at least as large as the temperature of the isothermal boundaries, then numerically optimised bounds are strictly smaller than $1/2$  until at least $\Ra=3.4\times 10^{5}$. While the computational results suggest that the best bound on $\langle wT\rangle$ approaches $1/2$ asymptotically from below as $\Ra\rightarrow \infty$, we prove that typical analytical constructions cannot be used to prove this conjecture.
\end{abstract}

\input{revised-paper/revised-v2}

\input{main.bbl}
%\addcontentsline{toc}{chapter}{References}
%\bibliographystyle{jfm}
%\bibliography{bibs/jfm.bib}

\end{document}

%% file: revised-paper/revised-v2.tex
% (So that I can see the structure)
% \setcounter{tocdepth}{3}
% \tableofcontent
\section{Introduction}
\label{sec:intro}

Internally-heated (IH) convection, in which the motion of a fluid is driven by buoyancy forces caused by internal sources of heat, is found in a wide variety of natural and built environments, and plays an essential role in disciplines such as geophysics and astrophysics. For example, radioactive decay drives convection in the Earth's mantle, which in turn influences plate tectonics and the planet's magnetic field \citep{bercovici2011mantle}. A similar mechanism explains geological patterns on the surface of Pluto \citep{trowbridge2016vigorous}, and buoyancy flows due to the absorption of solar radiation induce atmospheric turbulence on Venus \citep{tritton1975internally}.

Internal heating generalises thermal forcing at the boundaries typical of Rayleigh--B\'enard models of convection in both a theoretical and a practical sense, because internal heat sources can be concentrated near boundaries to produce the latter \citep{BouVjfm2019a}. The dynamics of IH convection are also closely related, and sometimes equivalent, to those of flows driven by internal sources of buoyancy besides temperature, such density stratification due to electromagnetic forces or chemical concentration differences \citep{goluskin2016internally}.
IH convection therefore warrants study in its own right to enhance fundamental understanding of buoyancy-driven turbulence, yet has received relatively little attention in comparison with  Rayleigh--B\'enard convection. 

A fundamental challenge in the study of IH convection, along with many other turbulent flows, is to characterise the flow's statistical properties as a function of its control parameters. Following previous work \citep{goluskin2012convection,goluskin2016internally}, we consider this problem in the idealized configuration illustrated in Figure~\ref{fig:config}, where a horizontal layer of fluid between isothermal plates of equal temperature is heated uniformly at a constant rate. The only control parameters for this setting are the Prandtl number of the fluid, \Pran, and the Rayleigh number based on the internal heating rate, \Ra. Particular statistical quantities of interest are the dimensionless mean temperature, $\langle T \rangle$, and the dimensionless mean vertical convective heat flux, $\langle wT \rangle$, where the mean is obtained by averaging over volume and infinite time. 

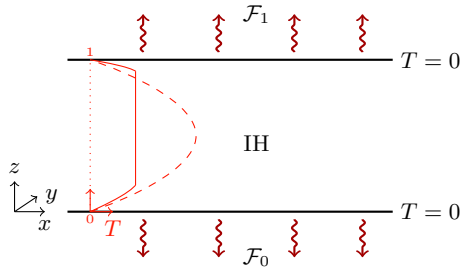
\begin{figure}
    \centering
    \begin{tikzpicture}[every node/.style={scale=0.95}]
    \draw[black,thick] (-2,0) -- (2.3,0) ;%;node [anchor=west] {$z=0$} ;
    \draw[black,thick] (-2,2) -- (2.3,2) ;  
    % \draw[grey] (-1.8,-0.3) -- (-1.3,-0.05) ; 
    % \draw[grey] (-1.3,-0.3) -- (-0.8,-0.05) ;
    % \draw[grey] (-0.8,-0.3) -- (-0.3,-0.05) ; 
    % \draw[grey] (-0.3,-0.3) -- (0.2,-0.05) ; 
    % \draw[grey] (0.2,-0.3) -- (0.7,-0.05) ;
    % \draw[grey] (0.7,-0.3) -- (1.2,-0.05) ;
    % \draw[grey] (1.2,-0.3) -- (1.7,-0.05) ;
    % \draw[grey] (-1.8,2.05) -- (-1.3,2.3) ; 
    % \draw[grey] (-1.3,2.05) -- (-0.8,2.3) ;
    % \draw[grey] (-0.8,2.05) -- (-0.3,2.3) ; 
    % \draw[grey] (-0.3,2.05) -- (0.2,2.3) ; 
    % \draw[grey] (0.2,2.05) -- (0.7,2.3) ;
    % \draw[grey] (0.7,2.05) -- (1.2,2.3) ;
    % \draw[grey] (1.2,2.05) -- (1.7,2.3) ;
    \draw[->] (-2.7,0) -- (-2.3,0) node [anchor=north] {$x$};
    \draw[->] (-2.7,0) -- (-2.4,0.2) node [anchor=west] {$y$};
    \draw[->] (-2.7,0) -- (-2.7,0.4) node [anchor=south] {$z$};
    % \draw[<->,black,dashed] (-1.8,0.15) -- (-1.8,1.85) node at (-2,1){$1$}; 
    \draw[dotted,colorbar10] (-1.7,0) -- (-1.7,2);
    \draw [dashed,colorbar10] plot [smooth, tension =1] coordinates {(-1.7,0) (-0.3,1) (-1.7,2)};
    \draw[->,colorbar10] (-1.7,0) -- (-1.4,0) node [anchor=north] {$T$} ;
    \draw[->,colorbar10] (-1.7,0) -- (-1.7,0.3) ;
    \node[colorbar10] at (-1.7,-0.1) {\tiny $0$};
    \node[colorbar10] at (-1.7,2.1) {\tiny $1$};
    \draw [colorbar10] plot [smooth,tension=1] coordinates {(-1.7,2) (-1.3,1.92) (-1.1,1.85)};
    \draw [colorbar10] (-1.1,1.85) -- (-1.1,0.35);
    \draw [colorbar10] plot [smooth, tension =1] coordinates {(-1.1,0.35) (-1.3,0.2) (-1.7,0)};
    \node at (0.5,-0.6) {$ \mathcal{F}_0 $};
    \node at (0.5,2.6) {$  \mathcal{F}_1 $};
    \node at (2.8,0) {$ T = 0 $};
    \node at (2.8,2) {$  T = 0 $};
    \node at (0.5,0.9) {$ \textrm{IH}$};
    % \draw [->,colorbar15,line join=round,decorate, decoration={snake, segment length=4, amplitude=.9,post=lineto, post length=2pt}, thick]  (-1.9,-0.1) -- (-1.9,-0.6);
    \draw [->,colorbar15,line join=round,decorate, decoration={snake, segment length=4, amplitude=.9,post=lineto, post length=2pt}, thick]  (-1,-0.1) -- (-1,-0.6);
    \draw [->,colorbar15,line join=round,decorate, decoration={snake, segment length=4, amplitude=.9,post=lineto, post length=2pt}, thick]  (0,-0.1) -- (0,-0.6);
    \draw [->,colorbar15,line join=round,decorate, decoration={snake, segment length=4, amplitude=.9,post=lineto, post length=2pt}, thick]  (1,-0.1) -- (1,-0.6);
    \draw [->,colorbar15,line join=round,decorate, decoration={snake, segment length=4, amplitude=.9,post=lineto, post length=2pt}, thick]  (1.9,-0.1) -- (1.9,-0.6);
    %  \draw [->,colorbar15,line join=round,decorate, decoration={snake, segment length=4, amplitude=.9,post=lineto, post length=2pt}, thick]  (-1.9,1.9) -- (-1.9,2.4);
    \draw [->,colorbar15,line join=round,decorate, decoration={snake, segment length=4, amplitude=.9,post=lineto, post length=2pt}, thick]  (-1,2.1) -- (-1,2.6);
    \draw [->,colorbar15,line join=round,decorate, decoration={snake, segment length=4, amplitude=.9,post=lineto, post length=2pt}, thick]  (0,2.1) -- (0,2.6);
    \draw [->,colorbar15,line join=round,decorate, decoration={snake, segment length=4, amplitude=.9,post=lineto, post length=2pt}, thick]  (1,2.1) -- (1,2.6);
     \draw [->,colorbar15,line join=round,decorate, decoration={snake, segment length=4, amplitude=.9,post=lineto, post length=2pt}, thick]  (1.9,2.1) -- (1.9,2.6);
    \end{tikzpicture}
    \caption{ Schematic diagram for convection driven by unit uniform internal heat generation, IH, between two isothermal parallel plates. The averaged heat fluxed through the top and bottom plates are denoted by $\mathcal{F}_1$ and $\mathcal{F}_0$, respectively. Red lines illustrate the temperature profile in the conductive regime ({\color{colorbar10}\dashedrule}) and a typical mean temperature profile in the turbulent regime ({\color{colorbar10}\solidrule}).}
    \label{fig:config}
\end{figure}

The dimensionless mean temperature $\langle T\rangle$ corresponds to the amount of thermal dissipation in the fluid and can be related qualitatively to the proportion of heat within the fluid that is transported by conduction, rather than by convection. As described by \citet[][p.87]{goluskin2012convection}, the average outward conduction above and below the plane over which the time- and plane-averaged temperature $\overline{T}$ is maximised is equal to $2\overline{T}$. 
% Assuming a homogenised temperature field, for sufficiently high \Ra\, $\langle T\rangle$ scales in the same way as the maximum $\overline{T}$ and the ratio of the total (predominantly convective) heat flux to the conductive heat flux, which corresponds to a Nusselt number, is $1/\langle T\rangle$. 
If one assumes that at high \Ra\ the temperature field is well mixed, then $\langle T\rangle$ scales in the same way as the maximum of $\overline{T}$. The ratio of the total (predominantly convective) heat flux to the conductive heat flux, which corresponds to a Nusselt number, is therefore $1/(2\overline{T}) \sim 1/\langle T\rangle$.
In contrast, $\langle wT\rangle$ quantifies the vertical asymmetry caused by the fluid's motion and is related to the heat fluxes $\mathcal{F}_1$ and $\mathcal{F}_0$ through the top and bottom boundaries by the exact relations~\citep{goluskin2016internally} 
%
%Along with other turbulent flows, a fundamental challenge in the study of IH convection is to characterise the flow's statistical properties as a function of its control parameters. In the case of IH convection the control parameters are the Rayleigh number \Ra\, based on the internal heating rate, and the Prandtl number \Pran. Following previous work \citep{goluskin2012convection,goluskin2016internally}, the basic configuration we consider is a horizontal layer of fluid between isothermal plates of equal temperature heated uniformly at a constant rate. The particular statistical quantities of interest are the dimensionless mean temperature, $\langle T \rangle$, and the dimensionless mean vertical convective heat flux, $\langle wT \rangle$, where the mean is obtained by averaging over volume and infinite time. The former equals the amount of thermal dissipation in the fluid and, being proportional to the ratio between the fixed total outward heat flux and the conductive outward heat flux, behaves like an inverse Nusselt number \citep{goluskin2012convection,goluskin2016internally}. The latter quantifies the vertical asymmetry caused by the fluid's motion and is related to the heat fluxes $\mathcal{F}_1$ and $\mathcal{F}_0$ through the top and bottom boundaries by the exact relations~\citep{goluskin2016internally} 
%
\begin{equation}\label{e:boundary-fluxes-vs-wT}
\mathcal{F}_1 = \frac12 + \langle wT \rangle, 
\qquad
\mathcal{F}_0 = \frac12 - \langle wT \rangle.
\end{equation}

Laboratory experiments \citep{kulacki1972thermal,Jahn1974,Mayinger1976,kakac1985natural,Lee2007} and direct numerical simulations \citep{Peckover1974,Straus1976,Tveitereid1978,emara1980,Worner1997,goluskin2012convection,goluskin2016penetrative} indicate that the dimensional mean temperature increases sublinearly with the heating rate, which, in nondimensional terms, implies that $\langle T \rangle$ decreases with $\Ra$. 
% \red{As explained by \citet{goluskin2012convection} and recently by \cite{Wang2020}, the scaling arguments proposed by \citet{GroSjfm2000a} for Rayleigh-B\'{e}nard convection can be applied to the top boundary layer of IH convection and imply that $\langle T \rangle \sim \Pran^{-1/3}\Ra^{-1/3}$ when $\Pran \lesssim \Ra^{-1/4}$ and $\langle T \rangle \sim \Ra^{-1/4}$ otherwise.}
Scaling arguments for Rayleigh-B\'{e}nard convection \citep{GroSjfm2000a} can be applied to the top boundary layer of IH convection \citep{goluskin2012convection,Wang2020} and imply that $\langle T \rangle \sim \Pran^{-1/3}\Ra^{-1/3}$ when $\Pran \lesssim \Ra^{-1/4}$ and $\langle T \rangle \sim \Ra^{-1/4}$ otherwise. 
% \red{Additionally, \cite{Wang2020} predict that .} 
The dependence of these predictions on the Rayleigh number agrees with rigorous lower bounds on $\langle T \rangle$ for both finite and infinite $Pr$ \citep{lu2004bounds,whitehead2011internal} up to logarithmic corrections. However, as mentioned in \cite{goluskin2012convection} and discussed in \S\ref{sec:heuristics} of this paper, understanding the scaling of $\langle wT\rangle$ in IH convection requires additional information pertaining to the bottom boundary layer.

In contrast to $\langle T\rangle$, the behaviour of the mean vertical convective heat flux $\langle wT \rangle$ remains fascinatingly unclear. While $\langle wT \rangle$ appears to increase with \Ra\ in experiments \citep{kulacki1972thermal} and three-dimensional simulations \citep{goluskin2016penetrative}, it displays little variation with respect to \Ra\ and non-monotonic behaviour in two-dimensional simulations~\citep{goluskin2012convection,Wang2020}. Power-law fits to experimental and numerical data summarised by \cite{goluskin2016internally} suggest that
\begin{equation}\label{e:wt-conjecture}
    \langle wT \rangle \approx \frac12 - \sigma \Ra^{-\alpha}\, ,
\end{equation}
with prefactors $\sigma \approx 1$ and exponents $\alpha$ between $0.05$ and $0.1$. Physical theories corroborating this power-law behaviour are lacking and the best rigorous mathematical result available to date is the uniform bound
$\langle wT \rangle \leq 1/2$ \citep{goluskin2012convection}.

%The present work sheds new light on the behaviour of $\langle wT \rangle$ in two different ways. First, we adapt heuristic arguments proposed by \cite{malkus1954heat} and \cite{priestley1954vertical} and by \cite{spiegel1963generalization} for Rayleigh--B\'enard convection to predict the exponent $\alpha$ in~\eqref{e:wt-conjecture}. Second, we investigate whether rigorous bounds that depend explicitly on the Rayleigh number and improve on the uniform bound of $1/2$ can be proved.

This work reports new \Ra-dependent upper bounds on $\langle wT \rangle$, some obtained computationally and some proved analytically, the derivation of which relies on two key ingredients. The first is a modern interpretation of the background method \citep{constantin1995variational,doering1994variational,doering1996variational} as a particular case of a broader framework for bounding infinite-time averages \citep{chernyshenko2014polynomial,Fantuzzi2016siads,chernyshenko2017relationship,tobasco2018optimal,goluskin2019ks,rosa2020optimal}. This interpretation makes it possible to formulate a variational bounding principle for $\langle wT \rangle$ even though, contrary to most past applications of the background method to convection problems \citep{doering1996variational,Constantin1996,Doering1998,constantin1999infinitePr,doering2001upper,Otero2002,Yan2004,Otero2004,doering2006bounds,whitehead2011ultimate,whitehead2012slippery,Goluskin2015,goluskin2016rough}, this quantity is not directly related to the thermal dissipation. %The optimization variables in this variational problem are a (rescaled) background temperature field and so-called balance parameters~\citep{Nicodemus1997}. 
The second key ingredient is a minimum principle, already invoked by~\cite{goluskin2012convection} and proved in Appendix~\ref{sec:proof_T_positive}, stating that the temperature of the fluid is either no smaller than that of the top and bottom plates, or approaches such a state exponentially quickly. Similar results have proved extremely useful in the context of Rayleigh--B\'enard convection \citep{constantin1999infinitePr,Yan2004,otto2011rayleigh,goluskin2016rough,choffrut2016upper} and, as we shall demonstrate, appear essential for proving \Ra-dependent bounds on $\langle wT \rangle$ for IH convection at large $\Ra$.

The rest of this work is structured as follows. Section~\ref{sec:model} describes the flow configuration and the corresponding governing equations. Heuristic scaling arguments for the mean vertical heat flux are presented in \S\ref{sec:heuristics}. In \S\ref{sec:fomulating_bound}, we derive two variational principles to bound $\langle wT \rangle$ rigorously from above: one that does not consider the minimum principle for the temperature, and one that enforces it by means of a Lagrange multiplier. Computational and analytical bounds obtained with the former are presented in \S\ref{sec:results}, while bounds obtained numerically with the latter are discussed in \S\ref{sec:positive_temp}, along with obstacles to analytical constructions. Section~\ref{sec:conclusion} offers concluding remarks.

Throughout the paper, overlines indicate averages over the horizontal directions and infinite time, while
angled brackets denote averages over the dimensionless volume $\Omega = [0,L_x]\times [0,L_y] \times [0,1]$ and infinite time. Precisely, for any scalar-valued function $f(\boldsymbol{x},t)$,
\begin{subequations}
	\begin{gather}
	\overline{f} = \limsup_{\tau\to\infty} \frac{1}{\tau L_{x}L_{y}}\int_{0}^{\tau} \int_{0}^{L_{x}}\int_{0}^{L_{y}}  
	f(\boldsymbol{x},t)\, \textrm{d}x \,\textrm{d}y \,\textrm{d}t,\\
	\langle f \rangle = \limsup_{\tau\to\infty} \frac{1}{\tau} \int_{0}^{\tau} \fint_{\Omega} f(\boldsymbol{x},t)\, \dVol \,\textrm{d}t,
	\end{gather}
\end{subequations}
where $\fint_\Omega (\cdot) \dVol = (L_xL_y)^{-1} \int_\Omega (\cdot) \dVol$ is the spatial average. Note that $\overline{f}=\overline{f}(z)$ depends only on the vertical coordinate $z$. We also write $\lVert f \rVert_{2}$ and $\lVert f \rVert_{\infty}$ for the usual $L^{2}$ and $L^{\infty}$ norms of univariate functions $f(z)$ on the interval $[0,1]$. Derivatives of univariate functions with respect to $z$ will be denoted by primes.

\section{Governing equations}
\label{sec:model}
% \subsection{Governing Equations}
% \label{sec:gov_eqs}
We consider a layer of fluid confined between two no-slip plates that are separated by a vertical distance $d$ and are held at the same constant temperature, which may be taken to be zero without loss of generality. The fluid has density $\rho$, kinematic viscosity $\nu$, thermal diffusivity $\kappa$, thermal expansion coefficient $\alpha$, specific heat capacity $c_p$, and is heated uniformly at a volumetric rate $Q$. This corresponds to the configuration denoted by IH1 in \cite{goluskin2016internally}.  To simplify the discussion we assume that the layer is periodic in the horizontal ($x$ and $y$) directions with periods $L_x d$ and $L_y d$, but all results presented in \S\ref{sec:results} and \S\ref{sec:positive_temp} will be independent of the domain aspect ratios $L_x$ and $L_y$. 

To make the problem nondimensional we use $d$ as the length scale, $d^2/\kappa$ as the time scale and $d^2 Q/\kappa$ as the temperature scale. Under the Boussinessq approximation, the Navier--Stokes equations governing the motion of the fluid in the nondimensional domain $\Omega = [0,L_x] \times [0,L_y] \times [0,1]$ are \citep{goluskin2016internally}
\begin{subequations}
\label{e:governing-equations}
\begin{align}
    \bnabla \cdot \boldsymbol{u} &= 0\, , \label{continuit} \\
    % \frac{\partial \boldsymbol{u}}{\partial t}
    \partial_t \boldsymbol{u} 
    + \boldsymbol{u}\cdot \bnabla \boldsymbol{u} + \bnabla p &= \Pran ( \bnabla^{2}\boldsymbol{u} + \Ra T \uvec{z})\, , \label{nondim_momentum} \\
    % \frac{\partial T}{\partial t} 
    \partial_t T 
    + \boldsymbol{u}\cdot \bnabla T  &=  \bnabla^{2}T + 1,
    \label{nondim_energy}
\end{align}
\end{subequations}
with boundary conditions
\begin{subequations}
	\label{e:boundary-conditions-all}
	\begin{align}
	\Eval{\boldsymbol{u}}{z=0}{} = \Eval{\boldsymbol{u}}{z=1}{} =0, & \label{bc_velocity}\\
	\Eval{T}{z=0}{} = \Eval{T}{z=1}{} = 0.&  \label{bc_T_IH1} 
	\end{align}
\end{subequations}
The dimensionless Prandtl and Rayleigh numbers are defined as
\begin{equation}
    \Pran = \frac{\nu}{\kappa}, \qquad
    \Ra = \frac{g \alpha Q d^{5}}{\rho c_p \nu \kappa^{2}}. 
\end{equation}
The former measures the ratio of momentum and heat diffusivity and is a property of the fluid, while the latter measures the destabilising effect of internal heating compared with the stabilising effect of diffusion and is the control parameter in this study.
 
For any value of \Ra\ and \Pran, equations~{(\ref{e:governing-equations}a--c)} admit the steady solution, $\boldsymbol{u}=\boldsymbol{0}$, $T=\frac12 z(1-z)$, which represents a purely conductive state. This solution is globally asymptotically stable for any values of the horizontal periods $L_x$ and $L_y$ when $\Ra < 26\,926$ \citep{goluskin2016internally} and is linearly unstable when $\Ra$ is larger than a critical threshold $\Ra_L \approx 37\,325$ \citep{Debler1959}, the exact value of which depends on the horizontal periods. Sustained convection ensues in this regime, but has also been observed at subcritical Rayleigh numbers \citep{Tveitereid1978}. 
Our goal is to characterize the mean vertical convective heat flux through the layer, $\langle wT \rangle$, as a function of \Ra.

\section{Heuristic scaling arguments}\label{sec:heuristics}

Phenomenological predictions for the variation of $\langle wT \rangle$ with the Rayleigh number can be derived by coupling the total heat budget through the layer with scaling assumptions for characteristic length scales $\delta_{\scriptscriptstyle T}$ and $\varepsilon_{\scriptscriptstyle T}$ of the
lower and upper thermal boundary layers, respectively. These length scales can be defined such that
$\langle T \rangle/\delta_{\scriptscriptstyle T} = \mathcal{F}_0$ and
$\langle T \rangle/\varepsilon_{\scriptscriptstyle T} = \mathcal{F}_1$. 
%Integration of the heat budget \label{eq:nondim_energy} in a statistically stationary regime 
Averaging~\eqref{nondim_energy} over space and infinite time
indicates that $\delta_{\scriptscriptstyle T}$ and
$\varepsilon_{\scriptscriptstyle T}$ satisfy
\begin{equation}
  \frac{\langle T \rangle}{\delta_{\scriptscriptstyle T}}+\frac{\langle T \rangle}{\varepsilon_{\scriptscriptstyle T}}
  %=\mathcal{F}_{0}+\mathcal{F}_{1}
  =1,
  \label{eq:scaling_budget}
\end{equation}
while the second identity in~\eqref{e:boundary-fluxes-vs-wT} yields
\begin{equation}
    \langle wT\rangle = \frac12  - \frac{\langle T \rangle}{\delta_{\scriptscriptstyle T}}.
\end{equation}
For the sake of definiteness, assume that the mean temperature and $\delta_{\scriptscriptstyle T}$ decay as power laws in \Ra, that is, $\langle T\rangle= \Ra^{-\alpha_{0}}/\sigma_{0}$ and $\delta_{\scriptscriptstyle T} = \Ra^{-\alpha_1}/\sigma_1$ with
$\alpha_{0},\alpha_1\geq 0$. If $\langle wT\rangle$ approaches a constant as $\Ra$ is raised, 
then $\langle T\rangle \leq O(\delta_{\scriptscriptstyle T})$ and $\alpha_1 \leq \alpha_0$, the inequality being strict if
$\langle wT\rangle \rightarrow 1/2$. Moreover, \eqref{eq:scaling_budget} implies that
\begin{equation}
    \frac{1}{\varepsilon_{\scriptscriptstyle T}} = \sigma_{0}\Ra^{\alpha_{0}}-\sigma_{1}\Ra^{\alpha_{1}}.
\end{equation}
The scalings behind IH convection with the isothermal boundary conditions \eqref{e:boundary-conditions-all} are therefore necessarily subtle, because the leading scaling of $\varepsilon_{\scriptscriptstyle T}$ (hence, of $\langle T \rangle$) and the correction implied by $\delta_{\scriptscriptstyle T}$ both play a crucial role. 
Any heuristic argument therefore needs to distinguish between the physics associated with the unstably stratified flow near the upper boundary from the (very different) stably stratified flow near the lower boundary.
In particular, one must determine whether $\langle T\rangle$ reduces at the same rate as $\delta_{T}$, meaning that $\langle wT\rangle$ tends to a constant value in the range $[0,1/2)$ determined by the relative magnitude of the prefactors $\sigma_0$ and $\sigma_1$, or slightly faster, implying that $\langle wT\rangle$ approaches $1/2$ as the Rayleigh number is raised.
%  
% \red{Indeed, despite boundary layers necessarily accompanying fluid that is well-mixed in the core, theoretically the average (dimensionless) temperature $\langle T\rangle$ could reduce at the same rate as $\delta_{T}$ or slightly faster, to produce an asymptotic value of $\langle wT\rangle$ in the range $[0,1/2)$ or equal to $1/2$, respectively.} 

%In contrast,
%for IH convection with an insulating lower boundary,
%where $\langle wT\rangle = 1/2-\overline{T}(0)$, the
%conjectured scaling based on empirical observations and lower
%bounds on $\langle T\rangle$ leads immediately to
%$\overline{T}(0)\sim \Ra^{-1/3}$, corresponding to
%$\alpha_{0}=1/3$ above \citep{goluskin2016internally}.

%As noted by \cite{whitehead2011internal}, if the upper boundary
% layer maintains the state of marginal stability \red{conjectured} by \cite{malkus1954heat} and
% \cite{priestley1954vertical}, then $\varepsilon_{\scriptscriptstyle T}$ adjusts itself to
% make the local Rayleigh number $Ra_{\varepsilon_{\scriptscriptstyle T}}$ (based on the average temperature and depth of the upper boundary layer) constant. 
As noted by \cite{whitehead2011internal}, one way to derive a scaling for $\varepsilon_{\scriptscriptstyle T}$ is to assume that the upper boundary layer maintains a state of marginal stability \citep{malkus1954heat,priestley1954vertical}. In this case, $\varepsilon_{\scriptscriptstyle T}$ adjusts itself such that the local Rayleigh number $Ra_{\varepsilon_{\scriptscriptstyle T}}$, based on the average temperature and depth of the upper boundary layer, remains constant.
Expressing $Ra_{\varepsilon_{\scriptscriptstyle T}}$ in terms of $\Ra$ to leading order, we conclude that
\begin{equation}
  Ra_{\varepsilon_{\scriptscriptstyle T}}=\langle T\rangle \varepsilon_{\scriptscriptstyle T}^{3}\Ra\sim \sigma_{0}^{4}\Ra^{1-4\alpha_{0}}\, ,
\end{equation}
should be independent of \Ra, which implies that $\alpha_{0}=1/4$, as noted by \cite[][Table 2]{goluskin2012convection} and
consistent with the scalings proposed by \citet[][see regimes
III$_{\infty}$ and IV$_{u}$]{Wang2020}. Alternatively, if one
uses an argument based on balancing a characteristic free-fall
velocity $\sqrt{\Pran \Ra\langle T\rangle}$ with the velocity
scale $1/\varepsilon_{\scriptscriptstyle T}$ implied by diffusion at the wall
\citep{spiegel1963generalization}, then to leading order
\begin{equation}
\varepsilon_{\scriptscriptstyle T}\sqrt{\Pran\Ra\langle T\rangle}\sim \sigma_{0}^{\frac32}Pr^{\frac12}\Ra^{1-3\alpha_{0}},
\end{equation}
is independent of $\Ra$, implying that $\alpha_{0}=1/3$. In either case
($\alpha_{0}=1/4$ or $\alpha_{0}=1/3$), the resulting scaling corresponds to
the first term in the asymptotic expansion of $\varepsilon_{\scriptscriptstyle T}$ and
% does not provide information about the small but crucial
% correction due to $\delta_{\scriptscriptstyle T}$.
does not provide any information about the
correction due to $\delta_{\scriptscriptstyle T}$, which is crucial to determine the asymptotic behaviour of $\langle wT\rangle$.

% Conjectures for the scaling of $\langle wT\rangle$ require an
% independent argument relating to $\delta_{\scriptscriptstyle T}$. Arguably the simplest,
% but not necessarily the most faithful, 
The simplest argument relating to $\delta_{\scriptscriptstyle T}$, although not necessarily the most faithful,
comes from assuming that in
some vicinity of the lower boundary there is a balance between
heating and diffusion because the flow is stably stratified. In
terms of the dimensionless variables used here, heating over
$\delta_{\scriptscriptstyle T}$ is proportional to $\delta_{\scriptscriptstyle T}$ and diffusion is equal to
$\langle T\rangle/\delta_{\scriptscriptstyle T}$, which implies that
$\delta_{\scriptscriptstyle T}^{2} \sim \langle T\rangle$. This requires $\alpha_{1}=\alpha_{0}/2$, leading to
$\alpha_{1}=1/8$ or $\alpha_{1}=1/6$ for scaling of $\varepsilon_{\scriptscriptstyle T}$ based
on \cite{malkus1954heat} or \cite{spiegel1963generalization},
respectively, and therefore to
$\langle wT\rangle = 1/2 - \sigma_{1}\Ra^{-\frac18}/\sigma_{0}$ or
$\langle wT\rangle = 1/2 - \sigma_{1}\Ra^{-\frac16}/\sigma_{0}$. Assuming
that $\max(\overline{T})$ scales in the same way as
$\langle T\rangle$, meaning that the average temperature is
approximately uniform away from boundaries, the possibility
that $\delta_{\scriptscriptstyle T}^{2} \sim \langle T\rangle$ (so
$\alpha_{1}=\alpha_{0}/2$) is in reasonably good agreement with data from experiments
and simulations \citep[][table 3.2]{goluskin2016internally}.

An alternative argument might consider a Richardson number $Ri$ at
the lower boundary layer to quantify the destablising effects of
turbulence relative to the stabilising effects of the density
stratification. In terms of dimensionless quantities, the density
stratification is $\langle T\rangle /\delta_{\scriptscriptstyle T}$ and we assume that the destabilising
shear across the lower boundary scales according to
$\sqrt{\langle T\rangle} / \delta_{\scriptscriptstyle T}$. Together, these scales imply
that $Ri\sim \delta_{\scriptscriptstyle T}$. This is significant because, if the flow
tends towards a state of marginal stability, then $Ri=1/4$
according to the Miles--Howard criterion for steady, laminar,
parallel and inviscid shear flow \citep{MilJjfm1961a,
  HowLjfm1961a}. We would therefore conclude that either
$\langle wT\rangle = 1/2 - \sigma_{1}\Ra^{-\frac14}/\sigma_{0}$ or
$\langle wT\rangle = 1/2 - \sigma_{1}\Ra^{-\frac13}/\sigma_{0}$,
corresponding to \cite{malkus1954heat} or
\cite{spiegel1963generalization} respectively. The latter scaling
would be consistent with the conjectured bound for insulating
lower boundary conditions, but, unlike the scaling argument outlined in the previous paragraph, is far from the wide range of scaling possibilities that have been inferred from experiments and simulations \citep{goluskin2016internally}. Indeed, available data is too scattered to provide conclusive information about the asymptotic behaviour of $\langle wT\rangle$, highlighting the need for further experiments and simulations in addition to the rigorous bounds pursued here.

\section{Formulation of rigorous bounds}
\label{sec:fomulating_bound}

We now turn our attention to bounding $\langle wT \rangle$ rigorously from above. In \S\S\ref{sec:AFM}--\ref{sec:U-explicit} we ignore the minimum principle for the temperature field and our analysis can be seen as a ``classical'' application of the background method. In \S\ref{sec:bounds-for-positive-temperature}, instead, we improve the analysis by taking the minimum principle into account through a Lagrange multiplier.
%xplicitly enforce nonnegativity of the temperature field, which provably holds if the initial temperature for~\eqref{nondim_energy} is nonnegative \citep{goluskin2012convection} and holds at asymptotically long times in general.

\subsection{Bounds via auxiliary functional}
\label{sec:AFM}
A rigorous upper bound on $\langle wT \rangle$ can be derived using the simple observation that the infinite-time average of the time derivative of any bounded function $\mathcal{V}\{\boldsymbol{u},T\}$ along solutions of the governing equations~(\ref{e:governing-equations}\textit{a--c}) vanishes, so
\begin{equation}
    \langle wT \rangle =
    \limsup_{\tau\to\infty} \frac{1}{\tau}\int_0^\tau \left[\fint_\Omega w T   + \frac{{\rm d}}{ {\rm d}t}\mathcal{V}\{\boldsymbol{u},T\} \; {\rm d}\boldsymbol{x} \right] {\rm d}t.
\end{equation}
In particular, if $\mathcal{V}$ can be chosen such that
\begin{equation}\label{e:bounding-condition}
    \mathcal{S}\{\boldsymbol{u},T\} := U - \fint_\Omega w T  + \frac{{\rm d}}{ {\rm d}t}\mathcal{V}\{\boldsymbol{u},T\} \; {\rm d}\boldsymbol{x} \geq 0
\end{equation}
for all time along solutions of~(\ref{e:governing-equations}\textit{a--c}) for some constant $U$, then $\langle wT \rangle \leq U$. The goal, then, is to construct $\mathcal{V}$ that satisfies~\eqref{e:bounding-condition}  with the smallest possible $U$.

While the evolution equations~\eqref{nondim_momentum} and~\eqref{nondim_energy} cannot be solved explicitly for all possible initial conditions, when $\mathcal{V}$ is given they can be used to derive an explicit expression for $\mathcal{S}$ as a function of $\boldsymbol{u}$ and $T$ alone. Then, to ensure that~\eqref{e:bounding-condition} holds along solutions~(\ref{e:governing-equations}\textit{a--c}) pointwise in time, it suffices to enforce that $\mathcal{S}\{\boldsymbol{u},T\}$ be nonnegative when viewed as a functional on the space of time-independent incompressible velocity fields and temperature fields that satisfy the boundary conditions,
\begin{equation}
    \mathcal{H} := \{ (\boldsymbol{u}, T) : \;\text{{(\ref{e:boundary-conditions-all}\textit{a,b})}, horizontal periodicity and } \bnabla \cdot \boldsymbol{u} = 0 \}.
\end{equation}
We therefore search for a function $\mathcal{V}$ and constant $U$ such that $\mathcal{S}\{\boldsymbol{u},T\} \geq 0$ for all velocity and temperature fields in $\mathcal{H}$.
%
% In fact, since we are interested in infinite-time averages and temperature fields $T(\boldsymbol{x},t)$ satisfying~\eqref{nondim_energy} are instantaneously nonnegative on the domain $\Omega$ at for all sufficiently large times, it suffices to impose that $\mathcal{S}$ be nonnegative on the positive subspace
% %
% \begin{equation}
%     \mathcal{H}_+ := \{ (\boldsymbol{u}, T) \in \mathcal{H}: \; T(\boldsymbol{x}) \geq 0 \text{ on } \Omega \}.
% \end{equation}
% %
This key relaxation makes this approach tractable and, remarkably, it may not be overly conservative. In fact, if the governing equations in~\eqref{e:governing-equations} were well posed (which is not currently known) and solutions eventually remained in a compact subset $\mathcal{K}$ of $\mathcal{H}$, then optimizing $\mathcal{V}$ over a sufficiently general class of functions whilst imposing $\mathcal{S}\{\boldsymbol{u},T\} \geq 0$ for all $\boldsymbol{u}$ and $T$ in $\mathcal{K}$, would yield an upper bound $U$ exactly equal to the largest possible value of $\langle wT \rangle$ \citep{rosa2020optimal}.
% would bring no loss of generality if the governing equations were well posed (which is not currently known). In this case, in fact, optimizing $\mathcal{V}$ over a sufficiently general class of functions would yield an upper bound $U$ exactly equal to the largest value of $\langle wT \rangle$ over solutions of~{(\ref{e:governing-equations}\textit{a--c})} \citep{rosa2020optimal}, \red{provided a sufficient weak class of solutions in an appropriate compact set is considered. }

Unfortunately, the construction of such an optimal $\mathcal{V}$ is currently beyond the reach of both analytical and computational methods. Nevertheless, progress can be made if we restrict the search to quadratic $\mathcal{V}$ in the form
\begin{equation}
    \mathcal{V}\{\boldsymbol{u},T\} = \fint_{\Omega} \frac{a}{2 \Pran \Ra} |\boldsymbol{u}|^{2} + \frac{b}{2}|T|^{2} - [\psi(z) + z - 1 ]T\, \dVol \, , \label{problemfunctional}
\end{equation}
where the nonnegative scalars $a$ and $b$ and the function $ \psi(z)$ are to be optimized such that~\eqref{e:bounding-condition} holds for the smallest possible $U$. As shown by~\cite{chernyshenko2017relationship}, this choice of $\mathcal{V}$ amounts to using the background method~\citep{constantin1995variational,doering1994variational,doering1996variational}: the profile $\frac{1}{b}[\psi(z)+z-1]$ corresponds to the background temperature field, while the scalars $a$ and $b$ are the so-called balance parameters. Note that the $z-1$ term could be removed by redefining $\psi(z)$, but we isolate it to simplify the analysis in what follows. Note also that, due to the periodicity in the horizontal directions, a ``symmetry reduction" argument following ideas in \citet[][Appendix A]{goluskin2019ks} proves that there is no loss of generality in taking $\psi$ to depend only on the vertical coordinate $z$. Similarly, one can show that the upper bound on $\langle wT \rangle$ cannot be improved by adding to $\mathcal{V}$ a term $\boldsymbol{\phi} \cdot \boldsymbol{u}$, proportional to the velocity field via a (rescaled) incompressible background velocity field $\boldsymbol{\phi}$.

To find an expression for $\mathcal{S}\{\boldsymbol{u},T\}$, we calculate the time derivative of the quadratic~$\mathcal{V}$ in~\eqref{problemfunctional} using the governing equations~{(\ref{e:governing-equations}\textit{b,c})}, substitute the resulting expression into~\eqref{e:bounding-condition}, and integrate various terms by parts using incompressibility and the boundary conditions~{(\ref{e:boundary-conditions-all}\textit{a,b})} to arrive at
\begin{multline}
    \mathcal{S}\{\boldsymbol{u},T\} =
    \fint_{\Omega} \bigg[ \frac{a}{R}|\bnabla \boldsymbol{u}|^{2} + b|\bnabla T|^{2} - (a-\psi')wT+(bz - \psi' - 1)\partial_z T + \psi \bigg]\, \dVol
    \\  + \psi(1)\overline{T}'(1) \, - (\psi(0)-1) \overline{T}'(0)+ 
   U - \frac12\, .
    \label{S_functional_OG}
\end{multline}
The best bound on $\langle wT \rangle$ that can be proved with quadratic $\mathcal{V}$ is %found upon solving the optimization problem
%
% \begin{align}\label{e:optimization-full}
%     \langle wT \rangle 
%     &\leq \inf_{U,\psi(z), a, b} \left\{U :\; \mathcal{S}\{\boldsymbol{u},T\} \geq 0 \; \forall (\boldsymbol{u},T) \in \mathcal{H}_+ \right\}\\
%     &\leq \inf_{U,\psi(z), a, b} \left\{U :\; \mathcal{S}\{\boldsymbol{u},T\} \geq 0 \; \forall (\boldsymbol{u},T) \in \mathcal{H} \right\}.
%     \label{e:optimization-full_1}
% \end{align}
\begin{equation}\label{e:optimization-full}
    \langle wT \rangle 
    \leq \inf_{U,\psi(z), a, b} \left\{U :\; \mathcal{S}\{\boldsymbol{u},T\} \geq 0 \; \forall (\boldsymbol{u},T) \in \mathcal{H} \right\}.
\end{equation}
The right-hand side of \eqref{e:optimization-full} is a linear optimisation problem because the optimisation variables, $U$, $\psi(z)$, $a$ and $b$, enter the constraint $\mathcal{S}\{\boldsymbol{u},T\} \geq 0$ and the cost $U$ linearly.
%We illustrate the significance of imposing the positivity of temperature, by solving the problem as stated without this principle being imposed, as in \eqref{e:optimization-full}. When the background method is used this would be the standard methodology. To handle \eqref{e:optimization-full_1} we introduce a Lagrange multiplier constraining $\mathcal{S}$.
% To highlight the role of the pointwise nonnegativity constraint on admissible temperature fields, which is not typically included in the ``classical'' formulation of the background method, we shall also consider the modified problem
% %
% \begin{equation}\label{e:optimization-classical}
%     \inf_{U,\phi(z), a, b} \left\{U :\; \mathcal{S}\{\boldsymbol{u},T\} \geq 0 \; \forall (\boldsymbol{u},T) \in \mathcal{H} \right\}.
% \end{equation}

\subsection{Fourier expansion}\label{sec:fourier-expansion}
The analysis and numerical implementation of the minimisation problem in~\eqref{e:optimization-full} can be considerably simplified by expanding the horizontally-periodic velocity and temperature as Fourier series,
%
% \begin{subequations}\label{e:Fourier}
% \begin{align}
%     T(x,y,z) &= \sum_{\boldsymbol{k}} \hat{T}_{\boldsymbol{k}}(z)\textrm{e}^{i(k_x x + k_y y)}, 
%     \\
%     \boldsymbol{u}(x,y,z) &= \sum_{\boldsymbol{k}} \hat{\boldsymbol{u}}_{\boldsymbol{k}}(z)\textrm{e}^{i(k_x x + k_y y)}.
% \end{align}
% \end{subequations}
\begin{equation}\label{e:Fourier}
    \begin{bmatrix}
    T(x,y,z)\\\boldsymbol{u}(x,y,z)
    \end{bmatrix}
    = \sum_{\boldsymbol{k}} 
    \begin{bmatrix}
    \hat{T}_{\boldsymbol{k}}(z)\\ \hat{\boldsymbol{u}}_{\boldsymbol{k}}(z)
    \end{bmatrix}
    \textrm{e}^{i(k_x x + k_y y)}.
\end{equation}
The sums are over wavevectors $\boldsymbol{k}=(k_x,k_y)$ compatible with the horizontal periods $L_x$ and $L_y$. We denote the magnitude of each wavevector by 
%$k = (k_x^2 + k_y^2)^{\frac12}$. 
$k = \sqrt{k_{\smash{x}}^2 + k_{\smash{y}}^2}$.
The complex-valued Fourier amplitudes satisfy the complex-conjugate relations $\hat{\boldsymbol{u}}_{-\boldsymbol{k}}=\hat{\boldsymbol{u}}_{\boldsymbol{k}}^*$ and $\hat{T}_{-\boldsymbol{k}} = \hat{T}_{\boldsymbol{k}}^*$, as well as the no-slip and isothermal boundary conditions. Using incompressibility and writing $\hat{w}_{\boldsymbol{k}}$ for the vertical component of $\hat{\boldsymbol{u}}_{\boldsymbol{k}}$, these can be expressed as
\begin{subequations}\label{e:Fourier-bc}
    \begin{gather}
        \label{e:Fourier-bc-wk}
        \hat{w}_{\boldsymbol{k}}(0) = \hat{w}_{\boldsymbol{k}}'(0) = \hat{w}_{\boldsymbol{k}}(1) = \hat{w}_{\boldsymbol{k}}'(1)= 0,\\
        \hat{T}_{\boldsymbol{k}}(0)=\hat{T}_{\boldsymbol{k}}(1) = 0.
        \label{e:Fourier-bc-Tk}
\end{gather}
\end{subequations}
% Moreover, the positivity constraint on admissible temperature fields in the set $\mathcal{H}_+$ implies $\hat{T}_0(z) \geq 0$ for all $z\in[0,1]$.

After inserting the Fourier expansions~{(\ref{e:Fourier})} into~\eqref{S_functional_OG} and applying standard estimates based on the incompressibility condition and Young's inequality to replace the horizontal Fourier amplitudes $\hat{u}_{\boldsymbol{k}}$ and $\hat{v}_{\boldsymbol{k}}$ as a function of $\hat{w}_{\boldsymbol{k}}$, the functional $\mathcal{S}\{\boldsymbol{u},T\}$ can be estimated from below as
\begin{equation}\label{e:S-estimate}
    \mathcal{S}\{\boldsymbol{u},T\} \geq \mathcal{S}_{0}\{\hat{T}_0\} + \sum_{\boldsymbol{k}} \mathcal{S}_{\boldsymbol{k}} \{\hat{w}_{\boldsymbol{k}},\hat{T}_{\boldsymbol{k}}\},
\end{equation}
where
\begin{equation}\label{S0}
    \mathcal{S}_{0}\{\hat{T}_0\} = \int_{0}^{1} \left[b \vert \hat{T}_{0}' \vert^{2}+ (bz-\psi')\hat{T}_{0}' + \psi  \right] \textrm{d}z
    +  \psi(1)\hat{T}_{0}'(1)-(\psi(0)-1)\hat{T}_{0}'(0) + U - \frac12
\end{equation}
and 
\begin{multline}\label{Sk}
    \mathcal{S}_{\boldsymbol{k}}\{\hat{w}_{\boldsymbol{k}},\hat{T}_{\boldsymbol{k}}\} = 
    \int_{0}^{1} \left[ 
    \frac{a}{R}\left( 
    \frac{1}{k^2} \abs{\hat{w}_{\boldsymbol{k}}''}^{2}
    + 2\abs{\hat{w}_{\boldsymbol{k}}'}^{2} 
    + k^{2}\abs{\hat{w}_{\boldsymbol{k}}}^{2}  \right) \right.
    \\
    \left.
    + b\vert \hat{T}_{\boldsymbol{k}}' \vert^{2} + bk^{2}\vert \hat{T}_{\boldsymbol{k}} \vert^{2} 
    - (a-\psi')\hat{w}_{\boldsymbol{k}}\hat{T}_{\boldsymbol{k}}^{*} \right] \textrm{d}z. 
\end{multline}
%
% Equality holds in~\eqref{e:S-estimate} if $\boldsymbol{u}$ and $T$ depend on only one of the two horizontal coordinates, but not both, and on $z$.% In the above decomposition $\mathcal{S}_{\boldsymbol{k}}$ is the spectral constraint.
Equality holds in~\eqref{e:S-estimate} if $\boldsymbol{u}$ has only one nonzero horizontal component because, in this case, the Fourier-transformed incompressibilty condition yields an exact relation between $\hat{w}_{\boldsymbol{k}}$ and either $\hat{u}_{\boldsymbol{k}}$ or $\hat{v}_{\boldsymbol{k}}$, so Young's inequality is not needed to eliminate the latter.

Velocity and temperature fields with a single nonzero Fourier mode are admissible in the optimization problem~\eqref{e:optimization-full}, so the right-hand side of~\eqref{e:S-estimate} is nonnegative if and only if each term is nonnegative. Moreover, since the real and imaginary parts of the Fourier amplitudes $\hat{w}_{\boldsymbol{k}}$ and $\hat{T}_{\boldsymbol{k}}$ give identical and independent contributions to $\mathcal{S}_{\boldsymbol{k}}$, we may assume them to be real without loss of generality. Thus, we may replace the minimization problem in~\eqref{e:optimization-full} with
\begin{equation}\label{e:optimization-Fourier}
    \begin{aligned}
        \inf_{U,\psi(z),a,b} \quad &U\\
        \text{subject to} \quad
        &\mathcal{S}_0\{\hat{T}_0\} \geq 0 &&\forall \hat{T}_0: \eqref{e:Fourier-bc-Tk},\\
        &\mathcal{S}_{\boldsymbol{k}}\{\hat{w}_{\boldsymbol{k}},\hat{T}_{\boldsymbol{k}}\} \geq 0 &&\forall \hat{w}_{\boldsymbol{k}},\hat{T}_{\boldsymbol{k}}: \text{(\ref{e:Fourier-bc}\textit{a,b})}, \quad \forall \boldsymbol{k}\neq 0.
    \end{aligned}
\end{equation}
%
% Similarly, by dropping the positivity condition on $\hat{T}_0$,~\eqref{e:optimization-classical} may be replaced with
% \begin{equation}
%     \begin{aligned}
%         \inf_{U,\psi(z),a,b} \quad &U\\
%         \text{subject to} \quad
%         &\mathcal{S}_0\{\hat{T}_0\} \geq 0 &&\forall \hat{T}_0:\text{\eqref{e:Fourier-bc-Tk}},\\
%         &\mathcal{S}_{\boldsymbol{k}}\{\hat{w}_{\boldsymbol{k}},\hat{T}_{\boldsymbol{k}}\} \geq 0 &&\forall \hat{w}_{\boldsymbol{k}},\hat{T}_{\boldsymbol{k}}: \text{(\ref{e:Fourier-bc}\textit{a,b})}.
%     \end{aligned}
% \end{equation}
% %
Any choice of $U$, $\psi(z)$, $a$ and $b$ satisfying the constraints yields a rigorous upper bound on the mean vertical convective heat flux $\langle wT \rangle$. Following established terminology, we refer to the inequalities on $\mathcal{S}_{\boldsymbol{k}}$ as the spectral constraints.

Just like~\eqref{e:optimization-full}, the optimization problem in~\eqref{e:optimization-Fourier} is linear %, as again the cost and the constraints depend \textit{linearly} on $U$, $\psi(z)$, $a$ and $b$, 
and its optimal solution can be approximated using efficient numerical algorithms after discretization. For computational convenience, however, we simplify the spectral constraints by dropping all nonnegative terms that depend explicitly on $k$. Specifically, we replace the spectral constraints with the stronger, but simpler, single condition
\begin{equation}\label{e:simplified-spectral-constraint}
    % \Tilde{\mathcal{S}}\{\hat{w}_{\boldsymbol{k}}, \hat{T}_{\boldsymbol{k}} \} :=\int_0^1 
    % \frac{2a}{R}\abs{\hat{w}_{\boldsymbol{k}}'}^{2} 
    % + b\vert \hat{T}_{\boldsymbol{k}}' \vert^{2}
    % - (a-\psi')\hat{w}_{\boldsymbol{k}}\hat{T}_{\boldsymbol{k}}
    % \,{\rm d}z \geq 0
    %  \qquad \forall \hat{T}_{\boldsymbol{k}}, \hat{w}_{\boldsymbol{k}}:\text{~(\ref{e:Fourier-bc}\textit{a,b})}.
    \Tilde{\mathcal{S}}\{\hat{w}, \hat{T} \} :=\int_0^1 
    \frac{2a}{R}\abs{\hat{w}'}^{2} 
    + b\vert \hat{T}' \vert^{2}
    - (a-\psi')\hat{w}\hat{T}
    \,{\rm d}z \geq 0
     \qquad \forall \hat{T}, \hat{w}:\text{~(\ref{e:Fourier-bc}\textit{a,b})}
\end{equation}
and solve
\begin{equation}\label{e:optimization-Fourier-simplified}
    \begin{aligned}
        \inf_{U,\psi(z),a,b} \quad &U\\
        \text{subject to} \quad
        &\mathcal{S}_0\{\hat{T}_0\} \geq 0 &&\forall \hat{T}_0: \eqref{e:Fourier-bc-Tk},\\
        &\tilde{\mathcal{S}}\{\hat{w},\hat{T}\} \geq 0 &&\forall \hat{w},\hat{T}: \text{(\ref{e:Fourier-bc}\textit{a,b})}
    \end{aligned}
\end{equation}
instead of~\eqref{e:optimization-Fourier}.
%Observe that $\mathcal{S}_{\boldsymbol{k}} \geq \Tilde{\mathcal{S}}$ for all wavenumbers $\boldsymbol{k}$. 
This simplification leads to suboptimal bounds on $\langle wT \rangle$ at a fixed \Ra\ but, as discussed in Appendix \ref{sec:comparison_Sk}, still captures the qualitative behaviour of the optimal ones. On the other hand, considering the simplified spectral constraint~\eqref{e:simplified-spectral-constraint} allows for significant computational savings when optimizing bounds numerically, because it removes the need to consider a large set of wavenumbers and it enables implementation using simple piecewise-linear basis functions (cf. appendix~\ref{sec:computational-details}). This allows for discretization of~\eqref{e:optimization-Fourier-simplified} and its generalization~\eqref{e:optimization-Fourier-positive} derived in \S\ref{sec:bounds-for-positive-temperature} below on very fine meshes, which is essential to resolve sharp boundary layers in $\psi$ accurately.

%retains the qualitative behaviour of the problem, though will change the quantitative results. Now in optimising computationally, $\Tilde{\mathcal{S}}$ is the same for all $\boldsymbol{k}$ and therefore needs to be implemented only once rather than over a large range of wavevectors to ensure all critial wavenumbers are accounted for. This brings significant computational savings and enables the use of finer discretised meshes, which, as we will see later, are important in capturing sharp boundary layers. Further justification for this simplification is shown in appendix \ref{sec:comparison_Sk}. Moreover~\eqref{e:simplified-spectral-constraint} makes it possible to use piecewise-linear basis for both $\hat{T}_{\boldsymbol{k}}$ and $\hat{w}_{\boldsymbol{k}}$, which simplifies the implementation. 

\subsection{Explicit formulation}\label{sec:U-explicit}
To simplify the analysis (but not the numerical implementation) of~\eqref{e:optimization-Fourier} it is convenient to eliminate the explicit appearance of $U$. This can be done upon observing that, given any profile $\psi(z)$ and balance parameters $a$, $b$, the smallest $U$ for which $\mathcal{S}_0\{\hat{T}_0\}$ is nonnegative over admissible $\hat{T}_0$ is
\begin{multline}\label{e:explicit-U-partial}
    U^* = \frac12 + \sup_{\substack{\hat{T}_{\boldsymbol{k}}(0)=0\\\hat{T}_{\boldsymbol{k}}(1) = 0}}
    \bigg\{
    -\int_{0}^{1} \left[b \vert \hat{T}_{0}' \vert^{2}+ (bz-\psi')\hat{T}_{0}' + \psi  \right] \textrm{d}z
    \\[-4ex]
    -\psi(1)\hat{T}_{0}'(1)
    +(\psi(0)-1)\hat{T}_{0}'(0) 
    \bigg\}.
\end{multline}

By modifying any admissible $\hat{T}_0$ in infinitesimally thin layers near $z=0$ and $z=1$, it is possible to show that this value is finite if and only if 
\begin{equation}\label{e:psi-bcs}
\tag{\theequation\textit{a,b}}
    \psi(0)=1 \quad \text{and} \quad \psi(1)= 0.
\end{equation}
Then, the optimal temperature field in~\eqref{e:explicit-U-partial} is given by
% $\hat{T}_0'(z) = \frac{1}{2b}[\psi'(z) - b(z-\frac12)]$
\begin{equation}\label{e:optimal-T0-general}
    \hat{T}_0'(z) = \frac{\psi'(z)+1}{2b} - \frac{z}{2} + \frac14,
\end{equation}
and we obtain
\begin{equation}
    \label{eq_for_U}
    U^* = \frac12 + \frac{1}{4b}\left\| bz-\frac{b}{2} - \psi'(z) - 1\right\|_2^2 - \int_{0}^{1} \psi\, \textrm{d}z.
\end{equation}
Thus, we may replace the minimization problem~\eqref{e:optimization-Fourier} with the more explicit version
\begin{equation}\label{e:explicit-U-optimization}
     \begin{aligned}
        \inf_{\psi(z),a,b} \quad &\frac12 + \frac{1}{4b}\left\| bz-\frac{b}{2} - \psi'(z) - 1\right\|_2^2 - \int_{0}^{1}\psi\, \textrm{d}z\\[1ex]
        \text{subject to} \quad
        &\psi(0)=1,\\
        &\psi(1)=0,\\
        &\mathcal{S}_{\boldsymbol{k}}\{\hat{w},\hat{T}\} \geq 0 \qquad\forall \hat{w},\hat{T}: \text{(\ref{e:Fourier-bc}\textit{a,b})}.
    \end{aligned}
\end{equation}
Note that although this formulation is not suitable for numerical implementation because the cost function is not convex with respect to $b$, it is more convenient when attempting to prove an upper bound on $\langle wT \rangle$ analytically.

\subsection{Restriction to nonnegative temperature fields}\label{sec:bounds-for-positive-temperature}

The upper bounding principle derived above can be improved by imposing a minimum principle, which guarantees that temperature fields solving the Boussinesq equations~{(\ref{e:governing-equations}\textit{a--c})} are nonnegative in the domain $\Omega$ at large time. More precisely, Appendix~\ref{sec:proof_T_positive} proves that the fluid's temperature is nonnegative on $\Omega$ at all times if it is so initially, and the negative part of the temperature decays exponentially quickly otherwise.
%When the initial temperature for the Boussinesq equations~{(\ref{e:governing-equations}\textit{a--c})} is nonnegative everywhere in the domain $\Omega$, the minimum principle ensures that $T$ remains nonnegative at all times \citep{goluskin2012convection}. In fact, temperature is nonnegative at sufficiently large times even for arbitrary initial temperatures, this is demonstrated with an argument in appendix \ref{sec:proof_T_positive}, similar to that from \cite{foias1987attractors} where the authors demonstrate the equivalent situation for Rayleigh-B\'enard convection.

Since $\langle wT \rangle$ is determined by the long-time behaviour of the velocity and temperature fields, the minimum principle enables us to replace the upper bound~\eqref{e:optimization-full} with
\begin{equation}\label{e:optimization-full-positive}
    \langle wT \rangle 
    \leq \inf_{U,\psi(z), a, b} \left\{U :\; \mathcal{S}\{\boldsymbol{u},T\} \geq 0 \quad \forall (\boldsymbol{u},T) \in \mathcal{H}_+ \right\},
\end{equation}
where the space $\mathcal{H}$ of admissible velocity and temperature fields has been replaced with its subset
\begin{equation}
    \mathcal{H}_+ := \{ (\boldsymbol{u}, T) \in \mathcal{H}: \; T(\boldsymbol{x}) \geq 0 \text{ on } \Omega \}.
\end{equation}
The constraint in the modified optimisation problem~\eqref{e:optimization-full-positive} is clearly weaker than the original one in~\eqref{e:optimization-full}, so imposing the minimum principle allows for a better bound on $\langle wT \rangle$ in principle. This is indeed the case at large Rayleigh numbers, as shall be demonstrated by numerical results in \S\ref{sec:results} and \S\ref{sec:positive_temp}.% for the set $\mathcal{H}$. 
%for $T$ cannot worsen the best upper bound on the mean vertical heat flux available within our framework. 
% That is,
% \begin{align}
%     \nonumber
%     \langle wT \rangle 
%     &\leq \inf_{U,\psi(z), a, b} \left\{U :\; \mathcal{S}\{\boldsymbol{u},T\} \geq 0 \quad \forall (\boldsymbol{u},T) \in \mathcal{H}_+ \right\}\\
%     &\leq \inf_{U,\psi(z), a, b} \left\{U :\; \mathcal{S}\{\boldsymbol{u},T\} \geq 0 \quad \forall (\boldsymbol{u},T) \in \mathcal{H} \right\}.
% \end{align}
% %
%However, as shall be demonstrated by numerical results in the next section, the inequality is strict when the Rayleigh number is sufficiently large for the set $\mathcal{H}$.

In order to impose the inequality constraint $\mathcal{S}\{\boldsymbol{u},T\} \geq 0$ for nonnegative temperatures, but relax it when $T$ is negative on a nonnegligible subset of the domain, we effectively use a Lagrange multiplier. Specifically, we search for a positive bounded linear functional $\mathcal{L}$ such that $\mathcal{S}\{\boldsymbol{u},T\} \geq \mathcal{L}\{T\}$ for all pairs $(\boldsymbol{u},T) \in \mathcal{H}$, which satisfy only the boundary condition and incompressibility. Indeed, the positivity of $\mathcal{L}$ implies $\mathcal{S}\{\boldsymbol{u},T\} \geq \mathcal{L}\{T\} \geq 0$ if $T$ is nonnegative, as desired. When $T$ is negative on a subset of the domain, instead, $\mathcal{L}\{T\}$ need not be positive and the constraint on $\mathcal{S}\{\boldsymbol{u},T\}$ is relaxed. 
%
%
% One simple choice is to take 
% %
% \begin{equation}\label{e:positive-linear-functional}
%     \mathcal{L}\{T\} = \fint_\Omega q'(z) T(\boldsymbol{x}) \, \dVol = -\fint_\Omega q(z) \frac{\partial T}{\partial z} \dVol
% \end{equation}
% %
% where $q:(0,1)\to \mathbb{R}$ is differentiable, square-integrable and nondecreasing function to be optimised alongside the bound $U$, the profile $\psi(z)$, and the balance parameter $a,b$. Here, the derivative $q'$ can be interpreted as a Lagrange multiplier for the condition $T\geq 0$ and, loosely speaking, quantifies the amount by which the upper bound $U$ on $\langle wT \rangle$ would increase if this condition were relaxed. However, note that the last integral in~\eqref{e:positive-linear-functional} makes sense even when $q$ is not differentiable, so in what follows we take $\mathcal{L}\{T\} = -\fint_\Omega q(z) \frac{\partial T}{\partial z} \dVol$ for a generic square-integrable and nondecreasing function $q$; analysis in Appendix~\ref{app:riesz-representation} proves
% that this brings no loss of generality. 
%% Using the Riesz representation theorem and a horizontal averaging argument based on the periodicity of $T$, it is possible to show 
%

Analysis in Appendix~\ref{app:riesz-representation} proves
that there is no loss of generality in taking
\begin{equation}\label{e:positive-linear-functional}
    \mathcal{L}\{T\} = -\fint_\Omega q(z)\partial_z T \dVol,
\end{equation}
where $q:(0,1)\to \mathbb{R}$ is a nondecreasing function, square-integrable but not necessarily continuous, to be optimised alongside the bound $U$, the profile $\psi(z)$, and the balance parameter $a,b$. If $q$ were differentiable, one could integrate by parts to obtain
\begin{equation}
    \mathcal{L}\{T\} = \fint_\Omega q'(z) T\; \dVol
\end{equation}
and identify $q'$ as a standard Lagrange multiplier for the condition $T(\boldsymbol{x}) \geq 0$; working with~\eqref{e:positive-linear-functional} simply removes the differentiability requirement from $q$. Moreover, the value of $\mathcal{L}\{T\}$ in~\eqref{e:positive-linear-functional} does not change when $q$ is shifted by a constant by virtue of the vertical boundary conditions on $T$, so we may normalize $q$ such that
\begin{equation}\label{e:q-condition}
    \int_0^1 q(z) \, {\rm d}z = \psi(1)-\psi(0).
\end{equation}
%The details of these argument are beyond the scope of this work, but the key conclusion is that 
The upper bound~\eqref{e:optimization-full-positive} can therefore be replaced with
\begin{align}\label{e:optimization-full-positive-with-multiplier}
    \nonumber
    \langle wT \rangle 
    \leq \inf_{U,\psi(z), q(z) a, b} 
    \bigg\{U :\; &q(z) \text{ nondecreasing, \eqref{e:q-condition} and}\\
    &\mathcal{S}\{\boldsymbol{u},T\} + \fint_\Omega q(z) \partial_z T \dVol\geq 0 \quad \forall (\boldsymbol{u},T) \in \mathcal{H},\bigg\}.
\end{align}
Observe that setting $q(z)=-1$ causes the integral $\fint_\Omega q(z) \frac{\partial T}{\partial z} \dVol$ to vanish because $T$ is zero at the top and bottom boundaries, so this choice of $q$ results in the upper bound~\eqref{e:optimization-full} derived without the minimum principle for $T$.
% The improved minimisation problem in~\eqref{e:optimization-full-positive} can be expanded using a Fourier series exactly as explained in \S\ref{sec:fourier-expansion}. In the context of the optimisation problem \eqref{e:optimization-full-positive}, only the zeroth Fourier mode, $\hat{T}_0$, in~\eqref{e:optimization-Fourier} needs to be constrained to be nonnegative on the interval $[0,1]$. This will ensure that the overall temperature field is nonnegative. To incorporate the nonnegativity constraint on $\hat{T}_{0}$, $S_{0}$ becomes 
% %
% \begin{equation}\label{e:lagrange-multiplier}
%     \mathcal{S}_0(\hat{T}_0) \geq \int_0^1 q'(z) \hat{T}_0(z) {\rm d}z\, , \qquad \forall \hat{T}_0: \hat{T}_0(0)=0=\hat{T}_0(1),
% \end{equation}
% %
% where $q'(z)\geq 0$ is a Lagrange multiplier to be optimised alongside $U$, $\psi(z)$, $a$ and $b$. The Lagrange multipler $q'(z)$ works by penalising negative temperature and is proportional (in the sense of an inner product) to the amount by which $U$ would increase (i.e. deteriorate as an upper bound) if the condition $T\geq 0$ were relaxed. Therefore, setting $q'(z)=0$ in \eqref{e:lagrange-multiplier} is equivalent to optimising over $\mathcal{H}$ rather than $\mathcal{H}_{+}$.  

The minimisation problem in~\eqref{e:optimization-full-positive-with-multiplier} can be expanded using a Fourier series exactly as explained in \S\ref{sec:fourier-expansion}. %, \red{ where once again $\mathcal{S}_{\boldsymbol{k}}$ is replaced with the wavenumber independent $\tilde{\mathcal{S}}$ for ease of computation, as explained in appendix \ref{sec:comparison_Sk} }. 
After estimating the functionals $\mathcal{S}_{\boldsymbol{k}}$ with $\tilde{\mathcal{S}}$, one concludes that $\langle wT \rangle$ is bounded above by the optimal value of the following linear optimization problem:%Integrating the right-hand side by parts using the boundary conditions on $\hat{T}_0$,~\eqref{e:optimization-Fourier} may be replaced with
\begin{equation}\label{e:optimization-Fourier-positive}
    \begin{aligned}
        \inf_{U,a,b,\psi(z),q(z)} \quad &U\\
        \text{subject to} \quad
        &\mathcal{S}_0\{\hat{T}_0\} + \int_0^1 q(z) \hat{T}_0'(z) {\rm d}z \geq 0\, \qquad\forall \,\hat{T}_0: \hat{T}_0(0)=0=\hat{T}_0(1),\\
        &\Tilde{\mathcal{S}}\{\hat{w},\hat{T}\} \geq 0 \qquad\forall \hat{w},\hat{T}: \text{(\ref{e:Fourier-bc}\textit{a,b})},\\
        &q(z) \text{ nondecreasing}.
    \end{aligned}
\end{equation}
If one does not simplify the spectral constraints, one obtains a very similar problem where the inequality on $\tilde{\mathcal{S}}$ is replaced by the same $\boldsymbol{k}$-dependent inequalities appearing in~\eqref{e:optimization-Fourier}. This problem gives a quantitatively better bound on $\langle wT \rangle$ but has a higher computational complexity than~\eqref{e:optimization-Fourier-positive} and, just as in Appendix~\ref{sec:comparison_Sk}, we do not expect qualitative improvements in the behaviour of the bounds with \Ra.
%Note that we recover~\eqref{e:optimization-Fourier} if we set $q=0$, which amounts to ignoring the minimum principle for the temperature.

The constant $U$ can be eliminated from~\eqref{e:optimization-Fourier-positive} by following the same procedure outlined in \S\ref{sec:U-explicit} in order to obtain a minimization problem that is more suitable for analysis, but less convenient for computations.
%
%by choosing a suitable integration constant when integrating~\eqref{e:lagrange-multiplier} by parts. 
Indeed, the normalization condition~\eqref{e:q-condition} implies that the functional $\mathcal{S}_0\{\hat{T}_0\}$ is minimised when
%$\hat{T}_0'(z) = \frac{1}{2b}[\psi'(z) - b(z-\frac12) - q(z)]$, 
\begin{equation}\label{e:optimal-T0-positive}
    \hat{T}_0'(z) = \frac{\psi'(z) - q(z)}{2b} - \frac{z}{2} + \frac14
\end{equation}
and the best $U$ for given $a$, $b$, $\psi(z)$ and  $q(z)$ is
\begin{equation}
    \label{eq_for_U_q}
    U^* = \frac12 + \frac{1}{4b}\left\| bz-\frac{b}{2} - \psi'(z) + q(z) \right\|_2^2 - \int_{0}^{1} \psi\, \textrm{d}z\,\,.
\end{equation}
As one would expect, this expression reduces to \eqref{eq_for_U} when $q(z)=-1$.
It is also possible to show that, since we are only interested in nonnegative $\hat{T}_0$, the conditions $\psi(0)=1$ and $\psi(1)=0$ in \S\ref{sec:U-explicit} may be weakened into inequalities $\psi(0)\leq 1$ and $\psi(1)\leq 0$. Thus, when the minimum principle for the temperature is imposed, the optimisation problem~\eqref{e:explicit-U-optimization} relaxes into
\begin{equation}\label{e:explicit-U-optimization-positive}
     \begin{aligned}
        \inf_{\psi(z),q(z),a,b} \quad &\frac12 + \frac{1}{4b}\left\| bz-\frac{b}{2} - \psi'(z) + q(z) \right\|_2^2 - \int_{0}^{1}\psi\, \textrm{d}z\\[1ex]
        \text{subject to} \quad
        &\psi(0) \leq 1,\, \psi(1) \leq 0,\\
        &q(z) \text{ nondecreasing, \eqref{e:q-condition}},\\
        &\Tilde{\mathcal{S}}\{\hat{w},\hat{T}\} \geq 0 \qquad\forall \hat{w},\hat{T}: \text{(\ref{e:Fourier-bc}\textit{a,b})}.
    \end{aligned}
\end{equation}

Finally, observe that setting $\psi(z)=0$, $q(z)=0$, $a=0$ and letting $b\to 0$ in this minimisation problem yields a sequence of feasible solutions with optimal cost approaching $1/2$ from above irrespective of \Ra. Thus, the uniform bound $\langle wT \rangle \leq 1/2$ proved by~\cite{goluskin2012convection} can be recovered within our approach by taking the auxiliary functional in~\eqref{problemfunctional} to be
\begin{equation}
    \mathcal{V}\{\boldsymbol{u},T\} = \fint_{\Omega}  (1-z) T\, \dVol,
\end{equation}
which corresponds to the flow's potential energy measured with respect to the upper boundary. As shown in \S\ref{sec:positive_temp}, however, more general choices can yield better bounds.

\section{Optimal bounds for general temperature fields}
\label{sec:results}

The best upper bounds on $\langle wT \rangle$ implied by problem~\eqref{e:optimization-Fourier} can be approximated numerically at any fixed Rayleigh number either by deriving and solving the corresponding nonlinear Euler--Lagrange equations~\citep{Plasting2003,Wen2013,Wen2015}, or by discretising it into a semidefinite programme (SDP) \citep{fantuzzi2015construction,Fantuzzi2016PRE,Tilgner2017,Tilgner2019,fantuzzi2018bounds}. Here, we choose the latter approach because it preserves the linearity of~\eqref{e:optimization-Fourier};
% and~\eqref{e:optimization-Fourier-positive}, so convergence to the global optimum is guaranteed, and because the inequalities on $\psi(0)$, $\psi(1)$ and $q'$ are readily imposed. 
details of our numerical implementation are summarised in Appendix \ref{sec:computational-details}. Numerically optimal solutions to~\eqref{e:optimization-Fourier} for $10^3 \leq \Ra \leq 10^7$
% and~\eqref{e:optimization-Fourier-positive} 
are presented in~\S\ref{sec:results_no_p},
% and~\S\ref{sec:positive_temp}, respectively.
while suboptimal but analytical bounds are proved in~\S\ref{sec:analytical-bound}.

\subsection{Numerically optimal bounds}
\label{sec:results_no_p}

%As described in appendix \ref{sec:computational-details}, we use a finite element discretisation to solve problem~\eqref{e:optimization-Fourier} for $10^3 \leq \Ra \leq 10^8$, on a uniform mesh.

% The number of Legendre basis functions used to discretise the tunable function $\psi$ and the the unknown fields $\hat{T}_0$, $\hat{T}_{\boldsymbol{k}}$ and $\hat{w}_{\boldsymbol{k}}$ until the upper bound changed by less than $10^{-9}$. We also fixed the horizontal periods to $L_x =2$ and $L_y=2$ for simplicity, but changing these values does not alter the results significantly.

\begin{figure}
    \centering
    \includegraphics[width=\textwidth]{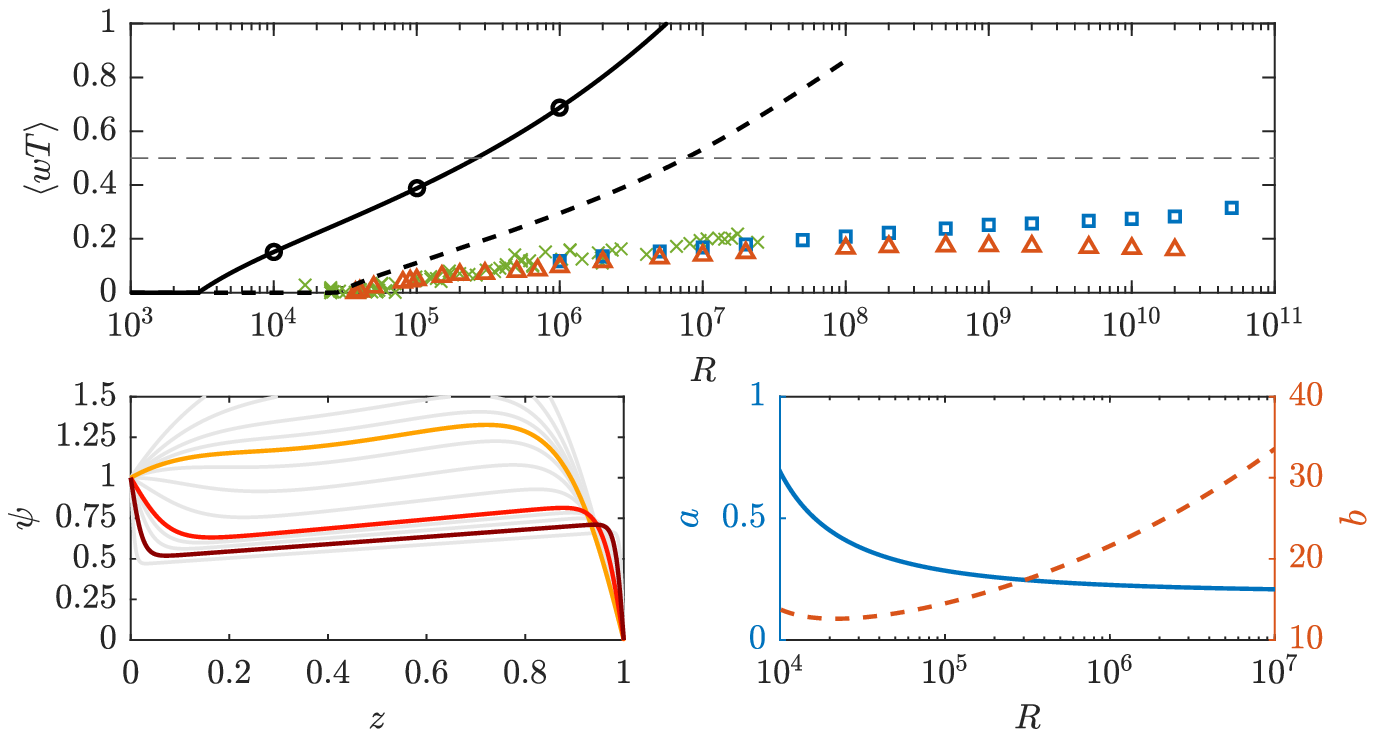}
    \\[-2ex]
    \begin{tikzpicture}[overlay]
        \node at (-5.2,6.8) {\textit{(a)}};
        \node at (-5.2,3.1) {\textit{(b)}};
        \node at (1.2,3.1) {\textit{(c)}};
    \end{tikzpicture}
    \caption{
    \textit{(a)} Optimal bounds on $\langle wT \rangle$ obtained by solving the wavenumber-dependent problem~\eqref{e:optimization-Fourier} with horizontal periods $L_x = L_y = 2$ (\dashedrule) and the simplified problem~\eqref{e:optimization-Fourier-simplified} (\solidrule). Also plotted are experiments by \cite{kulacki1972thermal} (\mycross{matlabgreen}), 2D DNSs by \cite{goluskin2012convection} (\mytriangle{matlabred}) and 3D DNSs by \cite{goluskin2016penetrative} (\mysquare{matlabblue}). Circles mark values of \Ra\ at which optimal profiles $\psi(z)$ are plotted in panel \textit{(b)}. The dashed horizontal line (\dashedrule) is the uniform bound $\langle wT \rangle \leq 1/2$.
    \textit{(b)} Optimal profiles $\psi(z)$ at $\Ra = 10^4$ ({\color{colorbar5}\solidrule}), $10^5$ ({\color{colorbar10}\solidrule}) and $10^6$ ({\color{colorbar16}\solidrule}).
    \textit{(c)} Optimal balance parameters $a$ ({\color{matlabblue}\solidrule}, left axis) and $b$ ({\color{matlabred}\dashedrule}, right axis) as a function of \Ra.
    }
    \label{fig:optimal-bounds-general-T}
\end{figure}

Figure~\ref{fig:optimal-bounds-general-T} compares the numerically optimal upper bounds $U$ on the mean vertical heat transfer $\langle wT \rangle$ to the experimental data by \cite{kulacki1972thermal} and the DNS data by \cite{goluskin2012convection} and \cite{goluskin2016penetrative}. The bounds were calculated by solving the minimization problem~\eqref{e:optimization-Fourier} for a fluid layer with horizontal periods $L_x = L_y = 2$, and with the simplified problem~\eqref{e:optimization-Fourier-simplified}, which is independent of wavevectors and, therefore, of $L_x$ and $L_y$. As expected, the bounds obtained with~\eqref{e:optimization-Fourier} are zero when the Rayleigh number is smaller than the energy stability limit $\Ra_E \approx 29\,723$, which differs slightly from the value $26\,927$ reported by \cite{goluskin2016internally} due to our choice of horizontal periods. They then increase monotonically with $\Ra$, showing the same qualitative behaviour as the bounds computed with the simplified problem~\eqref{e:optimization-Fourier-simplified}, which reach the value of $1/2$ at $\Ra = 259\,032$. Both sets of results exceed the the uniform upper bound $\langle wT \rangle \leq 1/2$ for sufficiently large Rayleigh numbers. The apparent contradiction is due to the fact that the uniform bound relies on the minimum principle for the temperature, which was disregarded in the formulation~\eqref{e:optimization-Fourier}.

Numerically optimal profiles of $\psi(z)$ for the simplified bounding problem~\eqref{e:optimization-Fourier-simplified} at selected Rayleigh numbers and the variation of the optimal balance parameters $a$ and $b$ with \Ra\ are illustrated in Figures~\ref{fig:optimal-bounds-general-T}\textit{(b)} and \ref{fig:optimal-bounds-general-T}\textit{(c)} respectively. As expected from the structure of the indefinite term $(a-\psi')\hat{w}\hat{T}$ of the functional $\Tilde{\mathcal{S}}$, the derivative $\psi'$ in the bulk of the domain approaches the value of $a$ as \Ra\ is raised and leads to the formation of two boundary layers. 
% Secondary layers near the edges of the principal boundary layers arise from bifurcations in the critical wavenumbers for which the constraint on the full $\mathcal{S}_{\boldsymbol{k}}$ in~\eqref{e:optimization-Fourier} is active; see also \cite{Nicodemus1997a,Plasting2003,Wen2015,Fantuzzi2016PRE} for similar observations. 
Moreover, the asymmetry of the boundary layers reflects qualitatively the asymmetry of the IH convection problem we are studying, which is characterized by a stable thermal stratification near the bottom boundary ($z=0$) and an unstable one near the top ($z=1$). However, note that while $\psi$ is related to the background temperature field, it is not a physical quantity and need not behave nor scale like the mean temperature in turbulent convection.

\begin{figure}
    \centering
    % \includegraphics[scale=0.7]{Figures/T0_samples_quinopt.eps}
    % \includegraphics[scale=0.6]{Figures/T0_near_half.eps}
    % \includegraphics[scale=0.6]{Figures/U_near_half.eps}
    % \begin{tikzpicture}[overlay]
    % \node[rotate=90,font=\fontsize{8}{8},anchor=center] at (0,2.35) {$\hat{T}_0(z)$};
    % \end{tikzpicture}
    % \hspace{5pt}
    % \includegraphics[scale=1]{Figures/optimal_T0_bottom_BL_general_T.eps}
    % \includegraphics[scale=1]{Figures/optimal_T0_general_T.eps} \hspace{-10pt}
    % \includegraphics[scale=1]{Figures/optimal_T0_top_BL_general_T.eps} \hspace{-5pt}
    % \input{Figures/colorbar-hot}\\[1ex]
    %\includegraphics[scale=0.75]{Figures/T0_near_half2.eps}
    % \includegraphics[scale=0.9]{Figures/TandU_near_half.eps}
    \includegraphics[width=\linewidth]{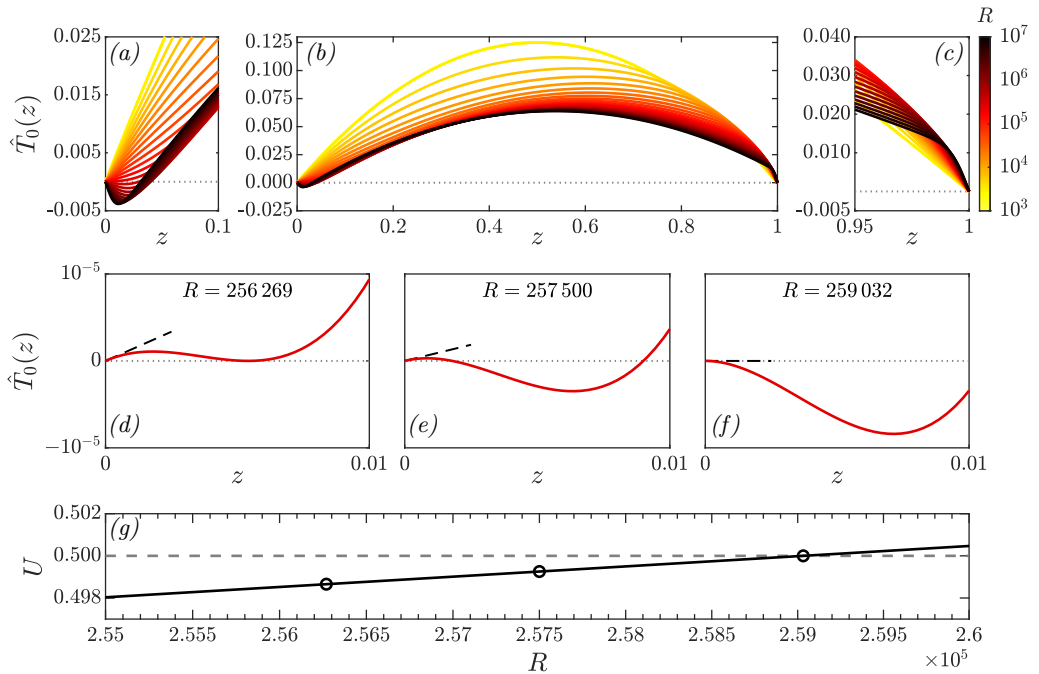}
    \begin{tikzpicture}[overlay]
        \node at (-5.15,8.5) {\textit{(a)}};
        \node at (-2.55,8.5) {\textit{(b)}};
        \node at (5.75,8.5) {\textit{(c)}};
        \node at (-5.15,3.6) {\textit{(d)}};
        \node at (-1.2,3.6) {\textit{(e)}};
        \node at (2.8,3.6) {\textit{(f)}};
        \node at (-5.15,2.25) {\textit{(g)}};
    \end{tikzpicture}
    \vspace*{-1ex}
    \caption{
    \textit{Top}: Critical temperature $\hat{T}_0(z)$ recovered using~\eqref{e:optimal-T0-general}. Colors indicate the Rayleigh number. Panels \textit{(a)} and \textit{(c)} show details of the boundary layers.
    %\textit{(d)} Detailed view of the temperature field near the bottom boundary for $\Ra = 2.56 \times 10^5$ ({\color{matlabyellow}\dashedrule}), $2.57 \times 10^5$ ({\color{matlabblue}\solidrule}\mytriangle{matlabblue}), $ 2.60 \times 10^5$ ({\color{matlabred}\dottedrule}\mycirc{matlabred}) and $2.62 \times 10^5$ ({\color{matlabgreen}\dotdashedrule}).
    \textit{Middle}: Detailed view of $\hat{T}_0$ for $\Ra = 256\,269$, $257\,500$, and $259\,032$. Dashed lines (\dashedrule) are tangent to $\hat{T}_0$ at $z=0$. In~\textit{(d)}, $\hat{T}_0$ is nonnegative and has minimum of zero inside the layer. In~\textit{(e)}, $\hat{T}_0$ is initially positive but has a negative minimum. In~\textit{(f)}, $\hat{T}_0'(0)=0$ and there is no positive initial layer.
    %\textit{(e)} Upper bounds on $\langle wT \rangle$ with the values of $\Ra$ at which temperature is plotted in \textit{(d)} highlighted, the vertical dashed line marks the smallest $\Ra$ at which $\hat{{T}}_0(z)$ becomes negative away from the wall and the horizontal dashed line denotes the uniform bound of $1/2$ at which $\hat{{T}}_0'(0)=0$.
    \textit{Bottom}: Upper bounds $U$ on $\langle wT \rangle$. Circles mark the values of $\Ra$ considered in \textit{(d--f)} and $U=1/2$ at $\Ra = 259\,032$.
    %\textit{(d)} highlighted, the vertical dashed line marks the smallest $\Ra$ at which $\hat{{T}}_0(z)$ becomes negative away from the wall and the horizontal dashed line denotes the uniform bound of $1/2$ at which $\hat{{T}}_0'(0)=0$.
    %\textbf{[UPDATE WITH PROFILES AT SMALLER \Ra]}
    }
    \label{fig:optimal-T0-general-profiles}
\end{figure}

Other insightful observations can be made by considering the critical temperature fields $\hat{T}_{0}$, which minimize the functional $\mathcal{S}_0\{\hat{T}_0\}$ for the optimal choice of $\psi$, $a$, $b$ and $U$. These critical temperatures can be recovered upon integrating~\eqref{e:optimal-T0-general}, and are plotted in Figure~\ref{fig:optimal-T0-general-profiles} for a selection of Rayleigh numbers. As one might expect, when $\Ra$ is sufficiently small such that $U=0$, $\hat{T}_{0}=\frac12 z(1-z)$ coincides with the conductive temperature profile. With the onset of convection and increasing \Ra, boundary layers form at $z=0$ and $z=1$ and the maximum of $|\hat{T}_0|$ decreases. The profiles are also consistent with the uniform rigorous bound $\langle T \rangle \leq 1/12$ \citep{goluskin2016internally}. However, for sufficiently high Rayleigh numbers they are evidently not related to the horizontal and infinite-time averages of the physical temperature field, because they become negative near $z=0$. Interestingly, as shown in Figure~\ref{fig:optimal-T0-general-profiles}\textit{(d)}\&\textit{(e)}, this unphysical behaviour first occurs away from the boundary at $\Ra = 256\,269$, while the numerical upper bound on $\langle wT \rangle$ reaches the value of $1/2$ only at $259\,032$, when $\hat{T}_0'(0)=0$. The latter is not surprising because the identity $\overline{T}'(0)=1/2-\langle wT\rangle$, derived from~\eqref{e:boundary-fluxes-vs-wT} upon recognizing that $\mathcal{F}_0 = \overline{T}'(0)$, implies that an upper bound of 1/2 on $\langle wT \rangle$ is equivalent to a zero lower bound on $\overline{T}'(0)$, which is obtained when $\hat{T}_0'(0)$ vanishes.
%
% The uniform bound $\langle wT \rangle \leq 1/2$ is derived by \cite{goluskin2012convection} from the identity
% %
% \begin{equation}
%     \langle wT \rangle = \frac12 - \overline{T}'(0)\, ,
% \end{equation}
% %
% upon observing that $\overline{T}'(0)\geq 0$ for all physical temperature fields that are nonnegative pointwise inside the fluid layer. 
%
It is therefore clear that the bounding problems~\eqref{e:optimization-Fourier} and~\eqref{e:optimization-Fourier-simplified} fail to improve the uniform bound $\langle wT \rangle \leq 1/2$ proved by \cite{goluskin2012convection} at large \Ra\ due to a violation of the minimum principle for the temperature, which was not taken into account when formulating them.%~\eqref{e:optimization-Fourier}.

%%%%%%%%%%%%%%%%%%%%%%%%%%%%%%%%%%%%%%%%%%%%%%%%%%%
\subsection{A new Rayleigh-dependent analytical bound}
\label{sec:analytical-bound}

\begin{figure}
    \centering
    \begin{tikzpicture}[every node/.style={scale=0.95}]
    \draw[->,black,thick] (-3.25,0) -- (3,0) node [anchor=west] {$z$};
    \draw[->,black,thick] (-3,-0.25) -- (-3,3) node [anchor=south] {$\psi(z)$};
    \draw[matlabblue,thick] (-3,2.5) -- (-2,1) -- (1.75,2) -- (2.5,0);
    \node[anchor=east] at (-3,2.5) {$1$};
    \node[anchor=north] at (2.5,0) {$1$};
    \node[rotate=15] at (0,1.25) {$\psi'(z)=a$};
    \draw[dashed] (-2,0) node[anchor=north] {$\delta$} -- (-2,1);
    \draw[dashed] (1.75,0) node[anchor=north] {$1-\varepsilon$} -- (1.75,2);
    \draw[dashed] (-3,1) node[anchor=east] {$\sigma$} -- (-2,1);
    \end{tikzpicture}
    \caption{General piecewise-linear $\psi(z)$, parametrized by the boundary layer widths $\delta$ and $\varepsilon$, the bulk slope $a$ and the boundary layer height $\sigma$. In our proof, we set $\varepsilon=\delta$ and $\sigma = \tfrac{1}{2}-a \left(\tfrac{1}{2}-\delta \right)$ to obtain a profile that is anti-symmetric with respect to $z=1/2$.}
    \label{fig:phi_g}
\end{figure}
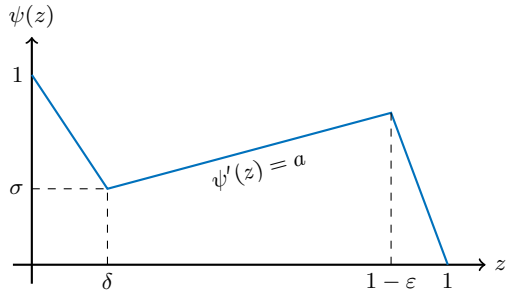

The numerical results in the previous section demonstrate that optimising $\psi$, $a$ and $b$ cannot improve the uniform bound $\langle wT \rangle \leq \frac12$ at arbitrarily large \Ra. Nevertheless, it is possible to derive better bounds analytically over a finite range of Rayleigh numbers by considering piecewise-linear profiles $\psi(z)$ with two boundary layers, such as the one sketched in Figure~\ref{fig:phi_g}. 
% The numerically optimal profiles in Figure~\ref{fig:optimal-bounds-general-T}, display symmetry with respect to the vertical midpoint $z=\frac12$ of the fluid layer, we impose anti-symmetry and take
% \red{Considering a profile anti-symmetric about the vertical midpoint we take}
Even though the numerically optimal profiles in Figure~\ref{fig:optimal-bounds-general-T} show no symmetry with respect to the vertical midpoint $z=\frac12$, we impose anti-symmetry and take
\begin{equation}
    \psi(z)= 
\begin{dcases}
    % 1- \left(\frac{a+1}{2\delta} -a \right)z ,&  0\leq z \leq \delta\\
    % az +\frac{1}{2}(1-a),  &  \delta \leq z \leq 1-\delta \\
    % \left(\frac{a+1}{2\delta} -a \right)(1-z), & 1-\delta \leq z \leq 1,
    1- \left(\tfrac{a+1}{2\delta} -a \right)z ,&  0\leq z \leq \delta\\
    az +\tfrac{1}{2}(1-a),  &  \delta \leq z \leq 1-\delta \\
    \left(\tfrac{a+1}{2\delta} -a \right)(1-z), & 1-\delta \leq z \leq 1,
\end{dcases}
\label{phicomplex}
\end{equation}
where $\delta$ is a boundary layer width to be specified later. This considerably simplifies the algebra, at the cost of a quantitatively (but not qualitatively) worse bound on $\langle wT \rangle$. Our goal is to determine values for $\delta$, $a$ and $b$ such that $\psi$ satisfies the constraints in the reduced optimisation problem~\eqref{e:explicit-U-optimization}, while trying to minimise its objective function.

To show that the functional $\mathcal{S}_{\boldsymbol{k}}$ is nonnegative, observe that the only sign-indefinite term in~\eqref{Sk} is
\begin{equation}\label{e:Sk-lower-bound}
    \int_0^1 (a-\psi')\hat{w}\hat{T} \, \textrm{d}z =
    \frac{a+1}{2\delta}\int_{[0,\delta] \cup [1-\delta,1]} \hat{w}\hat{T} \,\textrm{d}z.
\end{equation}
As detailed in Appendix~\ref{appendix:spectral estimates}, this term can be estimated using the fundamental theorem of calculus, the Cauchy--Schwarz inequality, and the boundary conditions~{(\ref{e:Fourier-bc}\textit{a,b})} %on $\hat{w}_{\boldsymbol{k}}$ and $\hat{T}_{\boldsymbol{k}}$ 
to conclude that $\mathcal{S}_{\boldsymbol{k}}$ is nonnegative if
\begin{equation}\label{e:delta-inequality}
    \delta \leq 
    \frac{8 \sqrt{2 a b}}{ (a+1) \sqrt{\Ra}}.
\end{equation}

According to the analysis in \S\ref{sec:U-explicit}, any choice of $a$, $b$ and $\delta$ satisfying this inequality produces the upper bound
\begin{align}
    \langle wT \rangle 
    &\leq \frac{1}{2} 
    + \frac{1}{4b}  \left\| bz -\frac{b}{2} - \psi'(z) - 1 \right\|_2^{2} 
    - \int^{1}_{0} \psi(z)\, \textrm{d}z
    \nonumber \\
    &= \frac{b}{48} + \frac{(a+1)^2}{8b\delta} - \frac{(a+1)^2}{4b}.
    % \nonumber \\
    % &\leq \frac{b}{48} + \frac{(a+1)^2}{8b\delta}.
\end{align}
%
% In the last inequality, we have dropped the positive term $(a+1)^2/(4b)$ to simplify the optimisation of $a$, $b$ and $\delta$. 
This bound is clearly minimised when $\delta$ is as large as allowed by~\eqref{e:delta-inequality}, leading to
%is as large as allowed by~\eqref{e:delta-inequality}, so we set
%%
%\begin{equation}
%    \delta = \frac{8 \sqrt{2 a b}}{ (a+1) \sqrt{\Ra}}
%\end{equation}
%%
%and obtain
%
\begin{equation}\label{e:implicit-analytical-bound}
    \langle wT \rangle 
    \leq \frac{b}{48} + \frac{ (a+1)^3 \sqrt{\Ra}}{64 b \sqrt{2 a b}} - \frac{(a+1)^2}{4b}.
\end{equation}

The values of $a$ and $b$ minimising the right-hand side of this inequality satisfy
\begin{subequations}
\label{e:optimal-ab}
    \begin{gather}
        -\sqrt{b} + \frac{\sqrt{R}}{64\sqrt{2a^3}}(a+1)(5a -1 )=0,\\
        1 + \frac{12(a+1)^2}{b^2} - \frac{9(a+1)^2\sqrt{R}}{8\sqrt{2ab^5}}= 0,
    \end{gather}
\end{subequations}
and can be computed numerically for fixed $\Ra$ to obtain upper bounds plotted as a blue dot-dashed line in Figure~\ref{fig:U_comp_num_analytic}.

Fully analytical bounds can instead be obtained if we drop the last term from~\eqref{e:implicit-analytical-bound}. In this case, the optimal $a$ and $b$ are found explicitly as
\begin{equation}
a = \frac15 \qquad\text{and}\qquad b = \frac95 \left(\frac{\Ra}{2}\right)^\frac15%\frac{9}{5\cdot2^{\frac15}} \Ra^{\frac15},
\label{a_and_b_scaling}
\end{equation}
such that we obtain
\begin{equation}\label{e:explicit-analytical-bound}
    \langle wT \rangle 
    \leq 2^{-\frac{21}{5}} R^{\frac15}\, .
\end{equation}
This bound, plotted as a solid line in Figure~\ref{fig:U_comp_num_analytic}, is smaller than the uniform bound of 1/2 up to $\Ra = 2^{16}=65\,536$, which is approximately 2.43 times larger than the energy stability threshold. 

\begin{figure}
    \centering
    \includegraphics[scale=1]{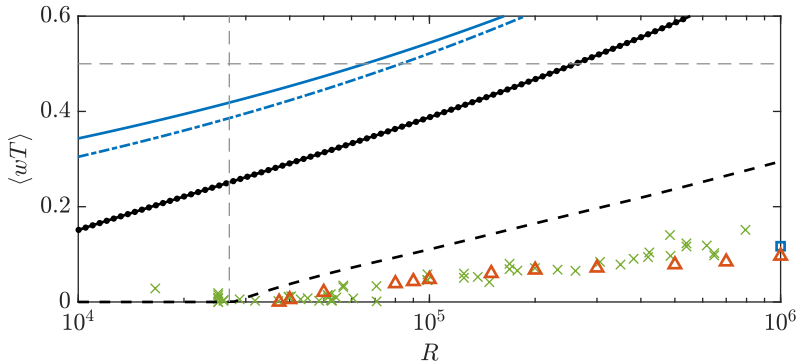}
    \caption{Comparison between the semi-analytical bounds computed with~\eqref{e:implicit-analytical-bound} and~{(\ref{e:optimal-ab}\textit{a,b})} ({\color{matlabblue}\dotdashedrule}), 
    the explicit bound~\eqref{e:explicit-analytical-bound} ({\color{matlabblue}\solidrule}), 
    the numerically optimal bounds obtained in \S\ref{sec:results_no_p} with~\eqref{e:optimization-Fourier} (\dashedrule, $L_x=L_y=2$) and with~\eqref{e:optimization-Fourier-simplified} ($\bullet$), and
    experimental and DNS data (see Figure~\ref{fig:optimal-bounds-general-T} for a key of the symbols).  %and suboptimal bounds obtained from~\eqref{e:explicit-U-optimization} with piecewise-linear $\psi$: semi-analytical bounds computed with~\eqref{e:implicit-analytical-bound} and~{(\ref{e:optimal-ab}\textit{a,b})} ({\color{matlabblue}\dotdashedrule}); the explicit bound~\eqref{e:explicit-analytical-bound} ({\color{matlabblue}\solidrule}). 
    A dashed vertical line ({\color{matlabgray}\dashedrule}) indicates the smallest energy stability threshold, $\Ra \approx 26\,926$, while a dashed horizontal line ({\color{matlabgray}\dashedrule}) indicates the uniform upper bound $\langle wT \rangle \leq \frac12$.
    }
    \label{fig:U_comp_num_analytic}
\end{figure}
%Further quantitative improvements over a wider (but still finite) range Rayleigh numbers could in principle be found by constructing more general piecewise-linear $\psi$ as sketched in Figure~\ref{fig:phi_g}. The analysis in appendix~\ref{appendix:spectral estimates} still applies, but the algebra becomes much more cumbersome because one must work with five tunable parameters ($a$, $b$, $\delta$, $\varepsilon$ and $\sigma$) instead of three. 
%Nevertheless, these parameters can be optimised numerically with MATLAB's \texttt{fmincon} function and the resulting bounds on the heat flux are shown by the black dashed line in Figure~\ref{fig:U_comp_num_analytic}.
%
%
%%%%%%%%%%%%%%%%%%%%%%%%%%%%%%%%%%%%
\section{Optimal bounds for nonnegative temperature fields}
\label{sec:positive_temp}

The analysis and computations in \S\ref{sec:results} demonstrate that, if one wants to improve on the uniform bound $\langle wT \rangle \leq 1/2$ at very large \Ra, one must invoke the minimum principle for the temperature explicitly to avoid unphysical critical temperatures $\hat{T}_0$ that are negative in the interior of the layer. As discussed in \S\ref{sec:bounds-for-positive-temperature}, upper bounds on $\langle wT \rangle$ that take this constraint into account can be found by solving~\eqref{e:optimization-Fourier-positive}. 
% \red{To start off with numerically optimal results obtained by solving \eqref{e:optimization-Fourier-positive} are presented, followed by \S\ref{sec:p_discussion} which demonstrates the conditions under which proving such a result fails mathematically. }
Numerically optimal solutions to \eqref{e:optimization-Fourier-positive} are presented in \S\ref{sec:num-opt-bnds}. Subsection~\ref{sec:p_discussion}, instead, gives general conditions under which analytical constructions that attempt to make our numerical results rigorous are guaranteed to fail.

%Section \S\ref{sec:num-opt-bnds} presents optimal solutions to this optimization problem for $2.0\times 10^5 \leq \Ra \leq 3.4 \times 10^5$, computed with the high-precision SDP solver \sdpagmp~\citep{yamashita2012latest} after finite-element discretation on a Chebyshev mesh with at least 6000 elements (see appendix \ref{sec:computational-details} for details). 
% Section \ref{} below reports bounds obtained computationally, while \S\ref{} discusses obstacles to proving similar results analytically.

\subsection{Numerically optimal bounds} \label{sec:num-opt-bnds}

Problem~\eqref{e:optimization-Fourier-positive} was discretised into an SDP and solved with the high-precision solver \sdpagmp~\citep{yamashita2012latest} for $2.0\times 10^5 \leq \Ra \leq 3.4 \times 10^5$. The MATLAB toolbox \sparsecolo ~\citep{Fujisawa2009} was used to exploit sparsity in the SDPs. At each Rayleigh number, we employed the finite-element discretisation approach described in Appendix \ref{sec:computational-details} on a Chebyshev mesh with at least 6000 piecewise-linear elements, increasing the resolution until the upper bounds changed by less than 1\%. Achieving this at $\Ra=3.4\times 10^5$ required approximately $12\,200$ elements. The numerical challenges associated with setting up the SDPs accurately in double-precision using \sparsecolo\ on even finer meshes prevented us from considering a wider the range of Rayleigh numbers.
% To make the computation tractable,  we replaced the constraints $\mathcal{S}_{\boldsymbol{k}}\{\hat{w}_{\boldsymbol{k}},\hat{T}_{\boldsymbol{k}}\} \geq 0$ with the stronger sufficient condition
%
% \begin{equation}\label{e:simplified-spectral-constraint}
%     \int_0^1 
%     \frac{2a}{R}\abs{\hat{w}_{\boldsymbol{k}}'}^{2} 
%     + b\vert \hat{T}_{\boldsymbol{k}}' \vert^{2}
%     - (a-\psi')\hat{w}_{\boldsymbol{k}}\hat{T}_{\boldsymbol{k}}
%     \,{\rm d}z \geq 0
%      \qquad \forall \hat{T}_{\boldsymbol{k}}, \hat{w}_{\boldsymbol{k}}:\text{~(\ref{e:Fourier-bc}\textit{a,b})}.
% \end{equation}
%
% This simplification strengthens the constraints of~\eqref{e:optimization-Fourier-positive}, leading to a quantitatively worse upper bounds on the heat transport, and prevents a direct quantitative comparison with the results presented in \S\ref{sec:results_no_p} (which were obtained using the exact expression for $\mathcal{S}_{\boldsymbol{k}}$). 
% Moreover, preliminary computations at much lower resolution and low \Ra\ revealed that, while this simplifications worsens the upper bounds on $\langle wT \rangle$ quantitatively, their qualitative behaviour is unchanged. % The observations made in \S\ref{sec:results_no_p} remain valid under this simplification, except that the numerically optimal bounds plotted in Figure~\ref{fig:optimal-bounds-general-T} are ``shifted'' to lower Rayleigh numbers.

\begin{figure}
    \centering
    \includegraphics[scale=1]{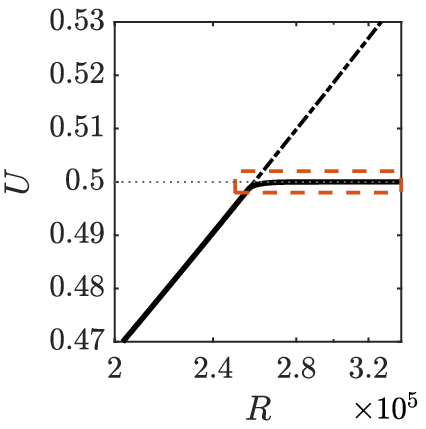}\hspace{10pt}
    \includegraphics[scale=1]{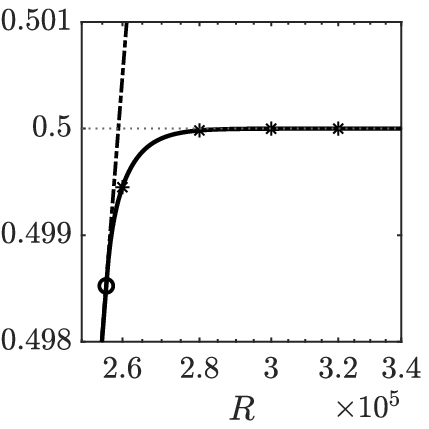}\hspace{5pt}
    \includegraphics[scale=1]{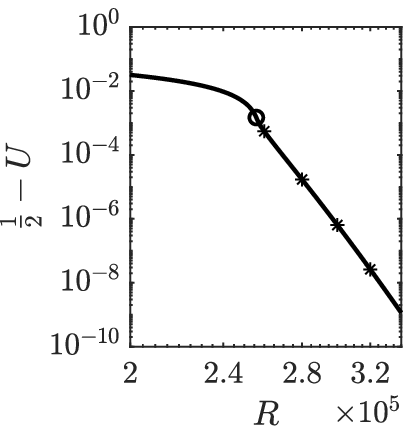}
    \begin{tikzpicture}[overlay]
        \node at (-5.2,4.1) {\textit{(a)}};    
        \node at (1.8,4.1) {\textit{(b)}};
        \node at (6.4,4.1) {\textit{(c)}};
    \end{tikzpicture}
    \label{fig:FE_results_1}
    \caption{
    \textit{(a)}~Optimal bounds $U$ on $\langle wT \rangle$ computed by solving \eqref{e:optimization-Fourier-positive}, which incorporates the minimum principle for temperature (\solidrule). Also plotted (\dotdashedrule) are the bounds computed in \S\ref{sec:results} without the minimum principle ($q(z)=-1$).
    \textit{(b)}~Detail of the region inside the red dashed box ({\color{matlabred}\dashedrule}) in panel~\textit{(a)}.
    \textit{(c)}~Difference between our optimal bounds and the uniform bound $\langle wT \rangle \leq 1/2$, shown in all panels as a dotted line ({\color{matlabgray}\dottedrule}). In \textit{(b)} and \textit{(c)}, a circle at $\Ra\approx 256\,269$ marks the point at which $q(z)$ begins to vary from $-1$, while stars ($\ast$) mark the Rayleigh numbers at which $\psi$ are plotted in Figure \ref{fig:psi_a_and_b}.
    }
    \label{fig:optimal-bounds-T0-positive}
\end{figure}

Figure~\ref{fig:optimal-bounds-T0-positive} compares numerical upper bounds on the heat transfer obtained with (solid line) and without (dot-dashed line) the minimum principle for the temperature, that is, by optimising $q(z)$ or by setting $q(z)=-1$ in \eqref{e:optimization-Fourier-positive}, respectively. The results for the latter case coincide with those described in \S\ref{sec:results_no_p} and shown in Figure~\ref{fig:optimal-bounds-general-T}.
% and cross the value of $\frac12$ before $\hat{T}_0'(0)$ becomes negative because we are considering the simplified constraint~\eqref{e:simplified-spectral-constraint}, rather that the original constraints on $\mathcal{S}_{\boldsymbol{k}}$. 

The choice $q(z)=-1$ is optimal for $\Ra < 256\,269$. For higher \Ra, the numerical upper bounds on $\langle wT \rangle$ with optimised $q(z)$ are strictly better than those with $q(z)=-1$ and, crucially, appear to approach the uniform bound $\frac12$ from below as $\Ra$ is raised. Note that although the deviation from $\frac12$ is small at the highest values of \Ra\ that could be handled, it is much larger than the tolerance ($10^{-25}$) used by the multiple-precision SDP solver \sdpagmp, giving us confidence that it is not a numerical artefact. Moreover, the deviation from $\frac12$ of the numerical bounds, shown in Figure \ref{fig:optimal-bounds-T0-positive}\textit{(c)}, appears to decay as a power law, suggesting that the optimal bound available within our bounding framework may have the functional form~\eqref{e:wt-conjecture}. Although the range of Rayleigh numbers spanned by our computations is too small to accurately predict the exponent and the prefactor, it is clear that the decay is much faster than predicted by the heuristic arguments in \S\ref{sec:heuristics}.

% The change in the behaviour of $a$ and $b$, at $2.6\times 10^5$, matches the requirements they need satisfy in the limit where we approach 1/2 of becoming zero. Before this $R$, $a$ and $b$ scale as $1/5 + \mathcal{O}(R^{-1/5})$ and \eqref{a_and_b_scaling}, after which the scaling from these results is the same with $R$. The range of $R$ for which numerical results could be computed is too small to extract exact asymptotic behaviour of the correction to the 1/2 bound.

\begin{figure}
    \centering
    \hspace*{-0.7cm}
    \includegraphics[scale=1]{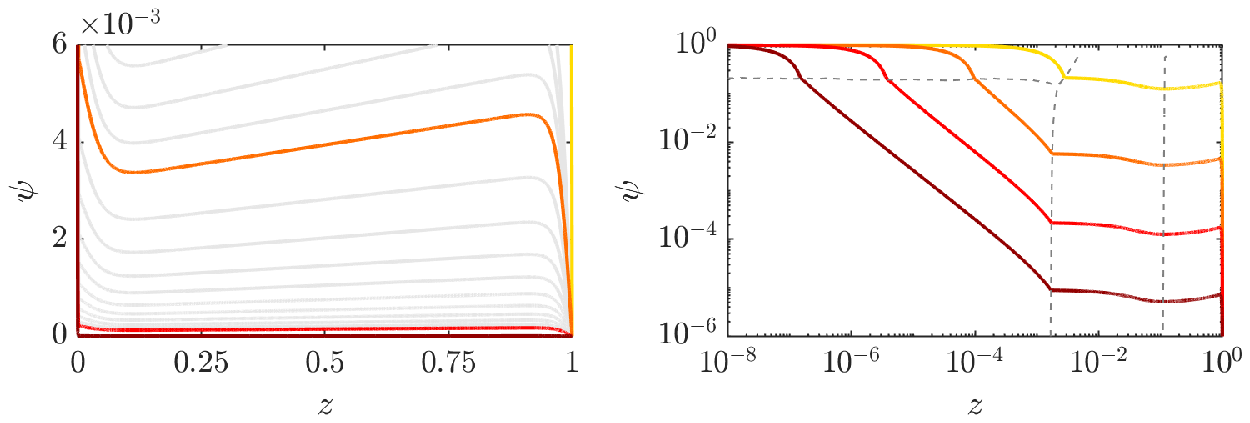}\\
    \includegraphics[scale=1]{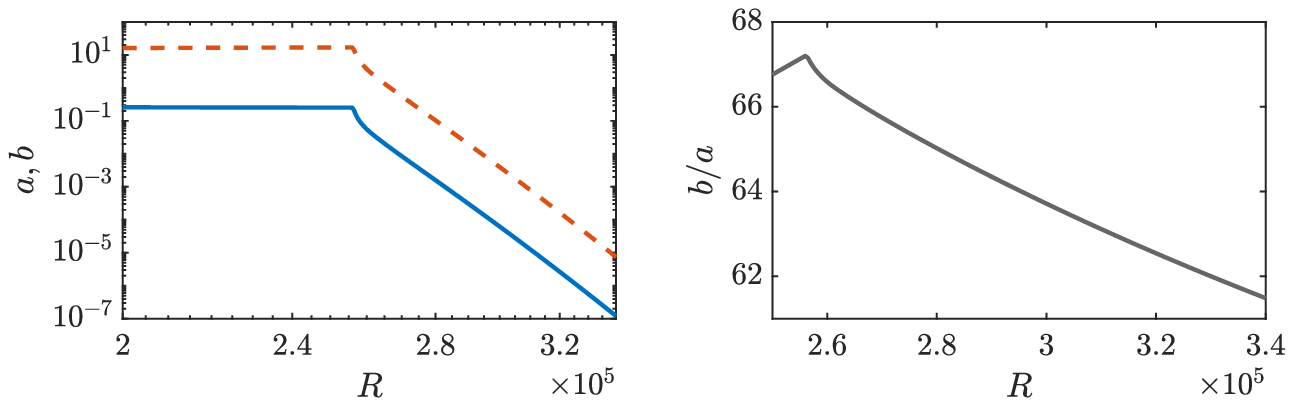}
    \label{fig:a_and_b_new_variations}
    \begin{tikzpicture}[overlay]
        \node at (-7.35,7.65) {\textit{(a)}};
        \node at (-0.7,7.65) {\textit{(b)}};
        \node at (-7.35,3.5) {\textit{(c)}};
        \node at (-0.7,3.5) {\textit{(d)}};
        \node at (-4.8,7.3) {$\delta_0$};
        \draw[black,thick] (-4.65,7.6) -- (-4.65,7.45) ;
        \node at (-2.3,5.25) {$\delta_1$};
        \draw[black,thick] (-2.1,5.45) -- (-2.1,5.3) ;
        \node at (-1.15,5.1) {$\delta_2$};
        \draw[black,thick] (-0.96,5.35) -- (-0.96,5.2) ;
    \end{tikzpicture}
    \caption{
    \textit{(a)} Optimal $\psi(z)$ in the entire domain, plotted for $0\leq\psi\leq0.006$ for visualization purposes. At all $\Ra$ values, $\psi(0)=1$. \textit{(b)} Logarithmic plot of $\psi(z)$, highlighting the behaviour near $z=0$. Coloured lines correspond to the $\Ra$ highlighted in Figure \ref{fig:optimal-bounds-T0-positive}, for $\Ra = 2.6\times 10^5$ ({\color{colorbar2}\solidrule}), $2.8\times 10^5$ ({\color{colorbar6}\solidrule}), $3.0\times 10^5$ ({\color{colorbar11}\solidrule}),  $3.2\times 10^5$ ({\color{colorbar16}\solidrule}). Dashed lines ({\color{grey}\dashedrule}) mark the boundaries of the sublayers $(0,\delta_0)$, $(\delta_0,\delta_1)$ and $(\delta_1,\delta_2)$, and the sublayer edges are labelled explicitly for $\Ra = 3.2\times10^5$.
    \textit{(c)}  Variation  with $\Ra$ of the optimal balance parameters $a$ ({\color{matlabblue}\solidrule}) and $b$ ({\color{matlabred}\dashedrule}) for~\eqref{e:optimization-Fourier-positive}. \textit{(d)} Ratio of balance parameters, $b/a$, when the Lagrange multiplier is active.}
    \label{fig:psi_a_and_b}
\end{figure}

Optimal profiles $\psi(z)$ for selected Rayleigh numbers are shown in Figures \ref{fig:psi_a_and_b}\textit{(a,b)} and differ significantly from the corresponding profiles in Figure \ref{fig:optimal-bounds-general-T}\textit{(b)} obtained when the minimum principle for the temperature is disregarded. When the minimum principle is enforced, $\psi$ appears to approach zero almost everywhere as \Ra\ is raised, but always satisfies $\psi(0)=1$. This leads to the formation of a very thin boundary layer near $z=0$, which at high \Ra\ consists of three distinct sublayers identified by two points, $z=\delta_0$ and $z=\delta_1$, at which $\psi$ is not differentiable. These are indicated by gray dashed lines in Figure \ref{fig:psi_a_and_b}\textit{(b)}. In the first sublayer, from $z=0$ to $z=\delta_1$, $\psi$ is observed to vary linearly. The second sublayer, $\delta_0 < z < \delta_1$, is observed only for $\Ra > 2.6\times10^5$ and we observe that $\psi\sim z^{-1}$ approximately. The third sublayer, from $z=\delta_1$ to the point $z=\delta_2$ at which $\psi$ attains a local minimum, does not have a simple functional form. 
%In \textit{(b)} for $\Ra = 2.6\times10^5$, the $\psi$ does not have the second sublayer, at $\Ra = 262\,000$ this second sublayer appears in the optimal profiles and then increases in size as $\Ra$ is increases. 
In the bulk, $\psi$ increases approximately linearly with slope very close to $a$, as was the case in \S\ref{sec:results}, and the condition $\psi(1)=0$, which emerges as a result of the optimization and is not imposed \textit{a priori}, is attained through a small boundary layer of width $\varepsilon$ near $z=1$. We choose the boundary of this layer as the point $z=1-\varepsilon$ at which $\psi$ has a local maximum.

Figure~\ref{fig:psi_a_and_b}\textit{(c)} shows the variation of the balance parameters with $\Ra$. Below $256\,269$, both coincide with the values plotted in Figure \ref{fig:optimal-bounds-general-T}\textit{(c)}. At higher $\Ra$, the minimum principle for the temperature becomes active and both balance parameters start to decay rapidly. It would be tempting to conjecture that $a\sim b \sim \Ra^{-p}$ for some power $p$ but, given the small variation of $b/a$ evident in Figure~\ref{fig:psi_a_and_b}\textit{(d)}, we cannot currently exclude that subtly different scaling exponents or higher-order corrections do not play an important role in obtaining an upper bound on $\langle wT \rangle$ that approaches $\frac12$ asymptotically from below.

\begin{figure}
    \centering
    \includegraphics[scale=0.98]{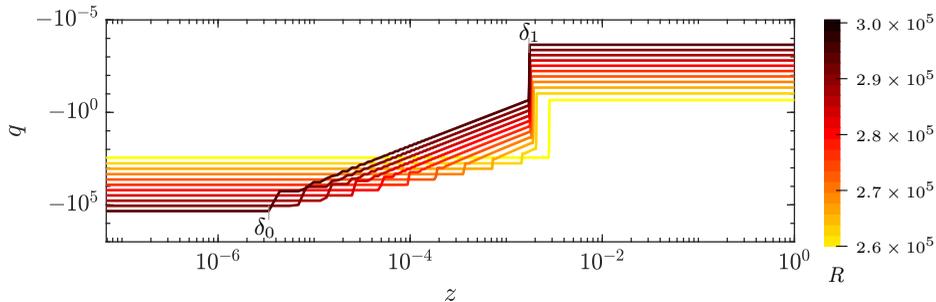}
    \input{Figures/colorbar-hot-2}
    \begin{tikzpicture}[overlay]
        \node at (-9.2,1) {$\delta_0$};
        \draw[grey] (-9.15,1.1) -- (-9.15,1.3);
        \node at (-5.7,3.57) {$\delta_1$};
        \draw[grey] (-5.71,3.3) -- (-5.71,3.5);
    \end{tikzpicture}
    \caption{
     Variation of $q(z)$ shown on a logarithmic scale to highlight the non-zero region of the Lagrange multiplier $q'(z)$, from $\Ra= 2.6 - 3.0\times 10^{5}$. The colorbar here applies for all $\psi$, $q$ and $\hat{T}_0$ presented in  \S\ref{sec:positive_temp}. The edges of the boundary sublayers $(0,\delta_0)$ and $(\delta_0,\delta_1)$ for the largest $\Ra$ value are explicitly labelled.}
    \label{fig:q'}
\end{figure}
\begin{figure}
    \centering
    \includegraphics[scale=1]{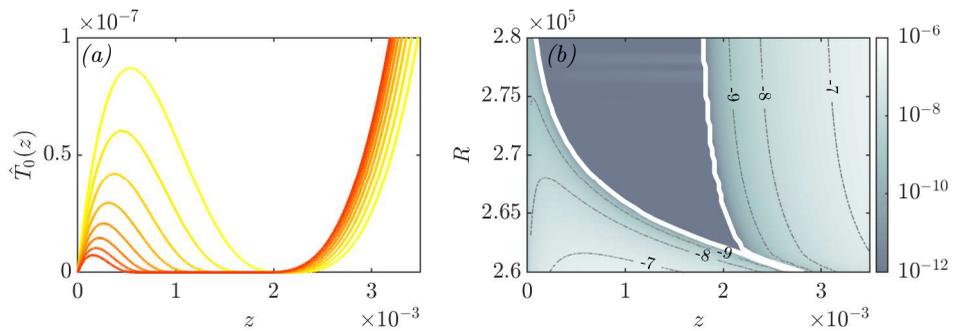}
    \begin{tikzpicture}[overlay]
        \node at (-4.75,4) {\textit{(a)}};
        \node at (1.4,4) {\textit{(b)}};
    \end{tikzpicture}
    \caption{Variation with $\Ra$ of the critical temperature profiles, $\hat{T}_0$, plotted for $0\leq z \leq 3.5\times 10^{-3}$. \textit{(a)} Plots of individual $\hat{T}_0$ from $\Ra= 2.62 - 2.69\times 10^{5}$. \textit{(b)} Contour plot of $\hat{T}_0$ vs $\Ra$. Solid white lines mark the boundary sublayer $\delta_0 \leq z \leq \delta_1$, within which $q'>0$ and $\hat{T}_0 = 0$. Values of $\hat{T}_0$ below $10^{-12}$ are assumed to be numerical zeros.}
    \label{fig:positive_T0_profiles}
\end{figure}

Figure~\ref{fig:q'} shows the structure of $q$, whose (distributional) derivative represents the Lagrange multiplier enforcing the minimum principle. The multiplier, therefore, is active in regions where $q$ is not constant. The choice $q(z)=-1$ is optimal for $\Ra  = 256\,269$. At higher Rayleigh numbers, the multiplier becomes active between $z=\delta_0$ and $z=\delta_1$, which is what causes the second boundary sublayer in $\psi$. In the immediate vicinity of the bottom boundary ($0\leq z \leq \delta_0$), $q$ is constant and it appears that $q(z) - \psi'(z) \approx b/2$ (cf. Figure~\ref{fig:FE_results_chi}\textit{(c)}). Indeed, inspection of the cost function in~\eqref{e:explicit-U-optimization-positive} suggests that $q(z) - \psi'(z) = b(1-z)/2$ should be optimal, but we could not identify the very small $z$-dependent correction in our numerical results. In the second sublayer ($\delta_0\leq z \leq \delta_1$), where the Lagrange multiplier is active, $q(z)\sim -z^{-2}$. Again, this is consistent with the minimization of the cost function in~\eqref{e:explicit-U-optimization-positive}, as one expects $q$ to cancel the very large contribution of $\psi'$ near the bottom boundary.

Further validation of our numerical results comes from inspection of the critical temperatures $\hat{T}_0(z)$, which can be recovered using~\eqref{e:optimal-T0-positive} and are shown in Figure \ref{fig:positive_T0_profiles}\textit{(a,b)} for selected values of $\Ra$. As expected, the critical temperatures are nonnegative for all $z$ and  vanish identically (up to small numerical tolerances) in the region $\delta_0 \leq z \leq \delta_1$, where $q$ is active. We note that, for a given Rayleigh number, this region is strictly larger than the range of $z$ values for which the critical temperatures in Figure \ref{fig:optimal-T0-general-profiles}\textit{(c)} are negative, indicating that the minimum principle alters the problem in a more subtle way than simply saturating the constraint $T\geq 0$. 
%are confirmed by observing how this region of $z$ corresponds to where $\hat{T}_{0}=0$, as highlighted in Figure \ref{fig:positive_T0_profiles}. Figures \ref{fig:positive_T0_profiles} and \ref{fig:q'} show that the size of this region increases with \Ra. Figure \ref{fig:positive_T0_profiles}\textit{(a)} shows the nonnegative $\hat{T}_{0}$ near $z=0$. It can be seen that with the minimum principle invoked, $\hat{T}_{0}\geq0$ and $\hat{T}_{0}'(0)>0$ in all cases, whereas the optimal $\hat{T}_0$ for general temperature fields, was negative over a larger of region $z$ for the same $\Ra$, Figure \ref{fig:optimal-T0-general-profiles}\textit{(c)}. Thus as expected for the modified problem the optimising $\hat{T}_0$ are also altered such that the lagrange multiplier does not just saturate the negative regions of the previous $\hat{T}_0$ to zero.  
\begin{figure}
    \centering
    \includegraphics[scale=1]{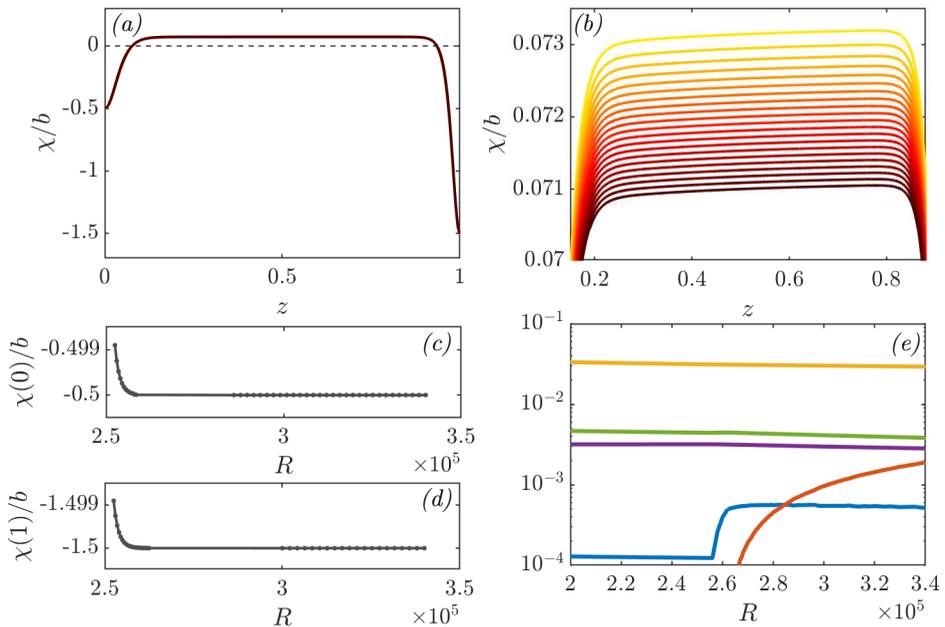}%
    \begin{tikzpicture}[overlay]
        \node at (-11.9,8.1) {\textit{(a)}};
        \node at (-5.8,8.1) {\textit{(b)}};
        \node at (-7.8,3.9) {\textit{(c)}};
        \node at (-7.8,1.85) {\textit{(d)}};
        \node at (-1.6,3.9) {\textit{(e)}};
    \end{tikzpicture}
    \caption{\textit{(a)} Profiles $\chi(z)/b$ for $ 2.6\times 10^5 \leq \Ra \leq 3.0 \times 10^{5}$. \textit{(b)} Detailed view of $\chi(z)/b$ in the bulk. \textit{(c,d)} Variation of $\chi(0)/b$ and $\chi(1)/b$ with $\Ra$.
    \textit{(e)} Contributions to the integral of $\psi/b$ in the regions $(0, \delta_0)$ ({\color{matlabblue}\solidrule}), $(\delta_0, \delta_1)$ ({\color{matlabred}\solidrule}), $(\delta_1, \delta_2)$ ({\color{matlabgreen}\solidrule}), $(\delta_2, 1-\varepsilon)$ ({\color{matlabyellow}\solidrule}) and $(1-\varepsilon, 1)$ ({\color{matlabpurple}\solidrule}), as a function of \Ra. %$\Ra = 2-3.4\times10^{5}$.
    }
    \label{fig:FE_results_chi}
\end{figure}

To analyse the results further we define the diagnostic function 
\begin{equation}
    \chi(z) = \psi'(z) - q(z)\, ,
\end{equation}
and rewrite \eqref{eq_for_U_q} as
\begin{equation}
    \langle wT \rangle \leq 
    \frac{1}{2} +  b\left[ \int^{1}_{0} \frac{1}{4}\left(z-\frac{1}{2} - \frac{\chi(z)}{b} \right)^{2}  - \frac{\psi(z)}{b} \,\textrm{d}z\right]   \, .
    \label{e:objective_chi2}
\end{equation}
%Plots of $\chi(z)/b$ and the final term in the integral are shown in Figure \ref{fig:FE_results_chi}.
Panels \textit{(a)}, \textit{(c)} \& \textit{(d)} in Figure \ref{fig:FE_results_chi} suggest that $\chi(0)/b \to -\tfrac{1}{2}$ and $\chi(1)/b \to -\tfrac{3}{2}$ as $\Ra\rightarrow\infty$. In fact, profiles $\chi(z)/b$ for different Rayleigh numbers collapse almost exactly throughout the layer, but subtle corrections are present; for instance, Figure~\ref{fig:FE_results_chi}\textit{(b)} demonstrates that, in the bulk, the mean value of $\chi(z)/b$ decreases with $\Ra$. It is also evident that $\chi(z)$ is not exactly constant throughout the bulk, but increases by approximately $10^{-4}$. Since $q(z)$ is constant in this region, we conclude that $\psi'(z)$ is not constant, but displays subtle and thus nontrivial variation. 

Figure~\ref{fig:FE_results_chi}\textit{(e)} illustrates the variation with $\Ra$ of contributions to the integral of $\psi(z) / b$ from regions $(0,\delta_0)$, $(\delta_0,\delta_1)$, $(\delta_1,\delta_2)$, $(\delta_2,1-\varepsilon)$ and $(1-\varepsilon,1)$. The largest contribution comes from the bulk ($\delta_2 \leq z \leq 1-\varepsilon$),  but it slowly decreases with $\Ra$. The same is true of the contribution of the top boundary layer ($1-\varepsilon \leq z \leq 1$) and the outermost boundary sublayer near $z=0$ ($\delta_1 \leq z \leq \delta_2$). Only in the first two boundary layers near $z=0$ does the value of the integral increase with $\Ra$, suggesting that the integral of $\psi(z)/b$ near the boundary layer may become the dominant term as $\Ra \to \infty$. %, when above the $\Ra$ at which the Lagrange multiplier enforces nonnegativity. Thus far we have observed the structure of $\chi/b$, $\psi$, the behaviour of $b$ and the integral of $\psi/b$. Hence in the expression for the bound, \eqref{e:objective_chi2}, the first term in the integral is clearly positive thus the $\psi/b $ is the only term which beneficial to achieving a smaller bound. 
While the range of Rayleigh numbers covered by our computations is too small to confirm or disprove this conjecture, it is certain that the integral of $\psi(z)/b$ must remain large enough to offset the positive term in~\eqref{e:objective_chi2} in order to obtain a bound on $\langle wT \rangle$ smaller than $1/2$. 

%While $\int_{\delta_1}^{1}\tfrac{\psi}{b}\,\textrm{d}z$ is negative, it decreases, very slowly with $\Ra$, leaving only $\int_{0}^{\delta_1}\tfrac{\psi}{b}\, \textrm{d}z$ to be increasing, possibly asymptotically. The numerically optimal bound shows us that $U$ is less than $\tfrac{1}{2}$, these additional observations indicate how this is possible, such that a bound of the form, $\langle w T \rangle \rightarrow \tfrac{1}{2}$ from below, appears, though this is either a strong algebraic scaling or exponential.

Finally, Figures~\ref{fig:q'} and~\ref{fig:FE_results_chi} suggest that, for sufficiently large $\Ra$, %a reasonable analytical ansatz for $q$ may be
% \begin{equation}
%      q(z) =  
%     \begin{dcases}
%           \psi'(z) + \tfrac{b}{2}, & 0 \leq z \leq \delta_0,\\
%           -q_{0}z^{-2}, & \delta_0 \leq z \leq \delta_1,\\
%           q_B, & \delta_1 \leq z \leq 1,
%     \end{dcases}
%     \label{rho_prime_ansatz}
% \end{equation}
% where $q_0$ and $q_B$ are constants. The numerical results further indicate,
% \begin{equation}
%      q_B =  
%     \begin{dcases}
%           -1, & 1 \leq R < R_{c},\\
%           \psi'(1) + \tfrac{3b}{2}, & R_{c} \leq R ,
%     \end{dcases}
%     \label{gamma_ansatz}
% \end{equation}
%with $R_{c}$ the critical $\Ra$ below which $q = q_B$ in the entire domain, in which case $\delta_1=0$. 
\begin{equation}
     q(z) =  
    \begin{dcases}
          \psi'(z) + \tfrac{b}{2}, & 0 \leq z \leq \delta_0,\\
          -q_{0}z^{-2}, & \delta_0 \leq z \leq \delta_1,\\
          \psi'(1) + \tfrac{3b}{2}, & \delta_1 \leq z \leq 1,
    \end{dcases}
    \label{rho_prime_ansatz}
\end{equation}
for some positive constant $q_0$, and that $a,\psi,q_0 \sim b$. 
The next section investigates whether such ansatz can lead to upper bounds on $\langle wT \rangle$ that are strictly smaller than $1/2$ at all Rayleigh numbers.
%Nevertheless, it is possible to show that such ansatze cannot lead to an upper bound on $\langle wT \rangle$ smaller than $1/2$ if used in conjunction with standard estimates for the spectral constraint in~\eqref{e:optimization-Fourier-positive} and with a choice of $\psi$ that increases linearly in the bulk of the layer.
%It's note worthy that $q_B \rightarrow 0 $ as $\Ra\rightarrow\infty$ and $q$ is negative for all $z$. These results may initially indicate $a$, $b$, $q_B$ (and hence $q(z)$ and $\psi(z)$) to scale, up to leading order, with the same exponent, for $\Ra > R_c$. However, the strong coupling of the problem between the objective function, spectral constraint and optimisation parameters does not yield the desired bound, when the problem is simply rewritten in terms of $b$. Indicating a more subtle approach is required in the analysis. Nevertheless we can make progress by showing what cannot work in an analytical proof of a bound tending to $\tfrac{1}{2}$ from below.

%%%%%%%%%%%%%%%%%%%%%%%%%%%%%%%
\subsection{Towards an analytical bound}
\label{sec:p_discussion}

The numerical evidence presented in  \S\ref{sec:num-opt-bnds} suggests that the upper bound on $\langle wT \rangle$ obtained from the optimization problem \eqref{e:explicit-U-optimization-positive} approaches $\tfrac{1}{2}$ asymptotically from below as $R \rightarrow \infty$. To confirm this observation with a proof requires, for every $R \geq 1$, construction of feasible decision variables $a,b,\psi(z),q(z)$ whose corresponding cost is strictly less than $\tfrac{1}{2}$. This section discusses the challenges presented by this goal. Specifically, we show that no construction is possible if one tries to mimic key properties of the numerically optimal decision variables presented in \S\ref{sec:num-opt-bnds} and, at the same time, enforces the spectral constraint using estimates typically used in successful applications of the background method.

To aid the discussion, the following definition introduces  three subsets $\mathcal{A}, \mathcal{B}$ and $\mathcal{C}$ of decision variables that capture some of the properties observed from our numerical study. %in \S\ref{sec:num-opt-bnds}. 

\begin{definition}
Let $0 < \delta < 1-\varepsilon <1$ and $\gamma>0$ and $R \geq 1$. Decision variables $(a,b, \psi(z),q(z))$ are said to belong to: 

\begin{itemize}
  \item[1.] The set $\mathcal{A}\{\delta,\varepsilon\}$ if  the following conditions hold:
\begin{itemize}
\item[(a)] The balance parameters satisfy $0 < a \leq b$;
\item[(b)] $\psi \in C^1[0,1]$ with boundary conditions $0 \leq \psi(0) \leq 1$ and $\psi(1)=0$;
\item[(c)] The derivative of $\psi$ satisfies
$\psi'|_{{[0,\delta]}} \leq 0$, $\psi'|_{{[\delta,1-\varepsilon]}} \geq 0$, and $\|\psi'\|_{L^\infty(1-\varepsilon,1)} \leq 2b$.
% \[
% \psi'|_{{[0,\delta]}} \leq 0, \qquad \psi'|_{{[\delta,1-\varepsilon]}} \geq 0, \quad \text{and} \quad \|\psi'\|_{L^\infty(1-\varepsilon,1)} \leq 2b.
% \]
\end{itemize}
\item[2.] The set $\mathcal{B}\{\gamma\}$ if both $\psi'(z)$ and $q(z)$ are %, possibly different, constants on the interval
constant on the interval
%
%\[
$(\tfrac{1}{2}-\gamma,\tfrac{1}{2}+\gamma).$
%\]
\item[3.] The set  $\mathcal{C}\{\delta,\varepsilon,R\}$ if %\eqref{e:spectral-constraint-sufficient-condition-general_simp} is satisfied. That is
\[
\delta^2 \|a-\psi'\|_{L^\infty (0,\delta)} + \varepsilon^2 \|a-\psi'\|_{L^\infty(1-\varepsilon,1)} \leq 8\sqrt{ \frac{2ab}{R} }. 
\]
\end{itemize}  
\end{definition}

The set $\mathcal{A}\{\delta,\varepsilon\}$ contains profiles $\psi$ which possess an initial (and potentially severe) boundary layer in an interval $[0,\delta]$, then increase in a bulk region $[\delta,1-\varepsilon]$, before approaching $\psi(1)=0$ in an upper boundary layer contained in the interval $[1-\varepsilon,1]$ and in which $\psi'$ is controlled by the balance parameter $b$, as seen in previous results. Optimal decision variables $(a,b, \psi(z),q(z))$ obtained in \S\ref{sec:num-opt-bnds}  appear, with compelling evidence, to belong to a set of the form $\mathcal{A}\{\delta,\varepsilon\}$ with the exception of the differentiability condition $\psi \in C^1[0,1]$. Indeed, the optimal profiles $\psi$ appear to be piecewise differentiable, losing differentiability at two points corresponding to the boundaries of non-constant behaviour  of the multiplier $q(z)$ observed in figure \ref{fig:q'}. However, since no higher derivatives of $\psi$ appear in the optimization problem \eqref{e:explicit-U-optimization-positive}, adding the constraint $(a,b,\psi(z),q(z)) \in \mathcal{A}$ to \eqref{e:explicit-U-optimization-positive} will not change its optimal cost.

The set $\mathcal{B}\{\gamma\}$ contains decision variables for which $\psi'$ and $q$ are constant in some interval centred at $\tfrac{1}{2}$. Figures \ref{fig:q'} and \ref{fig:FE_results_chi} reveal that this is not the case for the numerically optimal $\psi$, so the use of $\mathcal{B}\{\gamma\}$ corresponds to a proof which ignores subtle variations from a purely linear profile away from the boundaries. Without further assumptions (say, restriction to fluids with infinite Prandtl number), it is not clear how such variations can be exploited in analytical constructions. %$\chi'\neq \textrm{constant}$ behaviour in Figure \ref{fig:FE_results_chi}\textit{(c)}.

The set $\mathcal{C}\{\delta,\varepsilon,R\}$ relates to a choice of profiles $\psi$ and balance parameters $a,b$ for which the constraint $\mathcal{S}_{\boldsymbol{k}} \geq 0$ can be proven to hold for a given $R$. In particular, if $a,b,\psi(z)$ satisfy \eqref{e:spectral-constraint-sufficient-condition-general_simp} and it is the case that  $\psi'|_{{[\delta,1-\varepsilon]}}=a$, then the argument in Appendix \ref{appendix:spectral estimates} implies that $\tilde{\mathcal{S}} \geq 0$. Specifically, $\mathcal{C}\{\delta,\varepsilon,R\}$ provides sufficient control of the severity of the boundary layers of $\psi$ for the spectral constraint to be provably satisfied. While crude, estimates of this form in conjunction with constant $\psi'(z)$ in a bulk region $\delta \leq z \leq 1-\varepsilon$ are employed for almost all analytical constructions of background fields. 

We now return to the original question of attempting to upper-bound $\langle wT \rangle$ via an analytical construction of feasible decision variables $(a,b,\psi(z),q(z))$ for \eqref{e:explicit-U-optimization-positive}. It is not unreasonable, based upon the above evidence, to propose $R$-dependent balance parameters $a=a_R, b=b_R$, boundary layer widths $\delta = \delta_R, \varepsilon = \varepsilon_R$ and profiles $\psi_R,q_R$ which satisfy 
\[
(a_R,b_R,\psi_R(z),q_R(z)) \in \mathcal{A}\{ \delta_R,\varepsilon_R\} \cap \mathcal{B}\{\gamma\} \cap \mathcal{C}\{ \delta_R,\varepsilon_R,R\}, \qquad R \geq 1,
\]
for some $\gamma>0$. The following result shows that, for such a construction, there is a hard lower bound on the  optimal cost achievable using \eqref{e:explicit-U-optimization-positive}.

\begin{proposition} \label{prop:counter}
Let $0 < \delta < 1-\varepsilon <1$ and $\gamma>0, R \geq 1$. Suppose that
\begin{equation} \label{eq:abc}
(a,b,\psi(z),q(z)) \in \mathcal{A}\{\delta,\varepsilon\} \cap \mathcal{B}\{\gamma\} \cap \mathcal{C}\{\delta,\varepsilon,R\}.
\end{equation}
Then
\[
\frac{1}{4b}\left\| bz-\frac{b}{2} - \psi'(z) + q(z) \right\|_2^2 - \int_{0}^{1}\psi \, \textrm{d}z \geq \frac{b}{6} \left( \gamma^3 - 24 \cdot \frac{1 + 2\sqrt{2}}{R^{\frac14}} \right). 
\]
\end{proposition}

From the proof in Appendix~\ref{appendix:proof_of_eqs}, the consequence of Proposition \ref{prop:counter} is the following. If one constructs feasible decision variables for the optimization problem \eqref{e:explicit-U-optimization-positive} which satisfy \eqref{eq:abc}, then the best achievable bound $U$ for which  $\langle wT \rangle \leq U$ must satisfy 
\[
U \geq \frac12 + \frac{b}{6} \left( \gamma^3 - 24 \cdot \frac{1+2\sqrt{2}}{R^{\frac14}} \right).
\]
Hence, using such a construction with the assumption that $\psi'$ and $q$ are constant %(uniformly in $R$) 
in a bulk region $(\frac12-\gamma,\frac12+\gamma)$, it is not possible to prove that $\langle wT \rangle < \frac12$ at arbitrarily high Rayleigh number. 

Consequently, one must ask what conditions should be dropped if a rigorous bound $\langle wT \rangle < \frac12, R \geq 1$, is to be found. The numerical evidence presented in  \S\ref{sec:num-opt-bnds} suggests that the optimal decision variables do belong to $\mathcal{A}$. Consequently, either $\mathcal{B}$ or $\mathcal{C}$ must be dropped. Figure \ref{fig:q'} indicates that optimal multipliers $q(z)$ are constant outside a lower boundary layer. Hence, dropping $\mathcal{B}$ corresponds to choosing a profile with non-constant $\psi'(z)$ in the bulk; the cost function of \eqref{e:explicit-U-optimization-positive} and Figure \ref{fig:FE_results_chi}\textit{(b)} suggests that a quadratic ansatz for $\psi(z)$ may be beneficial. Dropping $\mathcal{C}$ corresponds to requiring more sophisticated analysis of the spectral constraint. Using these insights will be the focus of future research.

\section{Conclusions}
\label{sec:conclusion}

We obtained upper bounds $U$ on the vertical heat transport $\langle wT
\rangle$ in IH convection that improve on the best existing bound $\langle wT \rangle\leq 1/2$ \citep{goluskin2012convection}. 
% Crucially, we demonstrate that unless non-physical negativity of the optimal temperature field is prevented, the existing uniform bound can only be improved up to a finite value of the Rayleigh number \Ra.
Crucially, we demonstrated that, within our bounding framework, the existing uniform bound can be improved only up to a finite value of the Rayleigh number \Ra\ unless non-physical negativity of the optimal temperature field is prevented.
Specifically, we construct\red{ed} an analytical proof that $\langle w T \rangle \leq 2^{-\frac{21}{5}} R^{1/5}$, which improves on the uniform bound for $\Ra < 2^{16}=65\,536$ (2.43 times larger than the energy stability threshold). However, our numerical results suggest that the best available bound tends to $1/2$ asymptotically from below as $\Ra\rightarrow\infty$ when a Lagrange multiplier is introduced to enforce non-negativity of the optimal temperature field.

% Challenges and scope for further work (in addition to the above).

A numerical challenge encountered in this work was in the implementation of the optimisation problem~\eqref{e:optimization-Fourier-positive}. 
%The fixed boundary values of $\psi$ and the structure of the multiplier $q$, specifically their large variation near $z=0$, require high resolution and consequently a fine mesh. 
The sharp boundary layers in $\psi$ and $q$ near $z=0$ require extremely high resolution and, consequently, a fine mesh. This poses challenges not only in terms of computational cost, but also in terms of numerical accuracy. Using a simplification of the spectral constraint we obtained results for a limited range of $\Ra$, but in order to shed light on the problems highlighted above, possible improvements to the algorithm or problem's formulation remain a point of interest and opportunity for further research.

% Limitations of analytical and numerical bounds and relation to heuristic scaling arguments.

As explained in \S\ref{sec:p_discussion}, the prospect of obtaining an analytical proof that $U\rightarrow 1/2$ strictly from below as $R\rightarrow \infty$ requires the use of more sophisticated estimates to
satisfy the spectral constraint than those that are typically used. Such improvements are necessary, rather than sufficient, conditions because there is a possibility that $U=1/2$ for values of \Ra\ that lie outside the range that we have studied numerically. In this regard, the numerical results, which correspond to the best available bound for quadratic auxiliary functions, suggest either exponential or extremely strong algebraic scaling of $U$ towards $1/2$ in $R$ that is far from any of the heuristic scaling possibilities discussed in \S\ref{sec:heuristics}. More sophisticated treatments of the spectral constraint might yield a proof that $U\rightarrow 1/2$ strictly from below, but they are unlikely to prove bounds that correspond to heuristic scaling arguments. Either IH convection defies all of the heuristics invoked in \S\ref{sec:heuristics} or the quadratic auxiliary function framework is unable to access a crucial aspect of the system's properties. The latter could be explored further by investigating a larger class of auxiliary functions \citep{goluskin2019ks,Fantuzzi2016siads,chernyshenko2014polynomial}, but at the possible expense of analytical or even numerical tractability.

\vspace{2ex}\noindent
\textbf{Acknowledgements}  We are grateful to D. Goluskin for enlighting us on many aspects of IH convection and for sharing his DNS data. This work also benefited from conversations with C. Nobili, J. Whitehead, C. R. Doering, I. Tobasco and G. O. Hughes. AA acknowledges funding by the EPSRC Centre for Doctoral Training in Fluid Dynamics across Scales (award number EP/L016230/1). GF gratefully acknowledges the support of an Imperial College Research Fellowship and the hospitality of the 2018 GFD program at Woods Hole Oceanographic Institution, where this work was begun. 

\vspace{2ex}\noindent
\textbf{Conflict of interests} The authors report no conflict of interests. 

\appendix
\section{Minimum principle in IH convection}
\label{sec:proof_T_positive}

A minimum principle for IH convection can be proved by adapting arguments for Rayleigh--B\'enard convection \citep[][Lemma 2.1]{foias1987attractors}.  Let
\begin{equation}
    T_{-}(\boldsymbol{x},t) :=  \max\{-T(\boldsymbol{x},t) , 0 \}
\end{equation}
denote the negative part of $T$ and observe that $T_-$ is nonnegative on the fluid's domain $\Omega$.
Multiplying the advection-diffusion equation~\eqref{nondim_energy} by $T_-$ and integrating by parts over the domain using the boundary conditions and incompressibility yields
\begin{equation}
    \frac{1}{2}\frac{\textrm{d}}{\textrm{d}t}\lVert  T_{-} \rVert_{2}^{2}  
    = -\lVert \nabla T_{-} \rVert_{2}^{2} 
    - \int_{\Omega} T_{-}\, \dVol.
\end{equation}
Upon observing that the last integral on the right-hand side is positive and that $T_-$ vanishes at the top and bottom boundaries, so $\|\nabla T_-\|_2^2 \geq \mu \|T_-\|_2^2$ for some constant $\mu>0$ by the Poincar\'e inequality, we can estimate
\begin{equation}
     \frac{1}{2}\frac{\textrm{d}}{\textrm{d}t}\lVert  T_{-} \rVert_{2}^{2}  \leq -\mu \lVert  T_{-} \rVert_{2}^{2}.
\end{equation}
Gronwall's lemma then yields
\begin{equation}
    \lVert  T_{-}(t) \rVert_{2}^{2}  \leq  \lVert  T_{-}(0) \rVert_{2}^{2}\, \textrm{e}^{-\mu t},
\end{equation}
so $T_-$ tends to zero in $L^2(\Omega)$ at least exponentially quickly. This implies that $T(\boldsymbol{x},t)\geq 0$ almost everywhere on the global attractor and that $T(\boldsymbol{x},t)$ is nonnegative at all times if it is so at $t=0$. %This argument formalizes observations already made and exploited by \cite{goluskin2012convection}.

%%%%%%%%%%%%%%%%%%%%%%%%%%%%%%%%%%

\section{Comparing the full and simplified spectral constraints}
\label{sec:comparison_Sk}

This appendix provides further computational evidence that replacing the spectral constraints $\mathcal{S}_{\boldsymbol{k}}\geq 0$ with the single, stronger constraint $\tilde{\mathcal{S}}\geq 0$ does not affect the qualitative behaviour of the optimal bounds on $\langle wT \rangle$. The simplified optimization problem~\eqref{e:optimization-Fourier-simplified} was solved using the finite-element expansion approach described in appendix~\eqref{sec:computational-details}. For simplicity, instead, the wavenumber-dependent problem~\eqref{e:optimization-Fourier} was solved using the general-purpose toolbox \quinopt~\citep{Fantuzzi2017tac}, which implements Legendre series expansions. The critical wavenumbers were determined with the help of the following result.
\begin{figure}
    \centering
    \hspace*{-0.5cm}
    \includegraphics[scale=0.92]{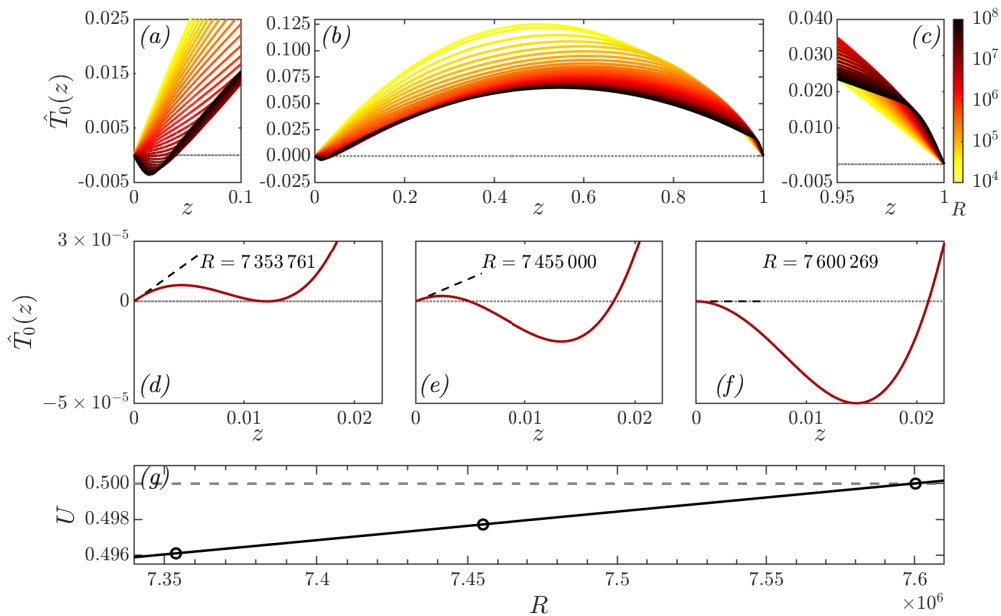}
    \begin{tikzpicture}[overlay]
        \node at (-5,8.6) {\textit{(a)}};
        \node at (-2.6,8.6) {\textit{(b)}};
        \node at (5.2,8.6) {\textit{(c)}};
        \node at (-5,4) {\textit{(d)}};
        \node at (-1.3,4) {\textit{(e)}};
        \node at (2.6,4) {\textit{(f)}};
        \node at (-5,2.8) {\textit{(g)}};
    \end{tikzpicture}
    \vspace*{-1ex}
    \caption{
    \textit{Top}: Critical temperature $\hat{T}_0(z)$ recovered using~\eqref{e:optimal-T0-general} when the spectral constraint is given by $\mathcal{S}_{\boldsymbol{k}}>0$. Colors indicate the Rayleigh number. Panels \textit{(a)} and \textit{(c)} show details of the boundary layers.
    \textit{Middle}: Detailed view of $\hat{T}_0$ for $\Ra = 7\,353\,761$, $7\,455\,000$, and $7\,600\,269$. Dashed lines (\dashedrule) are tangent to $\hat{T}_0$ at $z=0$. In~\textit{(d)}, $\hat{T}_0$ is nonnegative and has minimum of zero inside the layer. In~\textit{(e)}, $\hat{T}_0$ is initially positive but has a negative minimum. In~\textit{(f)}, $\hat{T}_0'(0)=0$ and there is no positive initial layer.
    \textit{Bottom}: Upper bounds $U$ on $\langle wT \rangle$. Circles mark the values of $\Ra$ considered in \textit{(d--f)} and $U=1/2$ at $\Ra = 7\,600\,269$.
    }
    \label{fig:optimal-T0-general-profiles2}
\end{figure}

%, we first observe that given any value of \Ra, balance parameters $a$ and $b$, and profile $\psi(z)$, only a finite set of wavevectors $\boldsymbol{k}$ need to be considered explicitly.
%
\begin{lemma}
Fix \Ra, $a$, $b$ and $\psi(z)$ satisfying $\psi(0)=1$ and $\psi(1)=0$. The inequality $\mathcal{S}_{\boldsymbol{k}}\{\hat{w}_{\boldsymbol{k}},\hat{T}_{\boldsymbol{k}}\} \geq 0$ holds for all $\hat{w}_{\boldsymbol{k}}$ and $\hat{T}_{\boldsymbol{k}}$ satisfying conditions (\ref{e:Fourier-bc}\textit{a,b}) if
\begin{equation}\label{e:k-bound-condition}
    k^{4} \geq \frac{\Ra}{4ab}\, \left\| a-\psi' \right\|_{\infty}^{2}.
\end{equation}
\end{lemma}
\begin{proof}
The H\"older and Cauchy--Schwarz inequalities yield
% \begin{equation}
%     \int_{0}^{1} (a-\psi')w_{k}T_{k} \, \textrm{d}z \leq \lVert (a-\psi')w_{k} \rVert_{2} \lVert T_{k}^{2} \rVert_{2}\, .
% \end{equation}
% If we assume initially a general smooth function $\psi$ we can take its measure in $L^{\infty}$ and further approximate the above as 
% \begin{equation}
%     \lVert (a-\psi')w_{k} \rVert_{2} \leq \lVert a-\psi' \rVert_{\infty} \lVert w_{k} \rVert_{2}
% \end{equation}
\begin{equation}
    \abs{\int_{0}^{1} (a-\psi')\hat{w}_{\boldsymbol{k}} \hat{T}_{\boldsymbol{k}} \, \textrm{d}z}
     \leq \left\| a-\psi'\right\|_\infty \left\| \hat{w}_{\boldsymbol{k}} \right\|_{2} \big\| \hat{T}_{\boldsymbol{k}} \big\|_2.
\end{equation}
Consequently,
\begin{equation}
    \mathcal{S}_{\boldsymbol{k}} 
    \geq 
    \frac{a k^{2}}{\Ra} \left\| \hat{w}_{\boldsymbol{k}} \right\|_{2}^{2} 
    - \left\| a-\psi'\right\|_\infty \left\| \hat{w}_{\boldsymbol{k}} \right\|_{2} \big\| \hat{T}_{\boldsymbol{k}} \big\|_2 
    + b k^{2} \big\| \hat{T}_{\boldsymbol{k}} \big\|_2^2.
\end{equation}
The right-hand side is a homogeneous quadratic form in $\| \hat{w}_{\boldsymbol{k}} \|_2$ and $\| \hat{T}_{\boldsymbol{k}} \|_2$ and is nonnegative for any choice of $\hat{w}_{\boldsymbol{k}}$ and $\hat{T}_{\boldsymbol{k}}$ if and only if~\eqref{e:k-bound-condition} holds.
\end{proof}

This result guarantees that, when implementing~\eqref{e:optimization-Fourier} numerically, it suffices to consider wavenumbers with
\begin{equation}\label{e:k-critical}
    k < \left(\frac{\Ra}{4ab}\right)^{\frac14} \|a-\psi'\|_{\infty}^{\frac12}.
\end{equation}
While the right-hand side of this inequality is unknown a priori, as it depends on the optimisation variables $a$, $b$ and $\psi$, in practice one can simply solve~\eqref{e:optimization-Fourier}  using all wavevectors with $k$ smaller than an arbitrarily chosen value. Then, one checks whether $\mathcal{S}_{\boldsymbol{k}}$ is indeed nonnegative for all $k$ satisfying~\eqref{e:k-critical}, and repeats the computation with a larger set of wavevectors if these checks fail. 
\begin{figure}
    \centering
    \hspace*{-0.8cm}
    \includegraphics[scale=0.95]{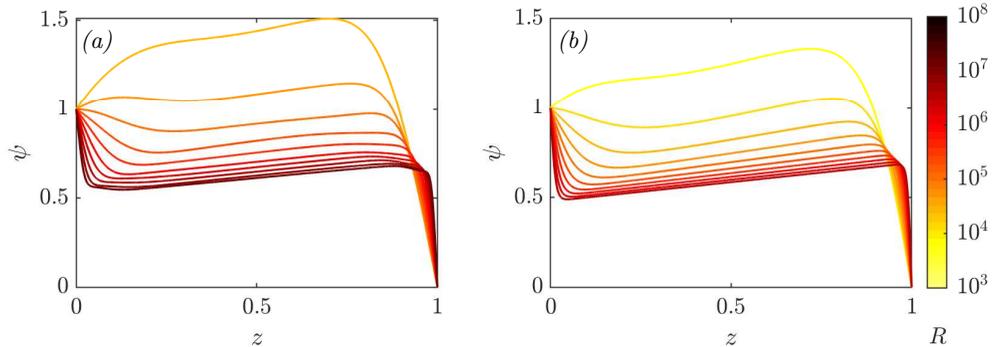}
    \begin{tikzpicture}[overlay]
        \node at (-5.35,4.45) {\textit{(a)}};
        \node at (0.95,4.45) {\textit{(b)}};
    \end{tikzpicture}
    \caption{(\textit{a}) Optimal $\psi(z)$ for selected Rayleigh numbers (see colorbar) obtained by solving optimization problem~\eqref{e:optimization-Fourier}, which imposes the exact spectral constraints $\mathcal{S}_{\boldsymbol{k}}\geq 0$. (\textit{b}) Optimal $\psi(z)$ at the same Rayleigh numbers but obtained by solving optimization problem~\eqref{e:optimization-Fourier-simplified}, which imposes the simplified spectral constraint $\tilde{\mathcal{S}}\geq 0$. }
    \label{fig:comparison_psi_diff_Sk}
\end{figure}
%
%For the computations using the full $\mathcal{S}_{\boldsymbol{k}}$, wavenumbers up to the critical were not considered for the largest $\Ra$ considered, $1\times10^{7}$. The bound surpasses 1/2 before this $\Ra$. 
%The advantage of $\Tilde{\mathcal{S}}$ can be seen here as we no longer need to check for the viable $k$ numerically, since the constraint is independent of $k$. 
%
Upper bounds obtained by solving the full problem~\eqref{e:optimization-Fourier} and the simplified problem~\eqref{e:optimization-Fourier-simplified} are shown in Figure~\ref{fig:optimal-bounds-general-T}. Here we demonstrate the equivalence of the results for both spectral constraints qualitatively. As expected, using the simplified spectral constraint yields worse bounds at a fixed Rayleigh number. While the full spectral constraint yields bounds that are zero up for all $\Ra$ up to the energy stability threshold, which depends on the choice of horizontal periods $L_x$ and $L_y$, the simplified functional $\Tilde{\mathcal{S}}$ is insensitive to these values and gives a conservative estimate for the nonlinear stability threshold. Nevertheless, both sets of result display the same qualitative increase as $\Ra$ is raised. This is further demonstrated in Figures \ref{fig:optimal-T0-general-profiles2}, which can be compared to Figure~\ref{fig:optimal-T0-general-profiles}

In particular, the upper bound reaches $1/2$ exactly when the critical temperature $\hat{T}_0$, which minimizes the functional $\mathcal{S}_0$, has zero gradient at $z=0$, as can be observed in Figures \ref{fig:optimal-T0-general-profiles2}\textit{(f)} \& \textit{(g)}. Shown in panels \textit{(a)-(c)} are the critical temperature fields, $\hat{T}_0$, for $10^{4}\leq \Ra \leq 10^{8}$. In the middle row of the figure, going from left to right, observe first that for $\Ra > 7\,353\,761$ the critical temperature is negative in the domain. Then panel \textit{(e)} shows a $\Ra$ at which  $\hat{T}_0$ is positive in a small region very close to the wall but clearly violates the minimum principle further away. In \textit{(f)}, where $\Ra = 7\,600\,269$, for our choice of $L_x=L_y=2$, the numerically optimal bound passes the uniform bound, at which point we have $\hat{T}_0'(0)=0$. These results for the $\boldsymbol{k}$-dependent spectral constraint qualitatively match the results for the simplified spectral constraint~\eqref{e:simplified-spectral-constraint} presented Figure \ref{fig:optimal-T0-general-profiles}.

As evidenced by Figure~\ref{fig:comparison_psi_diff_Sk}, the optimal $\psi$ obtained with the full and simplified spectral constraints also display similar features. The only notable difference is the non simple behaviour near the boundary layer of the optimal profiles $\psi$ for the full problem~\eqref{e:optimization-Fourier}, which arise after bifurcations in the critical wavenumbers as $\Ra$ is raised. In panel (\textit{b}), instead, $\psi$ exhibits a simple structure in both of the boundary layers.%three sub-domains $(0,\delta), (\delta,1-\varepsilon)$ and $(1-\varepsilon,1)$. 

These observations confirm that strengthening the spectral constraints using the wavevector-independent functional $\tilde{\mathcal{S}}$ only affects our computational results quantitatively, but preserves the overall qualitative behaviour.
%
%%%%%%%%%%%%%%%%%%%%%%%%%%%%

\section{Justification of \texorpdfstring{\eqref{e:positive-linear-functional}}{(\ref{e:positive-linear-functional})}}
\label{app:riesz-representation}

To justify the choice of bounded linear functional in~\eqref{e:positive-linear-functional}, we start with a technical lemma. In this appendix, $\mathcal{T}$ is the space of square-integrable temperature fields in the Sobolev space $H^1(\Omega)$ that are horizontally periodic and vanish at $z=0$ and $1$. We equip $\mathcal{T}$ with the inner product $(T_1, T_2) = \fint_\Omega \nabla T_1 \cdot \nabla T_2 \,\dVol$.

%%%%%%%%%%%%
\begin{lemma}\label{lemma:positive-functionals}
Let $q:(0,1)\to \mathbb{R}$ be a square-integrable function. The bounded linear functional $\mathcal{L}_q:\mathcal{T} \to \mathbb{R}$ given by
%
%\begin{equation*}
$\mathcal{L}_q(T) = -\fint_\Omega q(z) \frac{\partial T}{\partial z} \, \dVol$
%\end{equation*}
%
is positive if and only if $q(z)$ is nondecreasing.
\end{lemma}
%%%%%%%%%%%%
\begin{proof}
First, we prove that $q$ is nondecreasing if $\mathcal{L}_q$ is a positive functional. Fix any $z_1,z_2 \in (0,1)$ with $z_1<z_2$ and choose $\varepsilon>0$ small enough that $0<z_1-\varepsilon$ and $z_2+\varepsilon< 1$. Consider a temperature profile $T_\varepsilon(\boldsymbol{x}) = T_\varepsilon(z)$ that varies only in $z$ and satisfies
\begin{equation*}
    \partial_z T_\varepsilon := \begin{cases}
          \varepsilon^{-1} & z_1-\varepsilon \leq z\leq z_1+\varepsilon,\\
          -\varepsilon^{-1} & z_2-\varepsilon \leq z\leq z_2+\varepsilon,\\
          0 & \text{otherwise}.
    \end{cases}
\end{equation*}
Clearly,  $T_\varepsilon \in \mathcal{T}$ and is nonnegative, so the positivity of $\mathcal{L}_q$ yields
\begin{equation}
0 \leq \mathcal{L}_q\{T_\varepsilon\}
= -\fint_\Omega q(z) \partial_z T_\varepsilon \,\dVol 
= \frac1\varepsilon \int_{z_2-\varepsilon}^{z_2+\varepsilon} q(z) {\rm d}z -
  \frac1\varepsilon \int_{z_1-\varepsilon}^{z_1+\varepsilon} q(z) {\rm d}z.
\end{equation}
Letting $\varepsilon\to 0$ using Lebesgue's differentiation theorem and rearranging yields $q(z_1)\leq q(z_2)$, which implies that $q$ is nondecreasing since $z_1$ and $z_2$ are arbitrary.

To prove the reverse statement, suppose that $q$ is nondecreasing but that $\mathcal{L}_q$ is not positive. This means that there exist a constant $c>0$  and a temperature field $T_0 \in \mathcal{T}$, nonnegative on the domain $\Omega$, such that $\mathcal{L}_q(T_0) \leq -c$. By a standard approximation argument, we may also assume that $q(z)$ is smooth on $[0,1]$. Then, integration by parts using the boundary conditions on $T_0$ yields
\begin{equation*}
    \fint_\Omega q'(z) T \, \dVol = \mathcal{L}_q(T_0) \leq -c.
\end{equation*}
This is a contradiction because the left-hand side is a nonnegative quantity, as $q'(z) \geq 0$ ($q$ is nondecreasing) and $T_0(\boldsymbol{x})\geq 0$ on $\Omega$ by assumption.
\end{proof}
%%%%%%%%%%%%

Lemma~\ref{lemma:positive-functionals} guarantees that the bounded linear functional in~\eqref{e:positive-linear-functional} is positive, as required, if $q$ is nondecreasing. The next result shows that considering more general types of positive linear functionals is not helpful.

%%%%%%%%%%%%
\begin{proposition}
Suppose there exists a positive bounded linear functional $\mathcal{L}:\mathcal{T}\to \mathbb{R}$ such that $\mathcal{S}\{\boldsymbol{u},T\}\geq \mathcal{L}\{T\}$ for all pairs $(\boldsymbol{u},T) \in \mathcal{H}$. Then, there exists a nondecreasing square-integrable function $q:(0,1)\to \mathbb{R}$ such that
\begin{equation}\label{e:reduced-lagrange-mult-condition}
    \mathcal{S}\{\boldsymbol{u},T\} \geq - \fint_\Omega q(z) \partial_z T \, \dVol \qquad \forall (\boldsymbol{u},T) \in \mathcal{H}.
\end{equation}
\end{proposition}
%%%%%%%%%%%%
\begin{proof}
Since $\mathcal{T}$ is a Hilbert space and $\mathcal{L}$ is bounded, the Riesz representation theorem guarantees that there exists a fixed temperature field $\theta\in \mathcal{T}$ such that
\begin{equation}
    \mathcal{L}\{T\} = \fint_\Omega \nabla \theta \cdot \nabla T \,\dVol.
\end{equation}
Next, fix any pair $(\boldsymbol{u},T) \in \mathcal{H}$ and observe that, by virtue of the horizontal periodicity, the pair of translated fields $(\boldsymbol{u}_{r,s},T_{r,s}) := (\boldsymbol{u}(x+r,y+s,z),T(x+r,y+s,z))$ also belongs to $\mathcal{H}$. By assumption $\mathcal{S}\{\boldsymbol{u},T\}\geq \mathcal{L}\{T\}$ for all $(\boldsymbol{u},T) \in \mathcal{H}$, so
\begin{equation}\label{e:shifted-inequality-L}
    \mathcal{S}\{\boldsymbol{u}_{r,s},T_{r,s}\}\geq  \fint_\Omega \nabla \theta \cdot \nabla T_{r,s} \,\dVol = \fint_\Omega \nabla \theta_{-r,-s} \cdot \nabla T \,\dVol,
\end{equation}
where the second equality follows from a change of variables. The same change of variables shows that $\mathcal{S}\{\boldsymbol{u}_{r,s},T_{r,s}\} = \mathcal{S}\{\boldsymbol{u},T\}$, so averaging~\eqref{e:shifted-inequality-L} over all horizontal shifts $r$ and $s$ yields
\begin{equation}\label{e:averaged-inequality-L}
    \mathcal{S}\{\boldsymbol{u},T\}\geq  \fint_\Omega \nabla \overline{\theta} \cdot \nabla T \,\dVol = \fint_\Omega \overline{\theta}'(z) \partial_z T \,\dVol.
\end{equation}
The expression on right-hand side of this inequality is a positive linear functional because it is the average of the positive functionals $T \to \mathcal{L}\{T_{r,s}\}$. Since the pair $(\boldsymbol{u},T) \in \mathcal{H}$ is arbitrary, inequality~\eqref{e:reduced-lagrange-mult-condition} follows upon setting $q(z):=-\overline{\theta}'(z)$ and applying Lemma~\ref{lemma:positive-functionals} to conclude that $q$ is nondecreasing and square-integrable.
\end{proof}

%%%%%%%%%%%%%%%%%%%%%%%%%%%%%%%%%%
% \section{Comparing $\mathcal{S}_{\boldsymbol{k}}$ and $\Tilde{\mathcal{S}}$}

\section{Computational approach}\label{sec:computational-details}
The optimisation problems~\eqref{e:optimization-Fourier} and~\eqref{e:optimization-Fourier-positive} can be discretised into SDPs following a general strategy, and then solved using efficient algorithms for convex optimisation. This ``discretise-then-optimise'' approach preserves the linearity of the original infinite-dimensional problems and enables one to readily impose additional constraints, such as the inequalities on $\psi(0)$, $\psi(1)$ and the monotonicity constraint on $q$, that are not easy to enforce following ``optimise-then-discretise'' strategies based on the numerical solution of the Euler--Lagrange equations for~\eqref{e:optimization-Fourier} and~\eqref{e:optimization-Fourier-positive}.

The discretisation process starts by approximating the tunable functions $\psi$, $q$ and the unknown fields $\hat{T}_0$, $\hat{T}_{\boldsymbol{k}}$ and $\hat{w}_{\boldsymbol{k}}$ using a finite set of basis functions $\{\Phi_1(z),\ldots,\Phi_n(z)\}$, e.g.,
\begin{equation}\label{e:psi-expansion}
    \psi(z) = \sum_{i=1}^n A_i \Phi_i(z).
\end{equation}
Here we use a single set of basis functions for simplicity, but different fields could be approximated using different bases to improve accuracy or allow for varying degrees of smoothness. Note that while the functions $\psi$ and $q$ are arbitrary, so we are free to choose such a finite-dimensional representation without much loss of generality, assuming the same for the test functions $\hat{T}_0$, $\hat{T}_{\boldsymbol{k}}$ and $\hat{w}_{\boldsymbol{k}}$ represents a relaxation of the constraints in~\eqref{e:optimization-Fourier} and~\eqref{e:optimization-Fourier-positive}. Strictly speaking, therefore, our numerical results are not rigorous upper bounds on the vertical heat flux, but we expect convergence as $n \to \infty$.

Substituting expansions such as~\eqref{e:psi-expansion} into the inequalities on $\mathcal{S}_0$ and $\Tilde{\mathcal{S}}$ in~\eqref{e:optimization-Fourier} and~\eqref{e:optimization-Fourier-positive} reduces them to quadratic polynomial inequalities, where the independent variables are the (unknown) expansion coefficients of $\hat{T}_0$, $\hat{T}_{\boldsymbol{k}}$ and $\hat{w}_{\boldsymbol{k}}$, and the coefficients depend linearly on the optimisation variables---the scalars $U$, $a$, $b$ and the expansion coefficients of $\psi$ and $q$. These quadratic polynomial inequalities are equivalent to positive semidefinitess constraints on matrices that depend linearly on the polynomial coefficients, and hence on the optimisation variables. Moreover, the inequalities $\psi(0)\leq 1$, $\psi(1)\leq 0$ and the monotonicity constraint on $q(z)$ in~\eqref{e:optimization-Fourier-positive} can be projected onto the expansion basis to obtain a set of linear constraints on the expansions coefficients of $\psi$ and $q$. The discrete problems are therefore SDPs \citep{Boyd2004} and can be solved with a variety of algorithms \citep[see, e.g.,][]{Nemirovski2006}.

% Motivated by initial numerical experiments, we used two different types of basis functions in our computations. To discretise~\eqref{e:optimization-Fourier}, we used global polynomial expansions in the Legendre basis as implemented in the MATLAB toolbox \quinopt\ \citep{Fantuzzi2017tac,QUINOPT}, which result in SDPs with good numerical conditioning that can be solved to high accuracy with the solver \mosek\ \citep{mosek}. 
To tackle~\eqref{e:optimization-Fourier} and~\eqref{e:optimization-Fourier-positive}, we used a finite-element approximation similar to that considered by \cite{fantuzzi2018bounds}. The reason for this choice is twofold. First, a piecewise-linear finite-element representation for $q$ enables us to impose exactly the monotonicity constraint, which is key to enforcing the minimum principle on the temperature as discussed in \S\ref{sec:bounds-for-positive-temperature}. Second, the optimal $\psi$ and $q$ have steep boundary layers near the bottom boundary that cannot be approximated accurately at a reasonable computational cost using global polynomial expansions (e.g. Legendre series). Finite-element bases, instead, lead to SDPs with so-called \textit{chordal sparsity} \citep{fukuda2001exploiting} that can be solved extremely efficiently. For our particular problem, however, we also found that finite-element bases lead to SDPs with worse numerical conditioning than those obtained with other bases, such as Legendre polynomials. Accurate solution, therefore, required the multiple-precision solver \sdpagmp\ \citep{Fujisawa2008,Waki2012}. Despite this issue, which we do not expect to be generic, the enhanced sparsity of the finite-element approach resulted in significant efficiency gains compared to accurate Legendre series expansions.
\section{Estimates on the spectral constraint}
\label{appendix:spectral estimates}

A $\Ra$ dependent variation of the parameters follows from enforcing the spectral constraint. In the bulk of the domain the constraint can easily be satisfied for any profile $\psi(z)$ such that $\psi'(z)=a$ for $z \in (\delta, 1-\varepsilon)$. With this assumption, the only sign-indefinite term in~\eqref{Sk} is
\begin{equation}
    \int_0^1 (a-\psi')\hat{w}_{\boldsymbol{k}}\hat{T}_{\boldsymbol{k}}\,\textrm{d}z =
    \int_{[0,\delta]\cup [1-\varepsilon]} (a-\psi')\hat{w}_{\boldsymbol{k}}\hat{T}_{\boldsymbol{k}} \,\textrm{d}z.
\end{equation}
Where we take the real part of the product $\hat{w}_{\boldsymbol{k}}\hat{T}_{\boldsymbol{k}}^{*}$ .To estimate the integral over $[0,\delta]$, we apply the fundamental theorem of calculus, the boundary condition $\hat{w}_{\boldsymbol{k}}(0) = 0$ and the Cauchy--Schwarz inequality to bound
\begin{eqnarray}
     w(z) &=& \int^{z}_{0} w'(\xi) \, \textrm{d}\xi \leq \sqrt{z} \lVert w'_{\boldsymbol{k}} \rVert_{2},     \label{estimate_wk_1} \\
     w(z) &\leq& \tfrac{1}{\sqrt{2}} z^{\frac32} \lVert w''_{\boldsymbol{k}} \rVert_{2}.
     \label{estimate_wk_2}
\end{eqnarray}
Identical estimates show that $T(z) \leq \sqrt{z} \lVert T'_{\boldsymbol{k}} \rVert_{2}$.
% and by virtue of the boundary condition $\hat{T}_{\boldsymbol{k}}(0)=0$,
% \begin{equation}
%      T(z) \leq \sqrt{z} \lVert T' \rVert_{2}\, .
%      \label{estimate_Tk}
% \end{equation}
% % can proceed by controlling the function near the boundaries with estimates. Given that $w_{k}$ and $T_{k}$ are Lebesgue measurable 
% \begin{equation}
%      \left| \int_{0}^{1} w_{k}T_{k} \textrm{d}z \right| \leq   \int_{0}^{1} | w_{k}T_{k}| \textrm{d}z \, .
% \end{equation}
% Considering the velocity field, the boundary conditions imply 
% Moreover, combining the first equality in~\eqref{estimate_wk} with \eqref{estimate_wk_prime} we obtain 
%  \begin{equation}\label{estimate_wk_2}
%       w(z) \leq  \tfrac{2}{3} z^{\frac{3}{2}} \lVert w'' \rVert_{2}.
% \end{equation}
Using these inequalities, we can estimate in two ways such that 
% The same estimates hold at $z=1$ with $z$ replaced by $1-z$. Thus in some region of size $\delta$ near $z=0$, we use all the estimates above so as to write \eqref{Sk} in terms of two quadratic forms. Specifically let use approximate the indefinite integral as  
\begin{align}
    \left| \int^{\delta}_{0} (a-\psi') \hat{w}_{\boldsymbol{k}} \hat{T}_{\boldsymbol{k}}\, \textrm{d}z \right| 
    &\leq \lVert a-\psi' \rVert_{L^{\infty}(0,\delta)} \int^{\delta}_{0}\abs{\hat{w}_{\boldsymbol{k}}} \big\vert\hat{T}_{\boldsymbol{k}}\big\vert\, \textrm{d}z \nonumber \\
    &\leq \delta^{2}  \lVert a-\psi' \rVert_{L^{\infty}(0,\delta)} \left( \frac{\alpha}{2}\lVert \hat{w}'_{\boldsymbol{k}} \rVert_{2} \lVert \hat{T}'_{\boldsymbol{k}} \rVert_{2} + \frac{1-\alpha}{2\sqrt{2}} \lVert \hat{w}''_{\boldsymbol{k}} \rVert_{2} \lVert \hat{T}_{\boldsymbol{k}} \rVert_{2}  \right)\,  .
\label{e:estimate-bottom-bl}
\end{align}
where $\alpha \in [0,1]$ is a weighting parameter to be specified later
% , used to interpolate between inequalities~\eqref{estimate_wk} and~\eqref{estimate_wk_2}. 
Similar analysis near $z=1$ yields
\begin{align}\label{e:estimate-top-bl}
    \left| \int_{1-\varepsilon}^{1} \!(a-\psi') \hat{w}_{\boldsymbol{k}} \hat{T}_{\boldsymbol{k}} \textrm{d}z \right| 
    \leq \varepsilon^{2} \lVert a-\psi' \rVert_{L^\infty(1-\varepsilon,1)} \left(\frac{\beta}{2}\lVert \hat{w}'_{\boldsymbol{k}} \rVert_{2} \lVert \hat{T}'_{\boldsymbol{k}} \rVert_{2} + \frac{1-\beta}{2\sqrt{2}}\lVert \hat{w}''_{\boldsymbol{k}} \rVert_{2} \lVert \hat{T}_{\boldsymbol{k}} \rVert_{2} \right) ,
\end{align}
again $\beta \in [0,1]$ is a weighting parameter. Substituting inequalities~\eqref{e:estimate-bottom-bl} and~\eqref{e:estimate-top-bl} into \eqref{Sk} shows that
% \begin{eqnarray}
%     \mathcal{S}_{k} \geq \frac{a}{R k^{2}} \lVert w_{k}'' \rVert_{2}^{2} -  M_{1}  \lVert w_{k}'' \rVert_{2} \lVert T_{k} \rVert_{2} + b k^{2} \lVert T \rVert_{2}^{2}  \nonumber \\
%   + \frac{2 a}{R} \lVert w_{k}' \rVert_{2}^{2}- M_{2} \lVert w_{k}' \rVert_{2} \lVert T_{k}' \rVert_{2} + b \lVert T_{k} \rVert_{2}^{2} \nonumber \\ + \frac{a k^{2}}{R} \lVert w_{k} \rVert_{2}^{2}\, ,
% \end{eqnarray}
% \begin{align}
%     \mathcal{S}_{k} \geq 
%     &\frac{a}{R k^{2}} \lVert \hat{w}_{\boldsymbol{k}}'' \rVert_{2}^{2} 
%     -\frac{1-\beta}{3}\lambda(a,\delta,\varepsilon,\psi')  \lVert \hat{w}_{\boldsymbol{k}}'' \rVert_{2} \lVert \hat{T}_{\boldsymbol{k}} \rVert_{2} 
%     + b k^{2} \lVert \hat{T}_{\boldsymbol{k}} \rVert_{2}^{2} 
%     \nonumber \\
%   &+ \frac{2 a}{R} \lVert \hat{w}_{\boldsymbol{k}}' \rVert_{2}^{2}- \frac{\beta}{2}\lambda(a,\delta,\varepsilon,\psi') \lVert \hat{w}_{\boldsymbol{k}}' \rVert_{2} \lVert \hat{T}_{\boldsymbol{k}}' \rVert_{2} + b \lVert \hat{T}_{\boldsymbol{k}}' \rVert_{2}^{2},
% \end{align}
\begin{align}
    \mathcal{S}_{\boldsymbol{k}} \geq& 
    \frac{2 a}{R} \lVert \hat{w}'_{\boldsymbol{k}} \rVert_{2}^{2}- \frac{1}{2}\lambda_1(a,\delta,\varepsilon,\alpha,\beta,\psi') \lVert \hat{w}'_{\boldsymbol{k}} \rVert_{2} \lVert \hat{T}'_{\boldsymbol{k}} \rVert_{2} + b \lVert \hat{T}'_{\boldsymbol{k}} \rVert_{2}^{2} \nonumber \\
    &+ \frac{a}{Rk^2} \lVert \hat{w}''_{\boldsymbol{k}} \rVert_{2}^{2}- \frac{1}{2\sqrt{2}}\lambda_2(a,\delta,\varepsilon,\alpha,\beta,\psi') \lVert \hat{w}''_{\boldsymbol{k}} \rVert_{2} \lVert \hat{T}_{\boldsymbol{k}} \rVert_{2} + bk^2 \lVert \hat{T}_{\boldsymbol{k}} \rVert_{2}^{2}
\end{align}
where
\begin{eqnarray}
    \lambda_1(a,\delta,\varepsilon,\alpha,\beta,\psi') &:=& \alpha\delta^2\lVert a-\psi' \rVert_{L^{\infty}(0,\delta)} + \beta\varepsilon^2\lVert a-\psi' \rVert_{L^\infty(1-\varepsilon,1)}\, , \\
    \lambda_2(a,\delta,\varepsilon,\alpha,\beta,\psi') &:=& (1-\alpha)\delta^2\lVert a-\psi' \rVert_{L^{\infty}(0,\delta)} + (1-\beta) \varepsilon^2\lVert a-\psi' \rVert_{L^\infty(1-\varepsilon,1)}.
\end{eqnarray}
The right-hand side of the last inequality is two homogeneous quadratic forms in the variables $\|\hat{w}'_{\boldsymbol{k}}\|_2$, $\|\hat{T}'_{\boldsymbol{k}}\|_2$ and $\|\hat{w}''_{\boldsymbol{k}}\|_2$, $\|\hat{T}_{\boldsymbol{k}}\|_2$ , so it is nonnegative if the discriminant of both forms is nonpositive, i.e.,
\begin{eqnarray}
    \abs{\lambda_1(a,\delta,\varepsilon,\alpha,\beta,\psi')}^2 \leq \frac{32ab}{\Ra}, \\
    \abs{\lambda_2(a,\delta,\varepsilon,\alpha,\beta,\psi')}^2 \leq \frac{32ab}{\Ra},
\end{eqnarray}
from which it it is obvious that $\alpha = \beta = 1/2$ is optimal. Then it suffices to take square roots and rearrange, such that we conclude that $\mathcal{S}_{\boldsymbol{k}} \geq 0$ if
% \begin{equation}\label{e:spectral-constraint-sufficient-condition-general}
%     \delta^2\lVert a-\psi' \rVert_{L^{\infty}(0,\delta)} + \varepsilon^2\lVert a-\psi' \rVert_{L^\infty(1-\varepsilon,1)} \leq 
%     \frac{2 \sqrt{2ab}}{ (3\sqrt{2}-4) \sqrt{\Ra}}.
% \end{equation}
%
% Fixing $\beta=6\sqrt{2}-8$ to minimise the left-hand side, taking square roots and rearranging we obtain
%
% Taking instead the simplified $\mathcal{S}_{\boldsymbol{k}}$, we need not consider splitting the the integral and it suffices to approximate in terms of $w'_{\boldsymbol{k}}$ and $T'_{\boldsymbol{k}}$ alone. Thus
%
\begin{equation}\label{e:spectral-constraint-sufficient-condition-general_simp}
    \delta^2\lVert a-\psi' \rVert_{L^{\infty}(0,\delta)} + \varepsilon^2\lVert a-\psi' \rVert_{L^\infty(1-\varepsilon,1)} \leq 
    \frac{8 \sqrt{2ab}}{ \sqrt{\Ra}}.
\end{equation}
This condition reduces to~\eqref{e:delta-inequality} when $\delta=\varepsilon$ and $\psi$ is as in~\eqref{phicomplex}.
% \begin{equation}
%     M_{1}^{2} \leq \frac{4ab}{R}\,, \quad  M_{2}^{2} \leq \frac{8ab}{R}\, .
% \end{equation}
% where 
% \begin{eqnarray}
%     M_{1} &=& \tfrac{(1-\alpha) \delta^{2}}{3} \lVert a-\psi' \rVert_{L^{\infty}(0,\delta)}  + \tfrac{(1-\beta) \varepsilon^{2}}{3} \lVert a-\psi' \rVert_{L^{\infty}[1-\varepsilon,1]}\, ,\\
%     M_{2} &=& \tfrac{\alpha \delta^{2}}{2} \lVert a-\psi' \rVert_{L^{\infty}(0,\delta)}  + \tfrac{\beta \varepsilon^{2}}{2} \lVert a-\psi' \rVert_{L^{\infty}[1-\varepsilon,1]}\,.
% \end{eqnarray}
% A final note on this is that if $\delta = \varepsilon$ and $\alpha = \beta$, then we equate the two inequalities, which yields that 
% \begin{equation}
%     \beta = - 8 + 6\sqrt{2}\, .
% \end{equation}
%%%%%%%%%%%%%%%%%%%%%%%%%%%%%%%%%%%%%%%
%%%%%%%%
\section{Proof of Proposition \ref{prop:counter} }
\label{appendix:proof_of_eqs}
%
%\begin{proof}
It is assumed that $(a,b,\psi,q) \in \mathcal{A}\{\delta,\varepsilon\} \cap \mathcal{B}\{\gamma\} \cap \mathcal{C}\{\delta,\varepsilon,R\}$. The first step is to bound $\int_0^1 \psi(z) dz$ from above. To do this, we work back from $z=1$, using the assumptions to estimate $\psi$. 

Let $1-\varepsilon \leq z \leq 1$. Since $\psi \in C^1[0,1]$ and $\psi(1)=0$, the mean value theorem implies that there exist $z_\varepsilon \in (1-\varepsilon,1)$ such that 
\[
-\psi'(z_c) = \frac{\psi(z)}{1-z} \geq \frac{\psi(z)}{\varepsilon}.
\]
Since $(a,b,\psi(z))$ satisfy \eqref{e:spectral-constraint-sufficient-condition-general_simp}, it follows that 
\begin{eqnarray}
    \varepsilon^{4} \lVert a - \psi' \rVert^{2}_{L^{\infty}(1-\varepsilon,1)} \leq 128 \frac{ a\,b}{\Ra} &\implies& \varepsilon^{4} \left( a + \frac{\psi(z)}{\varepsilon}\right)^{2} \leq 128\frac{a\,b}{\Ra} \nonumber \\
    &\implies& \varepsilon^{4}\left( \frac{a}{b}\right)^{2} + \left( \frac{a}{b}\right)\left(\frac{2\varepsilon^{3}\psi(z)}{b} - \frac{128}{\Ra}\right) + \left( \frac{\varepsilon \psi(z)}{b}\right)^{2} \leq 0 \nonumber \\
    &\implies& \left( \frac{2\varepsilon^{3} \psi(z)}{b} - \frac{128}{\Ra}\right)^{2} \geq 4 \varepsilon^4 \left(\frac{\varepsilon\psi(z)}{b} \right)^{2} \nonumber \\ &\implies&
    \frac{\psi(z) \varepsilon^{3}}{b} \leq \frac{32}{\Ra} \label{e:psi_e_R}
\end{eqnarray}
Next, since $\psi(1)=0$, the assumption that $\|\psi'\|_{L^\infty(1-\varepsilon,1)} \leq 2b$ gives $|\psi(z)| \leq \varepsilon \|\psi'\|_{L^\infty(1-\varepsilon,1)} \leq 2 b \varepsilon$, which in turn implies
\begin{equation} \label{e:lb_psi_b_e}
 \varepsilon \geq \frac{1}{2b} \psi(z).
\end{equation}
Since $1-\varepsilon < z < 1$ was arbitrary in the above argument, it follows from \eqref{e:psi_e_R} and \eqref{e:lb_psi_b_e} and the fact that $\psi'(z) \geq 0$ for all $z \in [\delta,1-\varepsilon]$ that 
\begin{equation} \label{e:up_d_1}
\|\psi\|_{L^\infty(\delta,1)} \leq \|\psi\|_{L^\infty(1-\varepsilon,1)} \leq 4b R^{-\frac14}.
\end{equation}

We now estimate $\psi$ in the lower boundary layer $[0,\delta]$. The mean value theorem implies that there exists $z_\delta$ such that 
\[
\psi'(z_\delta) = \frac{\psi(\delta)-\psi(0)}{\delta}.
\]
Since $(a,b,\psi(z))$ satisfy \eqref{e:spectral-constraint-sufficient-condition-general_simp}, using the above equation and the assumption that $0 < a \leq b$, it then follows that 
\begin{eqnarray}
    \delta^{4}\lVert a- \psi' \rVert^{2}_{L^{\infty}(0,\delta)} \leq 128 \cdot \frac{a\, b}{\Ra} & \implies& \delta^{4} \left| a + \left(\frac{\psi(0)-\psi(\delta)}{\delta}\right)\right|^{2} \leq 128 \cdot \frac{a\,b}{\Ra}\, \nonumber \\
   & \implies &\delta^{2} |a\delta + \psi(0) - \psi(\delta)|^{2} \leq 128 \cdot \frac{b^2}{\Ra} \nonumber \\
   &\implies &\frac{\delta}{b} \psi(0) \leq \frac{8\sqrt{2}}{R^{\frac12}} + \frac{\delta \psi(\delta)}{b} \nonumber \\
   \text{(by \eqref{e:up_d_1})} & \implies& \frac{\delta}{b} \psi(0) \leq \frac{8\sqrt{2}}{R^{\frac12}} + \frac{4\delta}{R^{\frac14}} \label{eq:psi0_lb}
\end{eqnarray}
Using \eqref{e:up_d_1}, \eqref{eq:psi0_lb}, the assumption that $\psi'(z) \leq 0$ on $[0,\delta]$ and $R \geq 1$ gives
\begin{equation} \label{eq:psi_l1_bnd}
\frac{1}{b} \int_0^1 \psi(z) dz \leq \frac{\delta}{b} \psi(0) + \frac{(1-\delta)}{b} \|\psi\|_{L^\infty(\delta,1)} \leq \frac{8\sqrt{2}+4}{R^{\frac14}} 
\end{equation}

Next we consider the $L^2$ component of the cost function. Using the assumption that $\psi',q$ are constant in an interval $(1/2-\gamma,1/2+\gamma)$, 
\begin{eqnarray}
    \frac{1}{4b} \big\lVert  \psi' - q -b(z-1/2) \big\rVert^{2}_{2}\, &\geq& \, \frac{1}{4b} \int^{\tfrac{1}{2}+ \gamma}_{\tfrac{1}{2}-\gamma} \left[\psi' - q - b(z-1/2) \right]^{2}\, \textrm{d}z\, \nonumber \\
    &\geq &  \frac{b}{4} \int_{\frac12-\gamma}^{\frac12 + \gamma} (z-1/2)^2 \, \textrm{d}z  \nonumber \\
    &= &  \frac{b \gamma^{3}}{6}. 
\end{eqnarray}
Combining the above estimate with
\eqref{eq:psi_l1_bnd} gives the stated result
\begin{equation*}
    \frac{1}{4b} \big\lVert  \psi' - q -b(z-1/2) \big\rVert^{2}_{2} - \int_0^1 \psi(z) dz \geq \frac{b}{6}\left( \gamma^3 - 24 \cdot  \frac{1+2\sqrt{2}}{R^{\frac14}} \right). 
\end{equation*}
%\end{proof}

%% file: Figures/colorbar-hot-2.tex
% A custom-made colorbar
\begin{tikzpicture}
\draw[draw=white] (0,-0.25) rectangle ++(0.3, 0.7500) node[pos=.5, font=\fontsize{8}{8}] {\Ra};
\draw[draw=colorbar1, fill=colorbar1] (0,0.5000) rectangle ++(0.2,0.1420);
\draw[draw=colorbar2, fill=colorbar2] (0,0.6420) rectangle ++(0.2,0.1420);
\draw[draw=colorbar3, fill=colorbar3] (0,0.7840) rectangle ++(0.2,0.1420);
\draw[draw=colorbar4, fill=colorbar4] (0,0.9260) rectangle ++(0.2,0.1420);
\draw[draw=colorbar5, fill=colorbar5] (0,1.0680) rectangle ++(0.2,0.1420);
\draw[draw=colorbar6, fill=colorbar6] (0,1.2100) rectangle ++(0.2,0.1420);
\draw[draw=colorbar7, fill=colorbar7] (0,1.3520) rectangle ++(0.2,0.1420);
\draw[draw=colorbar8, fill=colorbar8] (0,1.4940) rectangle ++(0.2,0.1420);
\draw[draw=colorbar9, fill=colorbar9] (0,1.6360) rectangle ++(0.2,0.1420);
\draw[draw=colorbar10, fill=colorbar10] (0,1.7780) rectangle ++(0.2,0.1420);
\draw[draw=colorbar11, fill=colorbar11] (0,1.9200) rectangle ++(0.2,0.1420);
\draw[draw=colorbar12, fill=colorbar12] (0,2.0620) rectangle ++(0.2,0.1420);
\draw[draw=colorbar13, fill=colorbar13] (0,2.2040) rectangle ++(0.2,0.1420);
\draw[draw=colorbar14, fill=colorbar14] (0,2.3460) rectangle ++(0.2,0.1420);
\draw[draw=colorbar15, fill=colorbar15] (0,2.4880) rectangle ++(0.2,0.1420);
\draw[draw=colorbar16, fill=colorbar16] (0,2.6300) rectangle ++(0.2,0.1420);
\draw[draw=colorbar17, fill=colorbar17] (0,2.7720) rectangle ++(0.2,0.1420);
\draw[draw=colorbar18, fill=colorbar18] (0,2.9140) rectangle ++(0.2,0.1420);
\draw[draw=colorbar19, fill=colorbar19] (0,3.0560) rectangle ++(0.2,0.1420);
\draw[draw=colorbar20, fill=colorbar20] (0,3.1980) rectangle ++(0.2,0.1420);
\draw[draw=colorbar21, fill=colorbar21] (0,3.3400) rectangle ++(0.2,0.1420);
\draw[draw=black] (0.15,0.5) -- (0.3,0.5) node[anchor=west, font=\fontsize{6}{6}] {$2.6\times10^5$};
\draw[draw=black] (0.15,1.2355) -- (0.3,1.2355) node[anchor=west, font=\fontsize{6}{6}] {$2.7\times10^5$};
\draw[draw=black] (0.15,1.9710) -- (0.3,1.9710) node[anchor=west, font=\fontsize{6}{6}] {$2.8\times10^5$};
\draw[draw=black] (0.15,2.7065) -- (0.3,2.7065) node[anchor=west, font=\fontsize{6}{6}] {$2.9\times10^5$};
\draw[draw=black] (0.15,3.442) -- (0.3,3.442) node[anchor=west, font=\fontsize{6}{6}] {$3.0\times10^5$};
\end{tikzpicture}

%% file: main.bbl
\begin{thebibliography}{67}
\expandafter\ifx\csname natexlab\endcsname\relax\def\natexlab#1{#1}\fi
\def\au#1{#1} \def\ed#1{#1} \def\yr#1{#1}\def\at#1{#1}\def\jt#1{\textit{#1}}
  \def\bt#1{#1}\def\bvol#1{\textbf{#1}} \def\vol#1{#1} \def\pg#1{#1}
  \def\publ#1{#1}\def\arxiv#1{#1}\def\org#1{#1}\def\st#1{\textit{#1}}

\bibitem[Bercovici(2011)]{bercovici2011mantle}
{\sc \au{Bercovici, D.}} \yr{2011}  \at{Mantle convection}.  \jt{Encyclopedia
  of Solid Earth Geophysics. Gupta, HK,(ed.) Springer} .

\bibitem[Bouillaut {\em et~al.\/}(2019)Bouillaut, Lepot, Aumaître \&
  Gallet]{BouVjfm2019a}
{\sc \au{Bouillaut, Vincent}, \au{Lepot, Simon}, \au{Aumaître, Sébastien} \&
  \au{Gallet, Basile}} \yr{2019}  \at{Transition to the ultimate regime in a
  radiatively driven convection experiment}.  \jt{Journal of Fluid Mechanics}
  \bvol{861},  \pg{R5}.

\bibitem[Boyd \& Vandenberghe(2004)]{Boyd2004}
{\sc \au{Boyd, S.} \& \au{Vandenberghe, L.}} \yr{2004} {\em {Convex
  Optimization}\/}.  \publ{Cambridge University Press}.

\bibitem[Chernyshenko(2017)]{chernyshenko2017relationship}
{\sc \au{Chernyshenko, S.~I.}} \yr{2017}  \at{Relationship between the methods
  of bounding time averages}.  \jt{arXiv preprint arXiv:1704.02475} .

\bibitem[Chernyshenko {\em et~al.\/}(2014)Chernyshenko, Goulart, Huang \&
  Papachristodoulou]{chernyshenko2014polynomial}
{\sc \au{Chernyshenko, S.~I.}, \au{Goulart, P.~J.}, \au{Huang, D.} \&
  \au{Papachristodoulou, A.}} \yr{2014}  \at{{Polynomial sum of squares in
  fluid dynamics: a review with a look ahead}}.  \jt{Philosophical Transactions
  of the Royal Society A: Mathematical, Physical and Engineering Sciences}
  \bvol{372}~(2020),  \pg{20130350}.

\bibitem[Choffrut {\em et~al.\/}(2016)Choffrut, Nobili \&
  Otto]{choffrut2016upper}
{\sc \au{Choffrut, A.}, \au{Nobili, C.} \& \au{Otto, F.}} \yr{2016}  \at{{Upper
  bounds on Nusselt number at finite Prandtl number}}.  \jt{Journal of
  Differential Equations}  \bvol{260}~(4),  \pg{3860--3880}.

\bibitem[Constantin \& Doering(1995)]{constantin1995variational}
{\sc \au{Constantin, P.} \& \au{Doering, C.~R.}} \yr{1995}  \at{{Variational
  bounds on energy dissipation in incompressible flows. II. Channel flow}}.
  \jt{Physical Review E}  \bvol{51}~(4),  \pg{3192}.

\bibitem[Constantin \& Doering(1996)]{Constantin1996}
{\sc \au{Constantin, P.} \& \au{Doering, C.~R.}} \yr{1996}  \at{{Heat transfer
  in convective turbulence}}.  \jt{Nonlinearity}  \bvol{9}~(4),
  \pg{1049--1060}.

\bibitem[Constantin \& Doering(1999)]{constantin1999infinitePr}
{\sc \au{Constantin, P.} \& \au{Doering, C.~R.}} \yr{1999}  \at{{Infinite
  Prandtl number convection}}.  \jt{Journal of Statistical Physics}
  \bvol{94}~(1-2),  \pg{159--172}.

\bibitem[Debler(1959)]{Debler1959}
{\sc \au{Debler, W.~R.}} \yr{1959}  \at{{The onset of laminar natural
  convection in a fluid with homogenously distributed heat sources}}. PhD
  thesis, University of Michigan.

\bibitem[Doering \& Constantin(1994)]{doering1994variational}
{\sc \au{Doering, C.~R.} \& \au{Constantin, P.}} \yr{1994}  \at{{Variational
  bounds on energy dissipation in incompressible flows: shear flow}}.
  \jt{Physical Review E}  \bvol{49}~(5),  \pg{4087}.

\bibitem[Doering \& Constantin(1996)]{doering1996variational}
{\sc \au{Doering, C.~R.} \& \au{Constantin, P.}} \yr{1996}  \at{{Variational
  bounds on energy dissipation in incompressible flows. III. Convection}}.
  \jt{Physical Review E}  \bvol{53}~(6),  \pg{5957}.

\bibitem[Doering \& Constantin(1998)]{Doering1998}
{\sc \au{Doering, C.~R.} \& \au{Constantin, P.}} \yr{1998}  \at{{Bounds for
  heat transport in a porous layer}}.  \jt{Journal of Fluid Mechanics}
  \bvol{376},  \pg{263--296}.

\bibitem[Doering \& Constantin(2001)]{doering2001upper}
{\sc \au{Doering, C.~R.} \& \au{Constantin, P.}} \yr{2001}  \at{{On upper
  bounds for infinite Prandtl number convection with or without rotation}}.
  \jt{Journal of Mathematical Physics}  \bvol{42}~(2),  \pg{784--795}.

\bibitem[Doering {\em et~al.\/}(2006)Doering, Otto \&
  Reznikoff]{doering2006bounds}
{\sc \au{Doering, C.~R.}, \au{Otto, F.} \& \au{Reznikoff, M.~G.}} \yr{2006}
  \at{{Bounds on vertical heat transport for infinite-Prandtl-number
  Rayleigh--B{\'e}nard convection}}.  \jt{Journal of fluid mechanics}
  \bvol{560},  \pg{229--241}.

\bibitem[Emara \& Kulacki(1980)]{emara1980}
{\sc \au{Emara, A.~A.} \& \au{Kulacki, F.~A.}} \yr{1980}  \at{{A numerical
  investigation of thermal convection in a heat-generating fluid layer}}.
  \jt{Journal of Heat Transfer}  \bvol{102}~(3),  \pg{531--537}.

\bibitem[Fantuzzi {\em et~al.\/}(2016)Fantuzzi, Goluskin, Huang \&
  Chernyshenko]{Fantuzzi2016siads}
{\sc \au{Fantuzzi, G.}, \au{Goluskin, D.}, \au{Huang, D.} \& \au{Chernyshenko,
  S.~I.}} \yr{2016}  \at{{Bounds for deterministic and stochastic dynamical
  systems using sum-of-squares optimization}}.  \jt{SIAM Journal on Applied
  Dynamical Systems}  \bvol{15}~(4),  \pg{1962--1988}.

\bibitem[Fantuzzi {\em et~al.\/}(2018)Fantuzzi, Pershin \&
  Wynn]{fantuzzi2018bounds}
{\sc \au{Fantuzzi, G.}, \au{Pershin, A.} \& \au{Wynn, A.}} \yr{2018}
  \at{{Bounds on heat transfer for B{\'e}nard--Marangoni convection at infinite
  Prandtl number}}.  \jt{Journal of Fluid Mechanics}  \bvol{837},
  \pg{562--596}.

\bibitem[Fantuzzi \& Wynn(2015)]{fantuzzi2015construction}
{\sc \au{Fantuzzi, G.} \& \au{Wynn, A.}} \yr{2015}  \at{{Construction of an
  optimal background profile for the Kuramoto--Sivashinsky equation using
  semidefinite programming}}.  \jt{Physics Letters A}  \bvol{379}~(1-2),
  \pg{23--32}.

\bibitem[Fantuzzi \& Wynn(2016)]{Fantuzzi2016PRE}
{\sc \au{Fantuzzi, G.} \& \au{Wynn, A.}} \yr{2016}  \at{{Optimal bounds with
  semidefinite programming: An application to stress-driven shear flows}}.
  \jt{Physical Review E}  \bvol{93}~(4),  \pg{043308}.

\bibitem[Fantuzzi {\em et~al.\/}(2017)Fantuzzi, Wynn, Goulart \&
  Papachristodoulou]{Fantuzzi2017tac}
{\sc \au{Fantuzzi, G.}, \au{Wynn, A.}, \au{Goulart, P.~J.} \&
  \au{Papachristodoulou, A.}} \yr{2017}  \at{{Optimization with affine
  homogeneous quadratic integral inequality constraints}}.  \jt{IEEE
  Transactions on Automatic Control}  \bvol{62}~(12),  \pg{6221--6236}.

\bibitem[Foias {\em et~al.\/}(1987)Foias, Manley \& Temam]{foias1987attractors}
{\sc \au{Foias, C.}, \au{Manley, O.} \& \au{Temam, R.}} \yr{1987}
  \at{Attractors for the b{\'e}nard problem: existence and physical bounds on
  their fractal dimension}.  \jt{Nonlinear Analysis: Theory, Methods \&
  Applications}  \bvol{11}~(8),  \pg{939--967}.

\bibitem[Fujisawa {\em et~al.\/}(2008)Fujisawa, Fukuda, Kobayashi, Kojima,
  Nakata, Nakata \& Yamashita]{Fujisawa2008}
{\sc \au{Fujisawa, K.}, \au{Fukuda, M.}, \au{Kobayashi, K.}, \au{Kojima, M.},
  \au{Nakata, K.}, \au{Nakata, M.} \& \au{Yamashita, M.}} \yr{2008}  \bt{{SDPA
  (SemiDefinite Programming Algorithm) and SDPA-GMP User's Manual -- Version
  7.1.1}}. {\em Tech. Rep.\/}.  \org{Department of Mathematical and Computing
  Sciences, Tokyo Institute of Technology}, Tokyo, Japan.

\bibitem[Fujisawa {\em et~al.\/}(2009)Fujisawa, Kim, Kojima, Okamoto \&
  Yamashita]{Fujisawa2009}
{\sc \au{Fujisawa, Katsuki}, \au{Kim, Sunyoung}, \au{Kojima, Masakazu},
  \au{Okamoto, Y.} \& \au{Yamashita, Makoto}} \yr{2009}  \bt{{User's Manual for
  SparseCoLO: Conversion Methods for SPARSE COnic-form Linear Optimization
  Problems}}. {\em Tech. Rep.\/}.  \org{Department of Mathematical and
  Computing Sciences, Tokyo Institute of Technology}, Tokyo, Japan.

\bibitem[Fukuda {\em et~al.\/}(2001)Fukuda, Kojima, Murota \&
  Nakata]{fukuda2001exploiting}
{\sc \au{Fukuda, M.}, \au{Kojima, M.}, \au{Murota, K.} \& \au{Nakata, K.}}
  \yr{2001}  \at{{Exploiting sparsity in semidefinite programming via matrix
  completion I: General framework }}.  \jt{SIAM Journal on Optimization}
  \bvol{11}~(3),  \pg{647--674}.

\bibitem[Goluskin(2015)]{Goluskin2015}
{\sc \au{Goluskin, D.}} \yr{2015}  \at{{Internally heated convection beneath a
  poor conductor}}.  \jt{Journal of Fluid Mechanics}  \bvol{771},  \pg{36--56}.

\bibitem[Goluskin(2016)]{goluskin2016internally}
{\sc \au{Goluskin, D.}} \yr{2016} {\em {Internally heated convection and
  Rayleigh-B{\'e}nard convection}\/}.  \publ{Springer}.

\bibitem[Goluskin \& Doering(2016)]{goluskin2016rough}
{\sc \au{Goluskin, D.} \& \au{Doering, C.~R.}} \yr{2016}  \at{{Bounds for
  convection between rough boundaries}}.  \jt{Journal of Fluid Mechanics}
  \bvol{804},  \pg{370--386}.

\bibitem[Goluskin \& Fantuzzi(2019)]{goluskin2019ks}
{\sc \au{Goluskin, D.} \& \au{Fantuzzi, G.}} \yr{2019}  \at{{Bounds on mean
  energy in the Kuramoto-Sivashinsky equation computed using semidefinite
  programming}}.  \jt{Nonlinearity}  \bvol{32}~(5),  \pg{1705--1730}.

\bibitem[Goluskin \& van~der Poel(2016)]{goluskin2016penetrative}
{\sc \au{Goluskin, D.} \& \au{van~der Poel, E.~P.}} \yr{2016}  \at{{Penetrative
  internally heated convection in two and three dimensions}}.  \jt{Journal of
  Fluid Mechanics}  \bvol{791}.

\bibitem[Goluskin \& Spiegel(2012)]{goluskin2012convection}
{\sc \au{Goluskin, D.} \& \au{Spiegel, E.~A.}} \yr{2012}  \at{{Convection
  driven by internal heating}}.  \jt{Physics Letters A}  \bvol{377}~(1-2),
  \pg{83--92}.

\bibitem[Grossmann \& Lohse(2000)]{GroSjfm2000a}
{\sc \au{Grossmann, S.} \& \au{Lohse, D.}} \yr{2000}  \at{Scaling in thermal
  convection: a unifying theory}.  \jt{Journal of Fluid Mechanics}  \bvol{407},
   \pg{27–56}.

\bibitem[Howard(1961)]{HowLjfm1961a}
{\sc \au{Howard, L.~N.}} \yr{1961}  \at{Note on a paper of {J}ohn {W}.
  {M}iles}.  \jt{Journal of Fluid Mechanics}  \bvol{10}~(4),  \pg{509–512}.

\bibitem[Jahn \& Reineke(1974)]{Jahn1974}
{\sc \au{Jahn, M.} \& \au{Reineke, H.-H.}} \yr{1974} Free convection heat
  transfer with internal heat sources, calculations and measurements.  \bt{In
  {\em Proceedings of the 5\textsuperscript{th} International Heat Transfer
  Conference\/}},  \pg{pp. 74--78}. Tokyo.

\bibitem[Kakac {\em et~al.\/}(1985)Kakac, Aung \& Viskanta]{kakac1985natural}
{\sc \au{Kakac, S.}, \au{Aung, W.~M.} \& \au{Viskanta, R.}} \yr{1985}
  \at{{Natural convection: fundamentals and applications}}.  \jt{Washington,
  DC, Hemisphere Publishing Corp., 1985, 1191 p.} .

\bibitem[Kulacki \& Goldstein(1972)]{kulacki1972thermal}
{\sc \au{Kulacki, F.~A.} \& \au{Goldstein, R.~J.}} \yr{1972}  \at{{Thermal
  convection in a horizontal fluid layer with uniform volumetric energy
  sources}}.  \jt{Journal of Fluid Mechanics}  \bvol{55}~(2),  \pg{271--287}.

\bibitem[Lee {\em et~al.\/}(2007)Lee, Lee \& Suh]{Lee2007}
{\sc \au{Lee, S.~D.}, \au{Lee, J.~K.} \& \au{Suh, K.~Y.}} \yr{2007}
  \at{{Boundary condition dependent natural convection in a rectangular pool
  with internal heat sources}}.  \jt{Journal of Heat Transfer}  \bvol{129}~(5),
   \pg{679--682}.

\bibitem[Lu {\em et~al.\/}(2004)Lu, Doering \& Busse]{lu2004bounds}
{\sc \au{Lu, L.}, \au{Doering, C.~R.} \& \au{Busse, F.~H.}} \yr{2004}
  \at{{Bounds on convection driven by internal heating}}.  \jt{Journal of
  mathematical physics}  \bvol{45}~(7),  \pg{2967--2986}.

\bibitem[Malkus(1954)]{malkus1954heat}
{\sc \au{Malkus, W. V.~R.}} \yr{1954}  \at{{The heat transport and spectrum of
  thermal turbulence}}.  \jt{Proceedings of the Royal Society of London. Series
  A. Mathematical and Physical Sciences}  \bvol{225}~(1161),  \pg{196--212}.

\bibitem[Mayinger {\em et~al.\/}(1976)Mayinger, Jahn, Reineke \&
  Steibnberner]{Mayinger1976}
{\sc \au{Mayinger, F.}, \au{Jahn, M.}, \au{Reineke, H.} \& \au{Steibnberner,
  V.}} \yr{1976}  \bt{Examination of thermohydraulic processes and heat
  transfer in a core melt}. {\em Tech. Rep.\/} BMFT RS 48/1.  \org{Institut
  f\"ur Verfahrenstechnik der TU, Hanover Germany}.

\bibitem[Miles(1961)]{MilJjfm1961a}
{\sc \au{Miles, J.~W.}} \yr{1961}  \at{On the stability of heterogeneous shear
  flows}.  \jt{Journal of Fluid Mechanics}  \bvol{10}~(4),  \pg{496–508}.

\bibitem[Nemirovski(2006)]{Nemirovski2006}
{\sc \au{Nemirovski, A.}} \yr{2006} {Advances in convex optimization: Conic
  programming}.  \bt{In {\em {International Congress of Mathematicians}\/}}, ,
  \vol{vol.~1},  \pg{pp. 413--444}.

\bibitem[Otero {\em et~al.\/}(2004)Otero, Dontcheva, Johnston, Worthing,
  Kurganov, Petrova \& Doering]{Otero2004}
{\sc \au{Otero, J.}, \au{Dontcheva, L.~A.}, \au{Johnston, H.}, \au{Worthing,
  R.~A.}, \au{Kurganov, A.}, \au{Petrova, G.} \& \au{Doering, C.~R.}} \yr{2004}
   \at{{High-Rayleigh-number convection in a fluid-saturated porous layer}}.
  \jt{Journal of Fluid Mechanics}  \bvol{500},  \pg{263--281}.

\bibitem[Otero {\em et~al.\/}(2002)Otero, Wittenberg, Worthing \&
  Doering]{Otero2002}
{\sc \au{Otero, J.}, \au{Wittenberg, R.~W.}, \au{Worthing, R.~A.} \&
  \au{Doering, C.~R.}} \yr{2002}  \at{{Bounds on Rayleigh--B\'enard convection
  with an imposed heat flux}}.  \jt{Journal of Fluid Mechanics}  \bvol{473},
  \pg{191--199}.

\bibitem[Otto \& Seis(2011)]{otto2011rayleigh}
{\sc \au{Otto, F.} \& \au{Seis, C.}} \yr{2011}  \at{{Rayleigh--B{\'e}nard
  convection: improved bounds on the Nusselt number}}.  \jt{Journal of
  mathematical physics}  \bvol{52}~(8),  \pg{083702}.

\bibitem[Peckover \& Hutchinson(1974)]{Peckover1974}
{\sc \au{Peckover, R.~S.} \& \au{Hutchinson, I.~H.}} \yr{1974}  \at{{Convective
  rolls driven by internal heat sources}}.  \jt{Physics of Fluids}
  \bvol{17}~(7),  \pg{1369--1371}.

\bibitem[Plasting \& Kerswell(2003)]{Plasting2003}
{\sc \au{Plasting, S.~C.} \& \au{Kerswell, R.~R.}} \yr{2003}  \at{{Improved
  upper bound on the energy dissipation rate in plane Couette flow: the full
  solution to Busse's problem and the Constantin-Doering-Hopf problem with
  one-dimensional background field}}.  \jt{Journal of Fluid Mechanics}
  \bvol{477},  \pg{363--379}.

\bibitem[Priestley(1954)]{priestley1954vertical}
{\sc \au{Priestley, C. H.~B.}} \yr{1954}  \at{Vertical heat transfer from
  impressed temperature fluctuations}.  \jt{Australian Journal of Physics}
  \bvol{7}~(1),  \pg{202--209}.

\bibitem[Rosa \& Temam(2020)]{rosa2020optimal}
{\sc \au{Rosa, R.} \& \au{Temam, R.~M.}} \yr{2020}  \at{Optimal minimax bounds
  for time and ensemble averages of dissipative infinite-dimensional systems
  with applications to the incompressible navier-stokes equations}.  \jt{arXiv
  preprint arXiv:2010.06730} .

\bibitem[Spiegel(1963)]{spiegel1963generalization}
{\sc \au{Spiegel, E.~A.}} \yr{1963}  \at{A generalization of the mixing-length
  theory of turbulent convection.}  \jt{The Astrophysical Journal}  \bvol{138},
   \pg{216}.

\bibitem[Straus(1976)]{Straus1976}
{\sc \au{Straus, J.~M.}} \yr{1976}  \at{{Penetrative convection in a layer of
  fluid heated from within}}.  \jt{The Astrophysical Journal}  \bvol{209},
  \pg{179--189}.

\bibitem[Tilgner(2017)]{Tilgner2017}
{\sc \au{Tilgner, A.}} \yr{2017}  \at{{Bounds on poloidal kinetic energy in
  plane layer convection}}.  \jt{Physical Review Fluids}  \bvol{2}~(12),
  \pg{123502}.

\bibitem[Tilgner(2019)]{Tilgner2019}
{\sc \au{Tilgner, A.}} \yr{2019}  \at{{Time evolution equation for advective
  heat transport as a constraint for optimal bounds in Rayleigh-B{\'{e}}nard
  convection}}.  \jt{Physical Review Fluids}  \bvol{4}~(1),  \pg{1--11}.

\bibitem[Tobasco {\em et~al.\/}(2018)Tobasco, Goluskin \&
  Doering]{tobasco2018optimal}
{\sc \au{Tobasco, I.}, \au{Goluskin, D.} \& \au{Doering, C.~R.}} \yr{2018}
  \at{Optimal bounds and extremal trajectories for time averages in nonlinear
  dynamical systems}.  \jt{Physics Letters A}  \bvol{382}~(6),  \pg{382--386}.

\bibitem[Tritton(1975)]{tritton1975internally}
{\sc \au{Tritton, D.~J.}} \yr{1975}  \at{Internally heated convection in the
  atmosphere of venus and in the laboratory}.  \jt{Nature}  \bvol{257}~(5522),
  \pg{110--112}.

\bibitem[Trowbridge {\em et~al.\/}(2016)Trowbridge, Melosh, Steckloff \&
  Freed]{trowbridge2016vigorous}
{\sc \au{Trowbridge, A.~J.}, \au{Melosh, H.~J.}, \au{Steckloff, J.~K.} \&
  \au{Freed, A.~M.}} \yr{2016}  \at{{Vigorous convection as the explanation for
  Pluto’s polygonal terrain}}.  \jt{Nature}  \bvol{534}~(7605),  \pg{79--81}.

\bibitem[Tveitereid(1978)]{Tveitereid1978}
{\sc \au{Tveitereid, M.}} \yr{1978}  \at{{Thermal convection in a horizontal
  fluid layer with internal heat sources}}.  \jt{International Journal of Heat
  and Mass Transfer}  \bvol{21}~(3),  \pg{335--339}.

\bibitem[Waki {\em et~al.\/}(2012)Waki, Nakata \& Muramatsu]{Waki2012}
{\sc \au{Waki, H.}, \au{Nakata, M.} \& \au{Muramatsu, M.}} \yr{2012}
  \at{{Strange behaviors of interior-point methods for solving semidefinite
  programming problems in polynomial optimization}}.  \jt{Computational
  Optimization and Applications}  \bvol{53}~(3),  \pg{823--844}.

\bibitem[Wang {\em et~al.\/}(2020)Wang, Lohse \& Shishkina]{Wang2020}
{\sc \au{Wang, Q.}, \au{Lohse, D.} \& \au{Shishkina, O.}} \yr{2020}
  \at{Scaling in internally heated convection: a unifying theory}.
  \jt{Geophysical Research Letters}  \bvol{47},  \pg{e2020GL091198}.

\bibitem[Wen {\em et~al.\/}(2013)Wen, Chini, Dianati \& Doering]{Wen2013}
{\sc \au{Wen, B.}, \au{Chini, G.~P.}, \au{Dianati, N.} \& \au{Doering, C.~R.}}
  \yr{2013}  \at{{Computational approaches to aspect-ratio-dependent upper
  bounds and heat flux in porous medium convection}}.  \jt{Physics Letters A}
  \bvol{377}~(41),  \pg{2931--2938}.

\bibitem[Wen {\em et~al.\/}(2015)Wen, Chini, Kerswell \& Doering]{Wen2015}
{\sc \au{Wen, B.}, \au{Chini, G.~P.}, \au{Kerswell, R.~R.} \& \au{Doering,
  C.~R.}} \yr{2015}  \at{{Time-stepping approach for solving upper-bound
  problems: Application to two-dimensional Rayleigh-B{\'{e}}nard convection}}.
  \jt{Physical Review E}  \bvol{92}~(4),  \pg{043012}.

\bibitem[Whitehead \& Doering(2011{\natexlab{{\em
  a\/}}})]{whitehead2011internal}
{\sc \au{Whitehead, J.~P.} \& \au{Doering, C.~R.}} \yr{2011{\natexlab{{\em
  a\/}}}}  \at{{Internal heating driven convection at infinite Prandtl
  number}}.  \jt{Journal of mathematical physics}  \bvol{52}~(9),  \pg{093101}.

\bibitem[Whitehead \& Doering(2011{\natexlab{{\em
  b\/}}})]{whitehead2011ultimate}
{\sc \au{Whitehead, J.~P.} \& \au{Doering, C.~R.}} \yr{2011{\natexlab{{\em
  b\/}}}}  \at{{Ultimate State of Two-Dimensional Rayleigh--B\'enard Convection
  between Free-Slip Fixed-Temperature Boundaries}}.  \jt{Physical Review
  Letters}  \bvol{106}~(24),  \pg{244501}.

\bibitem[Whitehead \& Doering(2012)]{whitehead2012slippery}
{\sc \au{Whitehead, J.~P.} \& \au{Doering, C.~R.}} \yr{2012}  \at{{Rigid bounds
  on heat transport by a fluid between slippery boundaries}}.  \jt{Journal of
  Fluid Mechanics}  \bvol{707},  \pg{241--259}.

\bibitem[W{\"{o}}rner {\em et~al.\/}(1997)W{\"{o}}rner, Schmidt \&
  Gr{\"{o}}tzbach]{Worner1997}
{\sc \au{W{\"{o}}rner, M.}, \au{Schmidt, M.} \& \au{Gr{\"{o}}tzbach, G.}}
  \yr{1997}  \at{{Direct numerical simulation of turbulence in an internally
  heated convective fluid layer and implications for statistical modelling}}.
  \jt{Journal of Hydraulic Research}  \bvol{35}~(6),  \pg{773--797}.

\bibitem[Yamashita {\em et~al.\/}(2012)Yamashita, Fujisawa, Fukuda, Kobayashi,
  Nakata \& Nakata]{yamashita2012latest}
{\sc \au{Yamashita, M.}, \au{Fujisawa, K.}, \au{Fukuda, M.}, \au{Kobayashi,
  K.}, \au{Nakata, K.} \& \au{Nakata, M.}} \yr{2012}  \at{Latest developments
  in the sdpa family for solving large-scale sdps}.  \bt{In {\em Handbook on
  semidefinite, conic and polynomial optimization\/}},  \pg{pp. 687--713}.
  \publ{Springer}.

\bibitem[Yan(2004)]{Yan2004}
{\sc \au{Yan, X.}} \yr{2004}  \at{{On limits to convective heat transport at
  infinite Prandtl number with or without rotation}}.  \jt{Journal of
  Mathematical Physics}  \bvol{45}~(7),  \pg{2718--2743}.

\end{thebibliography}
